\documentclass[acmsmall, screen]{acmart}
\setcopyright{rightsretained}
\acmPrice{}
\acmDOI{10.1145/3563354}
\acmYear{2022}
\copyrightyear{2022}
\acmSubmissionID{oopslab22main-p769-p}
\acmJournal{PACMPL}
\acmVolume{6}
\acmNumber{OOPSLA2}
\acmArticle{191}
\acmMonth{10}

\usepackage{booktabs}   %% For formal tables:
\usepackage{subcaption} %% For complex figures with subfigures/subcaptions
\usepackage{amsmath,amssymb,stmaryrd}
\usepackage{amsthm}
\usepackage{enumerate}
\usepackage{bussproofs}
\usepackage{mathtools}
\usepackage{makecell}
\usepackage{comment}
\usepackage{proof}
\usepackage{graphicx}
\usepackage[strings]{underscore}
\usepackage[misc]{ifsym}
\usepackage[shortlabels]{enumitem}
\usepackage{makecell}
\usepackage{dcolumn,caption}
\usepackage{xspace}
\usepackage{colortbl}
\usepackage{hhline}
\usepackage{rotating}
\usepackage{xcolor}
\usepackage{enumerate}
\usepackage{siunitx}
\usepackage[utf8]{inputenc}
\usepackage[T1]{fontenc}
\usepackage{microtype}
\usepackage{flushend}
\usepackage{url}

\newcommand{\commentout}[1]{}

\newcommand{\pldicommentout}[1]{}
\newcommand{\adcomment}[1]{}
\newcommand{\madcomment}[1]{}
\newcommand{\lpcomment}[1]{}
\newcommand{\clcomment}[1]{}

\newcommand{\Mm}{\mathcal{M}}

\newcommand{\R}{\mathcal{R}}

\newcommand{\fs}{\sigma_{0}}

\newcommand{\ite}{\mathsf{ite}}
\newcommand{\nil}{\mathit{nil}}

\newcommand{\halloc}[1]{\Fr^{x;v}}

\newcommand{\folfp}{FO+\textit{lfp} }

\newcommand{\modelsfo}{\models_{\textsf{FO}}}
\newcommand{\modelslfp}{\models_{\textsf{LFP}}}

\newcommand{\Fr}{\textit{Sp}}

\newcommand{\lft}{\mathit{left}}
\newcommand{\rght}{\mathit{right}}

\newcommand{\htree}{\mathit{htree}}
\newcommand{\hbst}{\mathit{hbst}}

\newcommand{\bst}{\mathit{bst}}
\newcommand{\key}{\mathit{key}}

\newcommand{\tree}{\mathit{tree}}
\newcommand{\dagraph}{\mathit{dag}}

\newcommand{\lst}{\mathit{list}}

\newcommand{\dlist}{\mathit{dlist}}
\newcommand{\slist}{\mathit{slist}}
\newcommand{\rlist}{\mathit{rlist}}
\newcommand{\sdlist}{\mathit{sdlist}}
\newcommand{\lseg}{\mathit{lseg}}
\newcommand{\nxt}{\mathit{next}}

\newcommand{\minr}{\mathit{minr}}
\newcommand{\maxr}{\mathit{maxr}}
\newcommand{\leftmost}{\mathit{leftmost}}
\newcommand{\slst}{\mathit{slist}}

\newcommand{\hlist}{\mathit{hlist}}

\newcommand{\oddlist}{\mathit{odd\mathrm{-}list}}

\newcommand{\reach}{\mathit{reach}}

\newcommand{\keys}{\mathit{keys}}

\newcommand{\maxheap}{\mathit{maxheap}}
\newcommand{\slseg}{\mathit{slseg}}

\makeatletter
\newtheorem*{rep@theorem}{\rep@title}
\newcommand{\newreptheorem}[2]{%
\newenvironment{rep#1}[1]{%
 \def\rep@title{#2 \ref{##1}}%
 \begin{rep@theorem}}%
 {\end{rep@theorem}}}
\makeatother

\newreptheorem{theorem}{Theorem}

\makeatletter
\newtheorem*{rep@lemma}{\rep@title}
\newcommand{\newreplemma}[2]{%
\newenvironment{rep#1}[1]{%
 \begin{rep@lemma}}%
 {\end{rep@lemma}}}
\makeatother

\newreptheorem{lemma}{Lemma}

\newcommand{\sygus}{SyGuS\textrm{ }}

\newcommand{\lemmas}{\mathcal{L}}
\newcommand{\goodmod}{\mathcal{G}}
\newcommand{\axioms}{\mathcal{A}}
\newcommand{\ind}{\hspace{3ex}}
 
\newcommand{\falsemodel}{\mathit{Type\!-\!\!1}}
\newcommand{\lfpmodel}{\mathit{Type\!-\!\!2}}
\newcommand{\pfpmodel}{\mathit{Type\!-\!\!3}}

\newcommand{\bffalsemodel}{\mathbf{Type\!-\!\!1}}
\newcommand{\bflfpmodel}{\mathbf{Type\!-\!\!2}}
\newcommand{\bfpfpmodel}{\mathbf{Type\!-\!\!3}}
 
\newcommand{\pfp}{\mathsf{PFP}}
\newcommand{\ip}{\mathcal{IP}}
\newcommand{\modR}[1][R]{\mathcal{M}_{#1}}
\newcommand{\constR}[1][R]{\mathcal{C}_{#1}}

\newcommand{\recsym}{\R^{\mathit{rec}}}
\newcommand{\recdef}{\mathcal{D}}
\newcommand{\forecdef}{\recdef^{\mathit{fp}}}
\newcommand{\lfpdef}{:=_{\textit{lfp}}}

\newcommand{\SQI}{SQI}

\newcommand{\gf}{{\mathit{gf}\!}}

\newcommand{\natproofsalgo}{\textsf{\SQI}}
\newcommand{\synthesismodule}{\textsf{Synthesize}}
\newcommand{\provabilitycounterex}{\textsf{Counterexample}}
\newcommand{\truecounterex}{\textsf{BoundedCex}}

\newcommand{\fossil}{\textsc{FOSSIL}}

\newcommand{\fossilseq}{\textsc{FOSSIL-IP}}

\newcommand{\highlight}[1]{#1}

\newcommand{\mypara}[1]{\smallskip\noindent\emph{\textbf{#1}.\ }}

\newcommand{\egrecdef}{\recdef_*}
\newcommand{\eggoal}{\alpha_*}
\newcommand{\eglemma}{L_*}
\begin{document}

%% Title information
\title{Model-Guided Synthesis of Inductive Lemmas for FOL with Least Fixpoints}         %% [Short Title] is optional;
                                        %% when present, will be used in
                                        %% header instead of Full Title.
%\titlenote{with title note}             %% \titlenote is optional;
                                        %% can be repeated if necessary;
                                        %% contents suppressed with 'anonymous'
%\subtitle{Subtitle}                     %% \subtitle is optional
%\subtitlenote{with subtitle note}       %% \subtitlenote is optional;
                                        %% can be repeated if necessary;
                                        %% contents suppressed with 'anonymous'

%% Author information
%% Contents and number of authors suppressed with 'anonymous'.
%% Each author should be introduced by \author, followed by
%% \authornote (optional), \orcid (optional), \affiliation, and
%% \email.
%% An author may have multiple affiliations and/or emails; repeat the
%% appropriate command.
%% Many elements are not rendered, but should be provided for metadata
%% extraction tools.

%% Author with single affiliation.
\author{Adithya Murali}
%\authornote{with author1 note}          %% \authornote is optional;
                                        %% can be repeated if necessary
\orcid{0000-0002-6311-1467}             %% \orcid is optional
\affiliation{
  \department{Department of Computer Science}             %% \department is recommended
  \institution{University of Illinois, Urbana-Champaign}           %% \institution is required
  \country{USA}                   %% \country is recommended
}
\email{adithya5@illinois.edu}          %% \email is recommended

\author{Lucas Pe\~{n}a}
%\authornote{with author2 note}          %% \authornote is optional;
                                        %% can be repeated if necessary
\orcid{0000-0002-1898-439X}             %% \orcid is optional
\affiliation{
  \department{Department of Computer Science}             %% \department is recommended
  \institution{University of Illinois, Urbana-Champaign}           %% \institution is required
  \country{USA}                   %% \country is recommended
}
\email{lpena7@illinois.edu}         %% \email is recommended

\author{Eion Blanchard}
%\authornote{with author1 note}          %% \authornote is optional;
                                        %% can be repeated if necessary
\orcid{0000-0002-8270-8226}             %% \orcid is optional
\affiliation{
  \department{Department of Mathematics}             %% \department is recommended
  \institution{University of Illinois, Urbana-Champaign}           %% \institution is required
  \country{USA}                   %% \country is recommended
}
\email{eionmb2@illinois.edu}          %% \email is recommended

\author{Christof L\"oding}
%\authornote{with author1 note}          %% \authornote is optional;
                                        %% can be repeated if necessary
\orcid{0000-0002-1529-2806}             %% \orcid is optional
\affiliation{
  \department{Department of Computer Science}             %% \department is recommended
  \institution{RWTH Aachen}           %% \institution is required
  \country{Germany}                   %% \country is recommended
}
\email{loeding@automata.rwth-aachen.de}          %% \email is recommended

\author{P. Madhusudan}
%\authornote{with author1 note}          %% \authornote is optional;
                                        %% can be repeated if necessary
\orcid{0000-0002-9782-721X}             %% \orcid is optional
\affiliation{
  \department{Department of Computer Science}             %% \department is recommended
  \institution{University of Illinois, Urbana-Champaign}           %% \institution is required
  \country{USA}                   %% \country is recommended
}
\email{madhu@illinois.edu}          %% \email is recommended

%% Abstract
%% Note: \begin{abstract}...\end{abstract} environment must come
%% before \maketitle command
\begin{abstract}
Recursively defined linked data structures embedded in a pointer-based heap and their properties are naturally expressed in pure first-order logic with least fixpoint definitions (FO+lfp) with background theories. Such logics, unlike pure first-order logic, do not admit even complete procedures. In this paper, we undertake a novel approach for synthesizing inductive hypotheses to prove validity in this logic. The idea is to utilize several kinds of finite first-order models as counterexamples that capture the non-provability and invalidity of formulas to guide the search for inductive hypotheses. We implement our procedures and evaluate them extensively over theorems involving heap data structures that require inductive proofs and demonstrate the effectiveness of our methodology.
\end{abstract}

%% 2012 ACM Computing Classification System (CSS) concepts
%% Generate at 'http://dl.acm.org/ccs/ccs.cfm'.

%% End of generated code

%% Keywords
%% comma separated list
  %% \keywords are mandatory in final camera-ready submission

%% \maketitle
%% Note: \maketitle command must come after title commands, author
%% commands, abstract environment, Computing Classification System
%% environment and commands, and keywords command.
\maketitle

\section{Introduction}\label{sec:intro}

One of the key revolutions that has spurred program verification is \emph{automated reasoning of logics}. 
%The verification of properties of programs against a logical specification is a logical theorem, and hence automated verification is synonymous with automated logic reasoning. 
Particularly, in deductive verification, engineers write inductive invariants that punctuate recursive loops and contracts for methods and then use logical analysis to reason with \emph{verification conditions} that correspond to correctness of small, loop-free snippets. In this realm, automatic reasoning in combinations of quantifier-free theories using SMT solvers has been particularly useful; in turn, these tools are based on the logics having a \emph{decidable} validity (and satisfiability) problem~\cite{bradley07,barrett11}.

However, reasoning even with loop-free snippets of programs is challenging when the code manipulates \emph{linked data structures embedded in pointer-based heaps}.
Data structures are finite but unbounded structures that are often characterized using recursive definitions whose semantics are defined using both quantifiers and least fixpoints.
%From a model-theory perspective, they capture typically \emph{finite} but unbounded structures. 
\pldicommentout{data structure manipulation in imperative programming languages that modify the heap is complex as reasoning with them also requires capturing their footprint on the heap in order to do frame reasoning.}
%when code destructively modifies the heap. 
%Even full first-order logic on heaps with pointers modeled as unary functions is too weak to capture data structures and their properties.
%Consequently, logics over data structures often become undecidable. 

%Properties of such structures cannot be adequately expressed using quantifier-free logic, rendering much of the work on decidable SMT logics useless. While quantification can help state properties of such structures, even quantified logics cannot \emph{define} these structures (in general, finite unbounded structures are hard to axiomatize in first-order logic; the field of finite model theory is a witness to this difference~\cite{libkin04}). 

First-order logic with least fixpoint definitions (FO+\textit{lfp}) which accesses various background sorts or theories (e.g., integers and sets) is a powerful extension of first-order logic (FOL) that can define data structures and express their properties.
For example, fairly expressive dialects of separation logic have been translated to FO+\textit{lfp} in order to aid automated reasoning~\cite{calcagno05,madhusudan12,qiu13,pek14,suter10,loding18,murali20}. 
The focus of this paper is automated reasoning for first-order logics with least fixpoint definitions or recursive definitions that utilize SMT solvers for quantifier-free reasoning. 

The novel automation of FO+\textit{lfp} reasoning that we propose is a counterexample-guided synthesis of inductive lemmas utilizing \emph{complete} procedures for pure first-order (FO) reasoning.
Our framework requires the FO reasoning procedure to be able to compute counterexample models.
The technique we present can be parameterized over any FO reasoning engine able to provide counterexamples for provability. In this paper, we use a particular technique called \emph{natural proofs} that are based on systematic quantifier instantiation~\cite{loding18} and that is able to provide such counterexamples. 

\mypara{The Anatomy of Proofs for \folfp: Proofs by Induction} Unlike FOL, FO+\textit{lfp} does not admit complete procedures%\pldicommentout{
\footnote{Quick proof: define a ``number line'' (discrete linear order) using a constant $0$ and a unary function $s$ (representing successor) with FO axioms expressing that the successor of no element is $0$ and that the successor of no two different elements can be the same; second, encode the reachable configurations of a 2-counter machine (which is Turing powerful) as a relation defined using a least fixpoint, and express non-halting of the machine using this relation. This proof in fact shows that even a single recursive definition leads to validity being not recursively enumerable.}
%} %end of pldicommentout
(i.e., sound proof systems for FO+\textit{lfp} cannot admit proofs for every theorem).
Indeed, on a number line, true addition and true multiplication over the natural numbers are definable using \textit{lfp}. Hence by G\"odel's incompleteness theorem~\cite{enderton}, even quantifier-free logic with recursive definitions has an undecidable validity (and satisfiability) problem.

Humans usually prove properties involving recursive definitions (or least fixpoints) using \emph{induction}. 
We consider logics with recursive definitions, where each recursive definition is of the form 
$\forall \overline{x}.~ R(\overline{x}) \lfpdef \rho(\overline{x})$.
Theorems are expressed using first-order logic over a signature that includes these recursive definitions. 
An inductive proof of a theorem typically involves sub-proofs, 
which each identify a fairly strong property (the induction hypothesis) and its proof (the \emph{induction step}). %(base cases are not important as they can be seen as a particular case in the proof of an induction step).

In this paper, we use a more general notion of induction proofs based on pre-fixpoints, not requiring a concept of size or measure based on natural numbers upon which to induct.
We defer this notion until later and instead encourage the reader to simply think of
an inductive hypothesis as an \emph{inductive lemma} and
the induction step of the lemma as the \emph{pre-fixpoint (PFP) of the lemma}.
%More precisely, the PFP of a formula $\varphi: \forall \overline{x}. R(\overline{x}) \Rightarrow \psi(\overline{x})$, where $R$ has a recursive definition $R(\overline{x}) :=_\textit{lfp} \rho(\overline{x}, R)$ is $\textit{PFP}(\varphi): \forall \overline{x}.~\rho(\overline{x}, R \wedge \psi) \rightarrow \psi(\overline{x}) $.

The main proposal of this paper is to build automated reasoning for FO+\textit{lfp} with background theories using a combination of
(a) complete procedures for FO reasoning to prove theorems and PFPs of lemmas, and
(b) counterexample-guided expression synthesis for synthesizing lemmas (i.e., induction hypotheses) that aid in proving a theorem.

%However, as we said earlier, first-order logic with recursive definitions are extremely important in applications such as program verification because of the ubiquitous use of recursive definitions (which have least fixpoint semantics), and first order-logic is simply not sufficient in most scenarios. For example, reasoning with heaps and data structures requires more than FOL as FOL is \emph{local}. Saying that $x$ points to a finite linked list or that $x$ points to a finite tree  or that the list pointed to by $x$ is sorted is simply impossible in plain FOL (where the underlying universe models pointers in the data structure as unary functions).

We observe that proofs of the induction step (PFP) of the formula can be seen as reasoning using \emph{pure first-order logic reasoning without induction}. 
More precisely, we can think of a proof of a theorem in FO+\textit{lfp} as split into sub-proofs mediated by an \emph{induction principle} but otherwise consisting of pure FO reasoning. The induction principle says that proving the PFP (induction step) of any lemma proves the lemma. 
%We can even write the induction step as an FO axiom of the form $PFP(Lemma) \Rightarrow Lemma$, where $Lemma$ is a lemma and $PFP(Lemma)$ is its PFP/induction-step. 

%Our goal is now to separate these parts that are punctuated by induction principle.
%reasoning using a notion called \emph{inductive lemmas}.

%W\paragraph{Inductive lemmas}
%In this paper, we look at the sub-proofs that are proved by induction as \emph{inductive lemmas} that are proved using the induction principle as above. 

%Viewed in this manner, we can think of a proof of a theorem $\alpha$ in FO+\textit{lfp} as really as a proof aided by a set of lemmas $\mathcal{L}$ where each lemma is proved by proving the PFP formula corresponding to it using \emph{purely first order logic}. 
%Furthermore, $\alpha$ itself is proved using \emph{pure first order logic} using the lemmas in $\mathcal{L}$.

We can thus view the structure of an induction proof of a theorem $\alpha$ as identifying a finite set $\mathcal{L} = \{ L_1, \ldots, L_n\}$ of lemmas such that:
\begin{itemize}
    %\item Identify a set $\mathcal{L} = \{ L_1, \ldots, L_n\}$ of finitely many lemmas.%,   where each $L_i$ is of the form $\forall \overline{x}.~R_i(\overline{x}) \Rightarrow \psi_i(\overline{x})$, where each $R_i$ is a recursively defined relation and
      %$\psi_i$ is an arbitrary FO formula over $\overline{x}$.

    \item For each $i \in \{1,\dots,n\}$, there is a purely FO proof of $PFP(L_i)$ using
    the earlier lemmas $L_1, \ldots, L_{i-1}$ as assumptions, and %, i.e.,
     %$$\{L_1, \ldots, L_{i-1}\} \vdash_{\textit{FOL}}  PFP(L_i)$$
     
    \item There is a purely FO proof of $\alpha$ with the lemmas from $\mathcal{L}$ as assumptions.%, i.e.,
             %$$\{L_1, \ldots, L_n\} \vdash_{\textit{FOL}} \alpha$$
\end{itemize}

Notice that proofs of the above form lack any explicit induction proof and the purely FO proofs work under the assumption that each inductive relation $R$ is interpreted as a fixpoint definition (not least fixpoint) of the form 
$\forall \overline{x}.~ R(\overline{x}) \iff \rho(\overline{x})$ rather than
$\forall \overline{x}.~ R(\overline{x}) \lfpdef \rho(\overline{x})$.
The fact that proving $PFP(L_i)$ suffices as a proof of $L_i$ is implicit and marks the only appeal to the least fixpoint semantics of recursive definitions to argue that the above constitutes a proof of the theorem.
%and where the true least fixpoint semantics of inductive definitions is used (this is due to the special structure of the lemmas; see Section~2 for details).
%The above form of proofs is quite powerful and captures induction proofs in general, especially if the lemmas can use relations that have new inductive definitions. 
%Our goal hence divides into the two goals of finding lemmas and proving the theorem and the PFP of the lemmas in first order-logic. 

%In this paper, where we study automating finding such proofs, we will limit ourselves a bit on the kind of lemmas we consider: (a) we will assume that there are no new inductive relations, and lemmas can only express properties involving the inductive relations defined by the problem, and (b) that the formula $\psi$ doesn't have further quantification. We show that such restricted proofs are already quite powerful in our evaluation in aiding proofs by induction. The latter restriction is quite an important one for the techniques we develop as it ensures finite countermodels.

This view of an inductive proof of an FO+\textit{lfp} formula as pure FO proofs mediated by induction principles suggests a ``synthesis $+$ reasoning'' methodology: 
(a) synthesize lemmas that are \emph{likely} to be true and inductively provable, and (b) prove theorems and lemmas using pure FO reasoning.
%
%has the clear advantage  that we can \emph{automate} the proof of each of the required FO proofs using established \emph{complete} techniques such as quantifier instantiation. Consequently, the main challenge for automation becomes finding lemmas that are inductively provable using pure FO reasoning and help prove the main theorem, again using pure FO reasoning. 

\highlight{We emphasize that proving inductive lemmas followed  by pure FO reasoning to prove a theorem is itself not new.}
%
%having essentially FO proofs punctuated by lemmas that are inductive, and whose inductiveness is proved without using further induction (but which can use other lemmas that are in turn proved by induction), is quite natural in formal proof systems that use induction. 
For example, the induction axiom schema in Peano arithmetic is:\\
\centerline{$ \forall \overline{y}.~(\varphi(0, \overline{y}) \wedge (\forall x.~ \varphi(x, \overline{y}) \Rightarrow \varphi(S(x), \overline{y}))) \Rightarrow \forall x.~\varphi(x, \overline{y}) $}\\
for \emph{any} formula $\varphi$.
A proof using this axiom can hence be seen as divining formulas $\varphi$ and proving lemmas of form $\forall x.~\varphi(x, \overline{y})$ by using purely first-order logic over the non-inductive axioms to prove $\forall \overline{y}.~(\varphi(0, \overline{y}) \wedge (\forall x.~ \varphi(x, \overline{y}) \Rightarrow \varphi(S(x), \overline{y})))$.

\highlight{The idea of finding proofs by induction by 
synthesizing inductive hypotheses and proving them using simpler non-inductive reasoning is also not new. This technique is prevalent, for example, in program verification. In this setting, inductive hypotheses are written as loop invariants or method contracts that capture invariants of program states or effects of calling procedures. Synthesizing such invariants and contracts has been explored using a combination of inductive synthesis and reasoning (see work on the ICE framework~\cite{garg2014ICE}, for example, that explicitly takes this approach, and also the related work section). The novelty of our work lies in realizing this technique for proving theorems in FO+lfp using finite models that witness invalidity and non-provability for counterexample-guided synthesis.} 

%Similarly, in reasoning about programs with recursion using deductive verification, verification conditions can be seen as pure FO proofs with loop invariants/method contracts serving as inductive lemmas that need to be divined. In fact, automated verification approaches for loop invariant synthesis often look upon it as a synthesis problem~\cite{}, and some of them use  counterexamples~\cite{ICE}.
%Our work can be seen as a general framework for using synthesis guided by counterexample FO models and pure FO reasoning for proving theorems in FO+\textit{lfp}. 

%%The lemma of course must use a formula $\varphi$ that is a strong enough induction hypothesis.
%madhu: removing
%Similarly, verification of programs with iteration/recursion can also be seen as an induction proof, and the while-rule in Hoare logic:
%$$\frac{\{I \wedge B\}\ C\ \{I\}}
%{\{I\}\ \textbf{while } (B)\ C\ \{I \wedge \neg B\}}\text{~~~(while rule)}$$
%can be seen as an induction principle where the applications of the while-rule can be seen as proving lemmas of the form ${\{I\}\ \textbf{while } (B)\ C\ \{I \wedge \neg B\}}$ by
%proving the Hoare triple ${\{I \wedge B\}\ C\ \{I\}}$
%using non-inductive axioms for other statements.
%%Again, the inductive loop invariant $I$ must be chosen strong enough for the proof to go through.

% \smallskip
% \noindent
% {\bf Synthesizing Inductive Lemmas:}
\mypara{Synthesizing Inductive Lemmas} The primary technical contributions of this paper lie in
techniques for synthesizing lemmas that (a) can be proved inductively, with their own statement as the induction hypothesis, and (b) aid the proof of a target theorem. We embrace the paradigm of 
\emph{counterexample-guided synthesis}
that has met impressive success in automating 
verification and synthesis (e.g., in finding predicates for abstraction~\cite{slam,namjoshi00} or in program synthesis through the CEGIS paradigm~\cite{alur15, Armando2007Stencils, ArmandoSolarLezamaThesis}).
%, and we follow this paradigm in our effort to synthesize inductive lemmas. 
%Guidance using counterexamples as a paradigm has been used in verification for finding appropriate abstractions, where the abstractions are guided using counterexamples that depend on the property being verified (e.g., in predicate abstraction~\cite{namjoshi00}). In synthesis of expressions, in general, the paradigm of \emph{counterexample-guided inductive synthesis} has emerged as a powerful paradigm~\cite{alur15}. 
%We develop in this paper a novel technique for synthesizing inductive lemmas using counterexample \emph{first order models}. 
The salient feature of our technique is the use of \emph{finite first-order models} that act as counterexamples to guide the search for lemmas.

Suppose a theorem $\alpha$ in \folfp is desired to be proved valid.
Our technique for automated quantified FO reasoning (without least fixpoints), called \emph{natural proofs}, uses systematic quantifier instantiation followed by SMT-based validation of the resulting quantifier-free formula~\cite{loding18,pek14,qiu13}. %Given a FO formula to check validity, we can instead check for satisfiability of its negation, and Skolemize this negated formula to remove existential quantification.
Let $\SQI(k)$ be the method that systematically instantiates terms of depth $k$ for quantified variables then checks satisfiability of the resulting quantifier-free formula (the latter is a decidable problem).
As a simple consequence of Herbrand's theorem and compactness, we know that this method is complete in the sense that if $\beta$ is a valid formula in FOL, then there is some $k$ for which $\SQI(k)$ will prove the validity of $\beta$.

At any point of the lemma synthesis procedure, we would have synthesized a set of potentially useful lemmas already proved valid and then seek a new lemma to help prove $\alpha$.

We utilize \emph{three\,} kinds of counterexample models to guide the search for useful and provable lemmas. 
In our iterative framework for synthesizing useful and provable lemmas, a prover and a synthesizer interact: the synthesizer proposes lemmas, and the prover provides constraints for synthesizing new lemmas.
When the synthesizer proposes a lemma, the lemma can be (a) valid and provable using $\SQI(k)$ reasoning using existing lemmas, (b) invalid but easily shown to be so using a small model, or (c) valid or invalid, but in either case not provable using $\SQI(k)$ and existing lemmas.
Note that (a) and (c) cover all cases, and (b) overlaps with (c). 

These correspond to the three kinds of counterexamples, which we now name.
$\falsemodel$ models guide the search toward lemmas that help prove the theorem $\alpha$ and are obtained from the failure to prove $\alpha$ using FO reasoning via $\SQI(k)$.
$\lfpmodel$ models are small, simple counterexamples to validity of proposed lemmas and are obtained by searching for bounded models using SMT solvers.
$\pfpmodel$ models show non-provability of lemmas and are obtained from failure to prove the PFP of lemmas using FO reasoning via $\SQI(k)$.
By utilizing these three kinds of counterexample models, we narrow and guide the search space for lemma synthesis.

The main contribution of this paper is \fossil, a novel algorithmic framework for synthesizing lemmas that uses such counterexamples and proves both lemmas and target theorems using FO reasoning.
In each round, the algorithm begins with a target theorem $\alpha$ and tries proving it using the lemmas synthesized and proved valid so far.
If the proof of $\alpha$ fails, this failure precipitates a $\falsemodel$\, counterexample which will be used to guide the search towards lemmas that do help prove the theorem $\alpha$.
The lemma synthesis phase follows, generating a lemma that satisfies the $\falsemodel$\, counterexample and then attempting to prove the validity of its PFP.
If the proof of the PFP fails, this failure yields either a $\lfpmodel$\, counterexample (which is a bounded model) if possible or otherwise a $\pfpmodel$\, counterexample to show non-provability of the PFP.
We continue to seek new lemmas guided by these three kinds of counterexamples until a valid lemma is found, at which point we add the new lemma to our set of valid lemmas. We recurse, trying to prove the target theorem $\alpha$.
Off-the-shelf synthesis tools do not scale when employed in our framework; however, our synthesis engine works efficiently via constraint solving with SMT solvers, carefully representing counterexamples as \emph{ground formulas} and formulating synthesis constraints as ground constraints.

\mypara{Background Theories and Relative Completeness} The techniques for inductive reasoning that we develop in this paper are more involved than as described above.
First, many applications, such as program verification, require handling of domains that are constrained to satisfy certain theories, such as arithmetic and sets
(sets allow the expression of collections such as ``the set of keys stored in a list'' in heap-based verification and ``the set of heap locations that constitute a list'' in heaplets for frame reasoning).
Consequently, our framework maintains a \emph{foreground sort} modeling the heap with pointers as well as multiple background sorts, with the background sorts constrained by theories and that admit Nelson-Oppen style decision procedures for \emph{quantifier-free} reasoning.
%for quantifier-free formulas even for combined theories.
%(these are typically Nelson-Oppen combinations of theories supported by SMT solvers). 
In such settings, the work in~\cite{loding18} proved that for formulas that quantify only over the foreground sort (i.e., only involving quantification over locations of the heap), systematic  quantifier instantiation is still complete.
Moreover, satisfiability of quantifier-free formulas after instantiation are supported by SMT solvers, which can also return the three kinds of counterexamples we seek.

%We cater only to  theorems that quantify over the foreground sort in this work.
%, captures a large class of properties of heaps. Moreover, when quantifiers are instantiated using foreground terms, the resulting formula are quantifier-free, and are hence amenable to Nelson-Oppen based reasoning of combined theories.
%We build lemma synthesis procedures over such domains.
%for FO reasoning. 
%In particular, if we allow quantifiers over variables ranging over the foreground sort only, quantifier instantiation is complete, and furthermore, instantiation leads to decidable quantifier-free formulas which, when not valid, result in finite counterexample models. We hence restrict ourselves to proving theorems using lemmas where both theorems and lemmas quantify only over the foreground sort.

Second, we carefully build lemma search to admit relative completeness.
We show that if there is a proof of a theorem involving finitely many independently provable lemmas (in the grammar of lemmas provided by the user),
then our procedure is guaranteed to eventually find one.
More precisely, there are two \emph{infinities} to explore---one is the search for lemmas and the other is the instantiation depth $k$ chosen for finding proofs.
As long as our procedure fairly dovetails between these two infinities, it is guaranteed to find a proof. 

%We prove that our first algorithm is relatively complete in finding \emph{independent lemmas} that prove the theorem (lemmas proved on their own, without depending on other lemmas). Furthermore, we propose a modification of this algorithm that is even complete for finding \emph{seq}

%The relevance of this theorem is it shows that the counterexamples we generate to guide the lemma synthesis process is \emph{sound} in that they do not eliminate lemmas that could help the theorem and are valid. 

% \smallskip
% \noindent {\bf Evaluation:}
\mypara{Evaluation} We implement and evaluate our procedure for a logic that combines an uninterpreted foreground sort with background sorts, where the background sorts have quantifier-free fragments that  are decidable using SMT solvers. 
Our tool framework can employ generic syntax-guided synthesis (SyGuS) engines as well as a custom synthesis tool we built; 
both of these can synthesize lemmas using FO countermodels that are encoded using logical constraints.
%The latter is a synthesis engine  we develop that synthesizes lemmas from \emph{grounded} constraints which are sufficient to express constraints on first-order models. 
%We implement the term instantiation technique to prove FOL properties (used for proving both lemmas and the main theorem) as well as model generation for all counterexamples using an SMT solver. 
%And implement the search for lemmas that satisfies the constraints with respect to finite models using CVC4-Sy, a {\sc SyGuS} solving engine. 

We perform an extensive evaluation on two suites of benchmarks:
%of over 700 theorems.
one of 50 theorems on data structure verification and another of 673 synthetically generated theorems. 
Our experiments give evidence that the first-order counterexample-based techniques proposed in this paper are effective in synthesizing inductive lemmas and proving theorems. Apart from evaluating the efficiency of our tool, we evaluate the importance of several design decisions and optimizations in our tool. In particular, we study the efficacy of using various kinds of counterexamples and compare our custom synthesis engine with off-the-shelf state-of-the-art synthesis engines. 

%We also compare our tool, as much as possible, with tools that do inductive reasoning in other logics, in particular algebraic datatypes (ADTs) and separation logic. 

Lemma synthesis has been studied for related logics, in particular for logics over algebraic datatypes (ADTs)~\cite{reynolds15, Weikun19} and separation logic~\cite{slcomp19, trung17}. Though these logics are very different in expressive power and comparisons across tools are hard, we provide a comparison of our tool against tools for these logics on our benchmark theorems using appropriate encodings whenever feasible. 

\mypara{Contributions} The main contributions of this paper are: (1) a counterexample-guided synthesis framework, \fossil, for synthesizing inductive lemmas for proving validity in \folfp~ with relative completeness guarantees, (2) the formulation of three kinds of counterexamples that guide synthesis towards lemmas that are relevant to the theorem, lemmas that hold at least on small models, and are provable using induction, (3) efficient synthesis algorithms using specifications formulated as ground formulas, and (4) an implementation and evaluation of \fossil~ on two benchmark suites of theorems in the domain of heap data structures\footnote{Our benchmarks and tool can be found at: \url{https://github.com/muraliadithya/FOSSIL}}.

\section{Preliminaries and Problem Definition} 
\label{sec:prelim}

In this section, we define the first-order logic framework
we work with (first-order logic with recursive definitions that have \textit{lfp} semantics) and give the problem definition for solving theorems in FO+\textit{lfp} using synthesis of inductive
lemmas and first-order proofs.

%In this paper, we deal with dialects of first-order logics with least fixpoints; syntactically, least fixpoints will be presented using recursive definitions. We  work with both a pure logic (with an empty theory, i.e., the theory where equality is assumed to be congruence) as well as a logic that combines several background theories in a particular way.

\subsection{First-Order Logic over Theory-Constrained Background Sorts}
\label{sec:uct}
%We work with first-order logic (with and without recursive definitions that have 
%least fixpoint semantics. 
%\todomp{Needs fixing}
The first-order logics (with and without recursive definitions) that we work with are over a multisorted universe that has a single distinguished \emph{foreground} sort %$\fs$ 
and multiple \emph{background} sorts. % $\sigma_1, \ldots, \sigma_n$ ($n \geq 0$).
The universes of all these sorts are pairwise disjoint. 
The foreground sort and the functions and relations that refer to it (as part of the domain or codomain) are entirely uninterpreted (no axioms that constrain them).
Background sorts and functions and relations involving only background sorts are constrained by certain theories.
%In particular, we work only in settings where the combined theories admit a  decidable validity problem for  \emph{quantifier-free} formulas (e.g., those that can be combined using Nelson-Oppen combinations~\cite{demoura08,nelson80}) such as arithmetic with addition, sets of locations, sets of integers, etc.
%, while background sorts are restricted by certain theories. 
%Further, we allow recursive definitions be defined only on the foreground sort. In our applications the foreground sort will be used to model heaps, with pointers modeled using uninterpreted unary functions, pointer variables modeled using first-order variables, and data structures modeled using recursive definitions. 

Formally, we work with a signature of the form $\Sigma = (S; C; F; \R)$, where
$S$ is a finite non-empty set of sorts.
$C$ is a set of constant symbols, where each $c \in C$ has some sort $\sigma \in S$.
$F$ is a set of function symbols, where each function $f \in F$ has a type of the
form $\sigma_1 \times \ldots \times \sigma_m \rightarrow \sigma$ for some $m$, with
$\sigma_i, \sigma \in S$. 
$\R$ is a set of relation symbols, where each relation $R \in \R$ has
a type of the form $\sigma_1 \times \ldots \times \sigma_m$. 

We assume a designated foreground
sort, denoted by $\fs$. 
All other sorts in $S$ are called background sorts, and for
each such background sort $\sigma$ we allow the constant symbols of type $\sigma$, function symbols that have type
$\sigma^n \rightarrow \sigma$ for some $n$, and relation symbols that have
type $\sigma^m$ for some $m$ to be constrained using an arbitrary
theory $T_\sigma$. 
All other functions and relations that involve either the foreground sort or multiple background sorts are assumed to be uninterpreted (not constrained by any theory). %an also be assumed to be constrained by some (not necessarily complete) theory $T$. 
We consider standard first-order logic (FO) over these multisorted signatures, with standard syntax and semantics, under the combined theories~\cite{enderton}. 

\mypara{Counterexamples} We require that validity of \emph{quantifier-free} logic under the combined theories is \emph{decidable}. 
Furthermore, when a quantifier-free formula is not valid, we require this decision procedure to provide models that show satisfiability of the negation of the formula. \highlight{The truth value of the quantifier-free formula only depends on a finite portion of the model (corresponding to the terms used in the formula, since the formula is quantifier-free). This finite portion can be described by a conjunction of atomic ground formulae.
We require models to be given indirectly by such \emph{conjunctive ground formulae}.
Formally, given a quantifier-free formula $\varphi$ that is satisfiable, we require that the solver return a conjunctive ground formula $\gf$ such that (a) $\gf$ is satisfiable and (b) $\gf \Rightarrow \varphi$ is \emph{valid}. If $\varphi$ contains variables, then these are interpreted as or replaced by Skolem constants that are part of the signature of $\gf$. Intuitively, $\gf$ indicates the existence of one or more models such that $\varphi$ is satisfied on all of them. The formula $\gf$ encodes enough information about these models to ensure that $\varphi$ is satisfied in them. The following example illustrates these ideas.
%Intuitively, the formula $\gf$ encodes enough information of models that ensures that $\varphi$ is always satisfied. 

\begin{example}[Counterexample models as conjunctive ground formulas]
Consider the formula $(f(x)=y) \Rightarrow y>3$ where $f: \sigma_0 \rightarrow \mathit{Int}$ is an uninterpreted function, %from the foreground sort to the background sort of integers,
$x$ is of the sort $\sigma_0$, and $y$ is of sort $\mathit{Int}$. %is over the foreground sort and $y$ is over the sort of integers.
This formula is invalid, and we can witness the satisfiability of its negation $(f(x) = y) \land \neg(y>3)$ using a model $\Mm$ where $x$ is interpreted to an element $u$ and $f(u)$ is interpreted to $2$. $\Mm$ can be captured using the formula $\gf: f(x)=2$. Indeed, one can see that $\big(f(x) = 2\big) \Rightarrow \big((f(x) = y) \land \neg(y >3)\big)$ is a valid formula. It is also imminent that $\gf$ is satisfiable since $\Mm$ realizes it.
\end{example}
%For instance, for the formula $f(x)=y \Rightarrow y>3$ where $f: \sigma_0 \rightarrow \mathit{Int}$ is an uninterpreted function from the foreground sort to the background sort of integers, $x$ is over the foreground sort and $y$ is over the sort of integers), its negation is satisfied in a model that, say, interprets $x$ to be an element $u$ and with $f(u)=2$. This can be captured using the formula $\gf: f(x)=2$.
%
}

In our tools, we work with certain Nelson-Oppen combinable decidable theories~\cite{nelson80,demoura08,bradley07,nelson-oppen1979} (in particular linear arithmetic over integers, sets of integers). These are  supported by SMT solvers that guarantee both decidability of quantifier-free formulae as well as model generation as above.

\subsection{First-Order Logic with Recursive Definitions (\folfp)}
\label{sec:prelim-folfp}
Our target theorems are in a dialect of first-order logic over a multisorted universe (universes similar to the one above) but with recursive definitions that have least fixpoint semantics. 

We identify a subset $\recsym$ of the relational symbols $\R$ and endow them with definitions; these relations are not directly interpreted by models, rather they are defined uniquely by their definitions. In our work we assume that these recursive definitions only relate elements of the foreground sort. The set of recursive definitions $\recdef$ for the symbols $\recsym$ are of the form

\centerline{$R(\overline x) \lfpdef \rho_R(\overline{x})$}

\noindent
where $R \in \recsym$, $\overline{x}$ are variables over the foregreound sort, and $\rho_R(\overline{x})$ is a quantifier-free first-order logic formula. Note that a definition $\rho_R$ can utilize all the sorts and functions/relations in the model. We also assume that there is only one definition for each $R \in \recsym$.

\highlight{To ensure the well-definedness of definitions, %least fixpoint of these definitions to be well-defined, 
we assume that the symbols in $\recsym$ are ordered in layers, and that each $R' \in \recsym$ that occurs in the definition of $R$ is either in a smaller layer, or it is in the same layer and only occurs positively (under an even number of negations) in the definition of $R$ (similar to stratified Datalog~\cite{gradel-modeltheorybook}). The semantics of recursively defined relations is given by the
least fixpoint (\textit{lfp}) that satisfies the relational equations (the condition that each recursive definition only refers positively to recursively defined relations in the same layer ensures that the least fixpoint exists~\cite{tarski1955})\footnote{\highlight{Our definition of \folfp is similar to the one used in Finite Model Theory: see Libkin~\cite{libkin04}, Chapter 10. Notably, our notion of recursive definitions is more restrictive than general \folfp because recursive definitions should only be universally quantified and only over the foreground sort. This technical condition enables us to build effective complete FO validity procedures: see Section~\ref{sec:natural-proofs}.}}.

Our theoretical treatment assumes that there is only one layer for simplicity. Therefore, each recursive definition only mentions other recursively defined relations positively. %defined relation only occurs positively in all the definitions. 
However, the results also hold for several layers of recursive definitions, and indeed our experiments utilize them.} %Finally, we also assume that there is only one definition for each $R \in \recsym$.
%\highlight{By abuse of notation we also use $\recdef$ to denote the set of definitions; the use of $\recdef$ as either a set of definitions or a set of symbols that have recursive definitions will be mentioned explicitly when it is not clear from context.}

\begin{example}[Linked Lists]
Let $n$ be a unary function symbol modeling a pointer of type $\sigma_0 \rightarrow \sigma_0$, i.e., from the foreground sort to the foreground sort. Let $\textit{nil}$ be a constant of sort $\sigma_0$, and $\lst$ be a unary relation with the recursive definition

\centerline{
$\lst(x) \lfpdef \ite(x = \nil, \mathit{true}, \lst(n(x)))$}

Then, in any model $\Mm$ %that interprets $n$ and with
where $\lst$ is interpreted using its \textit{lfp} definition, $\lst$ holds precisely for those elements that are the head of a \emph{finite} linked list with $n$ as the next pointer. FOL without \textit{lfp} cannot describe such linked lists~\cite{libkin04}. Note that unlike Algebraic Datatypes (ADTs), if $\lst(x) \wedge \lst(y) \wedge x \not= y$ holds in a model, the lists pointed to by $x$ and $y$ are not necessarily disjoint and could ``merge'' in the model. We can also model disjointedness using heaplets, as we show in the following example.%model tree-like data structures as well by capturing heaplets of trees using recursive definitions and demanding disjointedness of left and right subtrees (similar to separation logic). 
\end{example}

\begin{example}[Trees and Heaplets]
\label{ex:tree-lfpdef}
% \todomp{Adding definition of  trees with heaplets below. That will show the difference with ADTs.}
%\madcomment{Fix tree def}
%\pldicommentout{

%Similarly we can express that $x$ is the root of a tree using the recursive definitions:
Consider the following recursive definition for a predicate $\mathit{tree}(x)$ which expresses that $x$ is the root of a binary tree on pointers $l$ (left) and $r$ (right):
\begin{flalign*}
\mathit{tree}(x) \lfpdef \mathit{ite}(x\!=\!nil,\, \mathit{true},\,& \mathit{tree}(\textit{l}(x)) \wedge \mathit{tree}(\textit{r}(x))&\\[-2pt]
&\land\, \mathit{htree}(\textit{left}(x)) \cap \mathit{htree}(\textit{right}(x)) \!\!=\!\! \emptyset&\\[-2pt]
&\land\, \textit{Singleton}(x) \cap \left(\mathit{htree}(\textit{l}(x)) \cup \mathit{htree}(\textit{r}(x))\right) \!\!=\!\! \emptyset)&\\
\mathit{htree}(x) \lfpdef ite(x=nil,\, \emptyset,\,& \mathit{Singleton}(x) \cup \mathit{htree}(\textit{l}(x)) \cup 
        \mathit{htree}(\textit{r}(x)) )
\end{flalign*}

% $\mathit{tree}(x) \lfpdef \mathit{ite}(x\!=\!nil, \mathit{true},$\\
% \centerline{$\mathit{tree}(\textit{l}(x)) \wedge \mathit{tree}(\textit{r}(x))
% \land\, \mathit{htree}(\textit{left}(x)) \cap \mathit{htree}(\textit{right}(x)) \!\!=\!\! \emptyset$}
% \centerline{$\land\, \textit{Singleton}(x) \cap \left(\mathit{htree}(\textit{l}(x)) \cup \mathit{htree}(\textit{r}(x))\right) \!\!=\!\! \emptyset)$}

% $\mathit{htree}(x) \lfpdef ite(x=nil, ~\emptyset, $\\
% \centerline{$        \mathit{Singleton}(x) \cup \mathit{htree}(\textit{l}(x)) \cup 
%         \mathit{htree}(\textit{r}(x)) )$}
        
%Note that here we have two pointers left $\textit{l}$ and right  $\textit{r}$. 
Observe again that since our data structures are unlike ADTS, pointers $l$ and $r$ may possibly point to the same element (``merge'') in arbitrary heaps/models.
Therefore, to define trees we define a recursive definition for the partial function expressing the \emph{heaplet} of a tree $\mathit{htree}: \sigma_0 \rightarrow \sigma_{sl}$ where $\sigma_{sl}$ is a background theory of sets of locations with which we demand that the left and right subtrees are disjoint. This  is similar to constraints used in Separation Logic to express trees~\cite{reynolds02}.
%Unlike separation logic, $\mathit{tree}(x)$ is true in any global heap where $x$ points to a tree (there can be other locations in the model that do not belong to the tree).
%Then if the above definition is the unique definition in $\recdef$, then $(\recdef, \lst(u) \wedge \lst(v))$ is a formula, which intuitively says that $u$ and $v$ point to two (finite) lists that end with $\textit{nil}$ (note that they need not be disjoint).
%}% end of pldicommentout
\end{example}

%\medskip
%A model that correctly interprets the relations from $\recdef$ as the LFP of their recursive definitions is called an LFP model. 
%Target theorems to prove valid are quantified formulas $\alpha$, where  quantification in $\alpha$ is over the foreground sort and $\alpha$ is expressed using both the interpreted functions/relations as well as the recursively defined relations.

\noindent
We now state the usual notion of validity/entailment in \folfp in the language introduced above.

\begin{definition}[\folfp Entailment]
For a sentence $\alpha$ and a set $\Gamma$ of formulas we write $\Gamma \cup \recdef \modelslfp \alpha$ if $\alpha$ is true in all models of $\Gamma$ using the lfp semantics for relations with definitions given in $\recdef$.
\end{definition}

\noindent
We conclude this section with some remarks.

%\paragraph{First-order abstractions of recursive definitions:} 
\mypara{First-Order Abstractions of Recursive Definitions}
Given an FO+\textit{lfp} formula, we can sometimes prove it valid using pure FOL. 
We can do this by interpreting recursive definitions in $\recdef$ to be
\emph{fixpoint definitions} (as opposed to \textit{lfp}).
More precisely, we constrain the relations using  FOL as 
$\forall \overline{x}.~ R(\overline x) \leftrightarrow \rho_R(\overline{x})$.
If $\alpha$ is valid under the fixpoint interpretation of recursive relations, then it is of course valid using least fixpoint interpretation as well, but the converse does not hold. 
Interpreting recursive definitions as fixpoint definitions rather than least fixpoint definitions is hence a form of sound abstraction. 
We write 
$\Phi \cup \forecdef \modelsfo \alpha$ to denote that $\alpha$ is valid
using the FO fixpoint abstractions $\forecdef$ of $\recdef$.
%if $\alpha$ is true in all models of $\Phi$ in which the relations $R \in \recdef$ satisfy the FO version of the recursive definitions. Note that 
%if $\Phi \cup \recdef \modelsfo \alpha$ holds, then $\Phi \cup \recdef \modelslfp \alpha$, but the converse is not necessarily true. 

\mypara{Partial Functions} The reader may have observed in Example~\ref{ex:tree-lfpdef} that we presented a recursively defined \emph{function} $\mathit{htree}$. Although we don't allow them in the theoretical treatment, our tools support recursively defined partial functions from the foreground sort to both foreground and background sorts (for modeling heaplets of structures, lengths of lists, heights of trees, etc.). %but we don't allow them in the theoretical treatment.
However, partial functions can be modeled using two predicates: one recursively defined predicate that captures the domain of the partial function and another predicate defined using only FOL that captures the map of the function. 

\highlight{
\mypara{\folfp~ Fragment} %We make some remarks about the expressive power of the \folfp fragment that we handle. 
In this work we only handle the validity of formulas whose quantification is purely over the foreground sort. This fragment is well suited for the domain of heap verification that we study. We can model the heap as the foreground sort and express recursively defined functions and properties that only quantify over the heap. However, it is not as powerful as full \folfp. For example, the logic cannot talk about array properties (where the array is modeled as a map $f:Int \rightarrow V$ from indices to values in a domain $V$) that quantify over integers, which is a background sort. We also cannot express theorems like ``For every positive integer $n$, there is a linked list of length $n$'' as this requires universal quantification over the background sort. %and existential quantification over the foreground sort. 
These restrictions are important as they allow us to leverage practical complete algorithms~\cite{loding18} for FOL validity for this restricted fragment in implementing the $\fossil$ framework (see Section~\ref{sec:natural-proofs}).}

\subsection{The Inductive Lemma Synthesis Problem for Proving  FO+\textit{lfp} Formulas}
\label{sec:induction-proofs}

In this work we develop algorithms that prove an \folfp\  formula $\alpha$ valid given a finite set $\axioms$ of axioms and a set $\recdef$ of recursive definitions with \textit{lfp} semantics. We want to show that $\axioms \cup \recdef \modelslfp \alpha$ mainly using first-order reasoning. Clearly, if $\axioms \cup \forecdef \modelsfo \alpha$, then $\axioms \cup \recdef \modelslfp \alpha$ as argued above. %since a least fixpoint is also a fixpoint. However, the other direction is not true in general.

\medskip
\noindent
\highlight{We use the following running example to illustrate ideas developed in the sequel:

\begin{example}[Running Example]
\label{ex:running}
Consider the recursively defined relation $\lseg(x,y)$ defining linked list \emph{segments} between locations $x$ and $y$ on the pointer $n$:

\centerline{
$\lseg(x,y) \lfpdef \ite(x = y, \mathit{true}, \lseg(n(x),y))$}

\smallskip
\noindent
Now, consider the following Hoare Triple:
%{\color{purple}
% \begin{verbatim}
% {@pre: lseg(x,y1)} if (y1 == nil) then y2 := y1; else y2 := y1.n; {@post: lseg(x,y2)}
% \end{verbatim}
%}

%\centerline{\texttt{\{@pre: lseg(x,y1)\} if (y1 == nil) then y2 := y1; else y2 := y1.n; \{@post: lseg(x,y2)\}}}

\centerline{$\{\texttt{@pre:} lseg(\texttt{x,y1})\} \texttt{ if (y1 == nil) then y2 := y1; else y2 := y1.n;} \{\texttt{@post:} lseg(\texttt{x,y2})\} $}

\smallskip
\noindent
The above triple generates the following Verification Condition (VC) $\eggoal$:
% \begin{equation*}
%   \forall x,y_1,y_2.\, \lseg(x,y_1) \Rightarrow \bigg(\ite(y_1 = \nil, y_2 = y_1, y_2 = n(y_1)) \Rightarrow lseg(x,y_2) \bigg) \tag{$\eggoal$}  
% \end{equation*}

\centerline{
$\lseg(x,y_1) \Rightarrow \bigg(\ite(y_1 = \nil, y_2 = y_1, y_2 = n(y_1)) \Rightarrow lseg(x,y_2) \bigg)$}

We denote by $\egrecdef$ the singleton set containing the definition of $\lseg$. We will use the problem of checking $\egrecdef \modelslfp \eggoal$ as a running example in this paper. Note that $\eggoal$ is actually valid in \folfp~ but it is not FO-valid, i.e., $\egrecdef \modelslfp \eggoal$ holds but $\egrecdef^{\mathit{fp}} \modelsfo \eggoal$ does not. This makes the problem a good candidate for lemma synthesis. We describe a run of our algorithm on this example in Section~\ref{sec:illustrative-example}. %\qed
%\madcomment{Put program here and intro alpha star as its VC. Also say that it is valid in fo-lfp but not FO-valid (which shows why we are using this asan example).}
\end{example}
}

\noindent
The overall idea in our approach is to use intermediate inductive lemmas to find an FO proof of the goal. %$\alpha$ %as described in Section~\ref{sec:algorithm}.
\highlight{We handle a particular fragment of \folfp in our work. First, we require the goal $\alpha$ to have quantification only over the foreground sort. Second, we only consider lemmas of the form $L = \forall \overline{x}.~ R(\overline{x}) \Rightarrow \psi(\overline{x})$ for variables $\overline{x}$ over the foreground sort, a quantifier-free formula $\psi$, and a recursively defined relation $R \in \recsym$. Finally, we prove lemmas valid using a specific form of induction called the pre-fixpoint (PFP) formula.} 
%The following induction principle $\textit{IP}(L)$ for the lemma $L$ 
%(called strong induction principle in \cite{loding18}) 
%expresses the fact that if
%$\psi(\overline{x})$ is a pre-fixpoint of the recursive definition of $R$, then the lemma is valid.
Given a lemma $L$ of the form above, the PFP of $L$ expresses that $R \wedge \psi$ is a pre-fixpoint of the definition of $R$:

\centerline{$\textit{PFP}(L) := \forall \overline{x}. \rho_R(\overline{x}, R \land \psi) \Rightarrow \psi(\overline{x})
$}

\noindent
where $\rho_R(\overline{x}, R \land \psi)$ is the formula
obtained from $\rho_R(\overline{x})$ by replacing every occurrence of
$R(t_1, \ldots, t_k)$ for terms $t_1, \ldots,t_k$ in $\rho_R$ by
$\psi(t_1, \ldots,t_k) \land R(t_1, \ldots, t_k)$.
It turns out that if $PFP(L)$ is FO-valid, then $L$ is a valid \folfp formula, as the following theorem states:
%The induction principle is valid for any lemma, and can be used to prove a lemma correct in the LFP semantics, as given by the following result:

\begin{theorem}
\label{thm:ip-correctness}\cite{loding18}
If $\axioms \!\cup\! \forecdef \!\modelsfo\! \textit{PFP}(L)$, then $\axioms\! \cup\! \!\recdef\! \modelslfp L$.
\end{theorem}
%This is an easy consequence of LFP semantics (see~\cite{loding18} for a formal proof).

\noindent
We use the above formalism to define the notion of an \emph{inductive lemma}, as well as the notion of a sequence of lemmas that prove a theorem using FO reasoning.%The above motivates the following definition of an inductive lemma.

\begin{definition}[Inductive Lemmas]
\label{def:lem-ind-proof}
A lemma $L$ is inductive for $\axioms \cup \forecdef$ if $\axioms \cup \forecdef \modelsfo PFP(L)$. If $\axioms$ and $\recdef$ are clear from the context, we omit them and just say that $L$ is inductive.%\qed
\end{definition}

\highlight{
\begin{example}[Running Example: Inductive Lemma]
\label{ex:running-pfp}
Consider in the setting of Example~\ref{ex:running} the following lemma $\eglemma$:
\begin{equation*}
\forall x,y_1,y_2.\, \lseg(x,y_1) \Rightarrow \bigg(\lseg(y_1,y_2) \Rightarrow \lseg(x,y_2)\bigg)\tag{$\eglemma$}
\end{equation*}

\noindent
which expresses that if we have a list segment pointed to by $x$ until $y_1$, as well as one pointed to by $y_1$ until $y_2$, then $x$ points to a list segment until $y_2$. It turns out that $\eglemma$ is inductive i.e., $\egrecdef^{\mathit{fp}} \modelsfo PFP(\eglemma)$. In other words, the $PFP$ of the lemma is provable in pure FOL, without induction, and with FO abstractions of the definitions (fixpoint instead of least fixpoint).

The crucial part of the proof is the following %``induction step'' that occurs as a 
subformula of $PFP(\eglemma)$:

\centerline{$\forall x,y_1,y_2.\,\big(\lseg(y_1,y_2) \Rightarrow \lseg(n(x),y_2)\big) \Rightarrow \big(\lseg(y_1,y_2) \Rightarrow \lseg(x,y_2)\big)$}

\noindent
which is valid given $\egrecdef^{\mathit{fp}}$ since, according to the definition of $\lseg$, if $lseg(n(x),y_2)$ holds then $\lseg(x,y_2)$ also holds (in the non-degenerate case). %As mentioned earlier, the above formula states that the right-hand side of $\eglemma$ is a pre-fixpoint, i.e., is inductively preserved with respect to the structure of the left-hand side predicate $\lseg(x,y_2)$.\adcomment{We provide the detailed proof in the Appendix --- Make sure to put this in the appendix}
\end{example}
}
%\begin{definition}[Inductive lemmas that prove a theorem]
%\label{def:prove-thm-using-lemma}
%A set of $\lemmas = \{L_1, \ldots, L_n\}$ 
%is said to prove a theorem $\alpha$ if $\axioms \cup \recdef \cup \lemmas \modelsfo \alpha$ and 
%the lemmas are inductive, i.e., $\axioms \cup \recdef \modelsfo \mathit{PFP}(L_i)$ for each $i \in \{1, \ldots,n\}$.\qed
%\end{definition}

%There are two kinds of inductive lemmas that are relevant in this paper. The first one is the central problem that we wish to solve, namely synthesizing sequential lemmas that result in proving a theorem:

%and synthesizing independent lemmas, where the former is more general than the latter. Our implementation and evaluation is for the much more common and more efficiently solvable case of synthesizing independent lemmas. We define these formally:

%\noindent Then, we define when a sequence of lemmas prove a theorem using only FO reasoning.

\noindent
We now define the notion of proving a theorem using lemmas as well as the synthesis problem that it poses which we tackle in this work.

\begin{definition}[Sequential Lemmas that Prove a Theorem]
\label{def:thm-seq-proof}
A sequence $(L_1, \ldots, L_n)$ of lemmas provides an inductive proof of $\alpha$ if $\axioms \cup \forecdef \cup \{L_1, \ldots, L_n\} \modelsfo \alpha$
%the set $\lemmas = \{L_1, \ldots, L_n\}$ proves $\alpha$ in the sense of Definition~\ref{def:prove-thm-using-lemma} 
and for each $1 \leq i \leq n$, $L_i$ is inductive for $\axioms \cup \forecdef \cup \{L_1, \ldots, L_{i-1}\}$ (i.e., $\axioms \cup \forecdef \cup \{L_1, \ldots,L_{i-1}\} \modelsfo PFP(L_i)$).%\qed
\end{definition}

% \smallskip
% \noindent We now define the central problem tackled in this paper:

\smallskip
\begin{definition}[\textbf{Sequential Lemma Synthesis Problem}]~
\label{def:seqlemsynth}
   Given a grammar $G$ for expressing lemmas and a theorem $\alpha$, find a sequence of lemmas admitted by $G$ that provides an inductive proof of $\alpha$ (as in Definition~\ref{def:thm-seq-proof}).%\qed
\end{definition}

\mypara{Independently Proven Lemmas} We can also define a simpler synthesis problem corresponding to a weaker class of inductive proofs. Specifically, we can require a set of lemmas that are \emph{independently} proven inductive and help prove a theorem:

\begin{definition}[Independent Lemmas that Prove a Theorem]
\label{def:thm-ind-proof}
A set $\{L_1, \ldots, L_n\}$ of lemmas provides an inductive proof of $\alpha$ if $\axioms \cup \forecdef \cup \{L_1, \ldots, L_n\} \modelsfo \alpha$
%the set $\lemmas = \{L_1, \ldots, L_n\}$ proves $\alpha$ in the sense of Definition~\ref{def:prove-thm-using-lemma} 
and for each $1 \leq i \leq n$, $\axioms \cup \forecdef \modelsfo PFP(L_i)$.%\qed
\end{definition}

The difference between the two classes of proofs is that the inductiveness of lemmas in a sequential proof can depend on previous lemmas. As one might expect, the notion of proof using independent lemmas is strictly weaker than the one that uses a sequence of lemmas. %, while in a proof based on independent lemmas, the proof of inductiveness of each lemma is independent of the other lemmas.
%Note that if a set of lemmas $\{L_1,\ldots, L_n\}$ is a set of independent lemmas for a theorem $\alpha$, then $(L_1, \ldots, L_n)$ is a sequence of lemmas for $\alpha$ as well, while the converse is not always true.
\highlight{We conclude this section with the running example.

\begin{example}[Running Example: Lemma Proving a Theorem]
\label{ex:running-thm-and-lemma}
Consider $\eglemma$ and $\eggoal$ introduced earlier in the running example. Now, observe that $\egrecdef^{\mathit{fp}} \cup \{\eglemma\} \modelsfo \eggoal$. This is because the crucial part of the validity of $\eggoal$ is the following formula:

\centerline{$\lseg(x,y_1) \Rightarrow \bigg((y_1 \neq \nil \land y_2 = n(y_1)) \Rightarrow lseg(x,y_2)\bigg)$}

\noindent
which captures the `else' case of the $\ite$ subformula of $\eggoal$ (see Example~\ref{ex:running}). We can see that $\eglemma$ entails the above formula in FO since (informally) $y_2 = n(y_1)$ is a special case of $lseg(y_1,y_2)$. Combined with the fact that $\eglemma$ is inductive (Example~\ref{ex:running-pfp}), we have that $\eglemma$ proves $\eggoal$ in the sense of Definition~\ref{def:thm-seq-proof}. We illustrate a run of our synthesis algorithm that proves $\eggoal$ by synthesizing $\eglemma$ in Section~\ref{sec:illustrative-example}.
\end{example}
}

In Section~\ref{sec:algorithm} we present our core algorithm $\fossil$ for solving the \emph{sequential lemma synthesis problem}. This algorithm, apart from being sound in producing sequential lemmas that prove the theorem, is accompanied by a relative completeness result:
it is guaranteed to find a proof as long as there is a set of \emph{independent} lemmas that prove the theorem. %$\fossil$ is not, however, relatively complete for finding sequential lemmas when they exist. We discuss some variations of the \fossil~ algorithm that are complete for sequential lemma synthesis in the full version of the paper\footnote{The full version of the paper can be found at: \url{https://arxiv.org/abs/2009.10207}}. % discusses variations of the algorithm that are indeed complete for sequential lemma synthesis, but being significantly more expensive, we do not pursue these in the main paper.

\subsection{Background: First-Order Validity using Systematic Quantifier Instantiation}
\label{sec:natural-proofs}

%We have so far seen the definition of the problem of Sequential Lemma Synthesis and briefly discussed our approach to solving it using a combination of a validity checker for FOL formulae and a synthesis engine guided by counterexample models. In this section we will formally describe the particular kind of FO provers that we require, the definition of a counterexample, and discuss an FO prover from existing literature that instantiates these desiderata.

In this section we describe the Systematic Quantifier Instantiation (\SQI) mechanism for 
FO validity (without recursive definitions/lfp) that we use, developed in the work~\cite{loding18}. The results in this section are derived from the work in~\cite{loding18} and are not contributions of this paper.

%Quantifier instantiation is a well-known technique for proving quantified FO formulae over a single-sorted signature. 
Let $\varphi$ be an FO formula. To check the validity of $\varphi$, we negate and Skolemize it --- introducing both Skolem constants and Skolem functions --- and obtain a purely universally quantified formula $\psi$ such that $\varphi$ is valid if and only if $\psi$ is unsatisfiable. 
%Then, by Herbrand's theorem we have that $\Phi$ is unsatisfiable if and only if there exists a finite set of terms $T$ such that the $\Phi[T]$ is unsatisfiable over many-sorted FO (see Section~\ref{sec:uct}). 
Let $\psi$ be of the form $\forall \overline{x}.\,\eta(\overline{x})$ where 
%all the quantified variables $\overline{x}$ are of the foreground sort, and 
$\eta(\overline{x})$ quantifier-free.
For a set of ground terms $T$, we denote by $\psi[T]$ the set of all quantifier-free formulas that are obtained by instantiating the variables $\overline{x}$ in $\psi$ by terms in $T$, i.e.,

\centerline{$
\psi[T] := \{\eta(\overline{t}) \mid \overline{t} \textit{~is a tuple of terms in T of arity ~} |\overline{x}| \}.
$}

It follows that if $\psi[T]$ is unsatisfiable then $\psi$ is unsatisfiable and therefore $\alpha$ is valid. Since we assume in our setting that satisfiability/validity of quantifier-free formulas is decidable (see Section~\ref{sec:uct}), checking whether $\psi[T]$ is unsatisfiable is decidable. %since we assume that quantifier-free formulas admit decision procedures for satisfiability/validity.

%quantify only over the foreground sort, quantifier instantiation results in replacing variables using only ground terms of the foreground sort. Hence $\Phi[T]$ is a set of quantifier-free formulas over the combined theories, which is decidable since we assume that the background theories permit decidable Nelson-Oppen~\cite{nelson80} style combination. 
%Furthermore, it turns out that this procedure is also complete (see~\cite{loding18}). 

\mypara{Systematic Quantifier Instantiation} The above suggests a complete semi-decision procedure for validity based on systematic quantifier instantiation (\SQI). Let $\psi \equiv \forall \overline{x}.~ 
\eta(\overline{x})$ be the formula that we want to check for unsatisfiability where $\overline{x}$ are variables of the foreground sort and $\eta$ is quantifier-free. For any $k \in \mathbb{N}$, let $T_k$ denote the set of all ground terms whose type is the foreground sort and are of \emph{depth} at most $k$ (we assume that the signature contains at least one constant symbol for the foreground sort).
Then, starting with $k=0$, we check whether $\psi[T_k]$ is unsatisfiable. If it is then we halt and report that $\varphi$ is valid; otherwise, we increment $k$ and repeat. This motivates the following definition:

\begin{definition}[Provability at depth $k$ using \SQI]
\label{def:depthk-provability}
A formula $\varphi$ is provable at depth $k$ using \SQI~
%, denoted by $\models_{\sqi(k)} \alpha$, 
if the negated and Skolemized formula $\psi$ is such that $\psi[T_k]$ is unsatisfiable.\qed
\end{definition}
%\adcomment{Using double instead of single turnstile is weird notation. Fix it?}
The above is a sound procedure, i.e., if $\varphi$ is provable at depth $k$ using SQI (for some $k$) then it is clearly valid. It is also a complete procedure for validity in pure first-order logic without any theories (i.e., just uninterpreted functions). %, the above is in fact a complete procedure. 
This follows from Herbrand's theorem and the compactness theorem. It turns out that this continues to be a complete procedure in the multisorted setting for the kind of FOL formulas that we work with. i.e., those that \emph{quantify only over 
the foreground sort}. We formally state below this result from the work in~\cite{loding18}: %, the above procedure is also a complete procedure, which is a result from~\cite{loding18}.

\begin{theorem}[From~\cite{loding18}]
\label{thm:natproofs-complete-fo-bg-safe}
Let $\varphi$ be a formula with quantification only over the foreground sort. Then $\varphi$ is valid if and only if there exists $k \in \mathbb{N}$ such that $\varphi$ is provable at depth $k$ using \SQI. 
%$\models_{\sqi(k)} \alpha$.
\end{theorem}

%A solver can hence choose increasingly larger bounds $k$ and attempt to prove the given quantified goal by generating instantiations using terms of depth $k$. We will refer to this method as proof by using \emph{natural proofs of depth $k$}.

%Theorem~\ref{thm:natproofs-complete-fo-bg-safe} proved in~\cite{loding18}) shows that natural proofs is a complete technique for proving formulas in FO for formulas that quantify only over the foreground sort. 
%We use this result to prove the completeness of our own algorithms. 
%\adcomment{Do we know how to handle formulae that have quantification over the background sort if we don't want completeness? What will we do for counterexamples? This bit likely needs to be rewritten saying that we restrict ourselves to this fragment because.....completeness is a good guiding principle, maybe?}
%In particular our benchmarks use one-way functions, as can be seen from Section~\ref{sec:evaluation}.

We implement and use \SQI~ for proving validity of first-order logic formulae in this work.
%In this work, we construct and use an implementation of such a procedure.

%\paragraph{Skolemization.}

%We assume that the set $\Phi$ that is instantiated consist only of universal formulas. In the algorithms presented in Section~\ref{sec:algorithm}, we sometimes have to extend $\Phi$ by the negation of a universal formula, for example a formula $\neg PFP(L)$. Then the universal quantifiers become existential, and we have to use Skolemization, which means in  this setting that the existentially quantified variables are replaced by constants. We sometimes write $\neg PFP(L[\overline{c}])$ to explicitly refer to these constants.

%The induction principle for $L$ is of the form $PFP(L) \rightarrow L$, which as a disjunction is $\neg PFP(L) \lor L$. So when we add an induction principle to a set of formulas that is instantiated, $\neg PFP(L)$ is also Skolemized. We sometimes use the notation $IP(L,\overline{c})$ to indicate that the constants $\overline{c}$ are used for the Skolemization of the $\neg PFP(L)$ part of the formula.
%Note that this Skolemization extends the signature by new constants, so the instantiation to depth $k$ needs to consider more terms when a Skolemized formula is added. 

\section{The FOSSIL Algorithm for Sequential Lemma Synthesis}\label{sec:algorithm}

In this section, we present the fundamental contribution of this paper: \fossil~ (First-Order Solver with Synthesis of Inductive Lemmas), our algorithm for solving the Sequential Lemma Synthesis problem formulated in Definition~\ref{def:seqlemsynth}. Figure~\ref{fig:orchestra} shows the components of our framework %that the \fossil~ algorithm depends on, and 
which we describe 
%these components first 
in Section~\ref{sec:fossil-components}. \fossil~ is a counterexample-based lemma synthesis algorithm that orchestrates interactions between these external components through three kinds of counterexamples. %: a first-order verifier based on natural proofs, a synthesis engine (for synthesizing lemmas given logical constraints on them), and a bounded counterexample generator (for proposed lemmas). 
We formally define these counterexamples in Section~\ref{sec:counterexamples}. We then present the \fossil~ algorithm 
%in Figure~\ref{fig:alg1} and describe it step-by-step
in Section~\ref{sec:algorithm-independent-lemmas}. Finally, we illustrate a run of \fossil~ on our running example in Section~\ref{sec:illustrative-example}. %We state and prove the relative completeness of \fossil~ with respect to synthesis of independent lemmas in Section~\ref{sec:fosil-ind-completeness}
(the algorithm is guaranteed to find a proof if there is a set of independent lemmas that prove the goal).
%We follow this description with a fully worked-out illustrative example in Section~\ref{sec:illustrative-example}. Then, we describe each of the components in Figure~\ref{fig:orchestra} in detail in Sections~\ref{sec:synthesis-module} and~\ref{sec:counterexample-module}.

\begin{figure}[ht]
\centering
\includegraphics[width=0.9\textwidth]{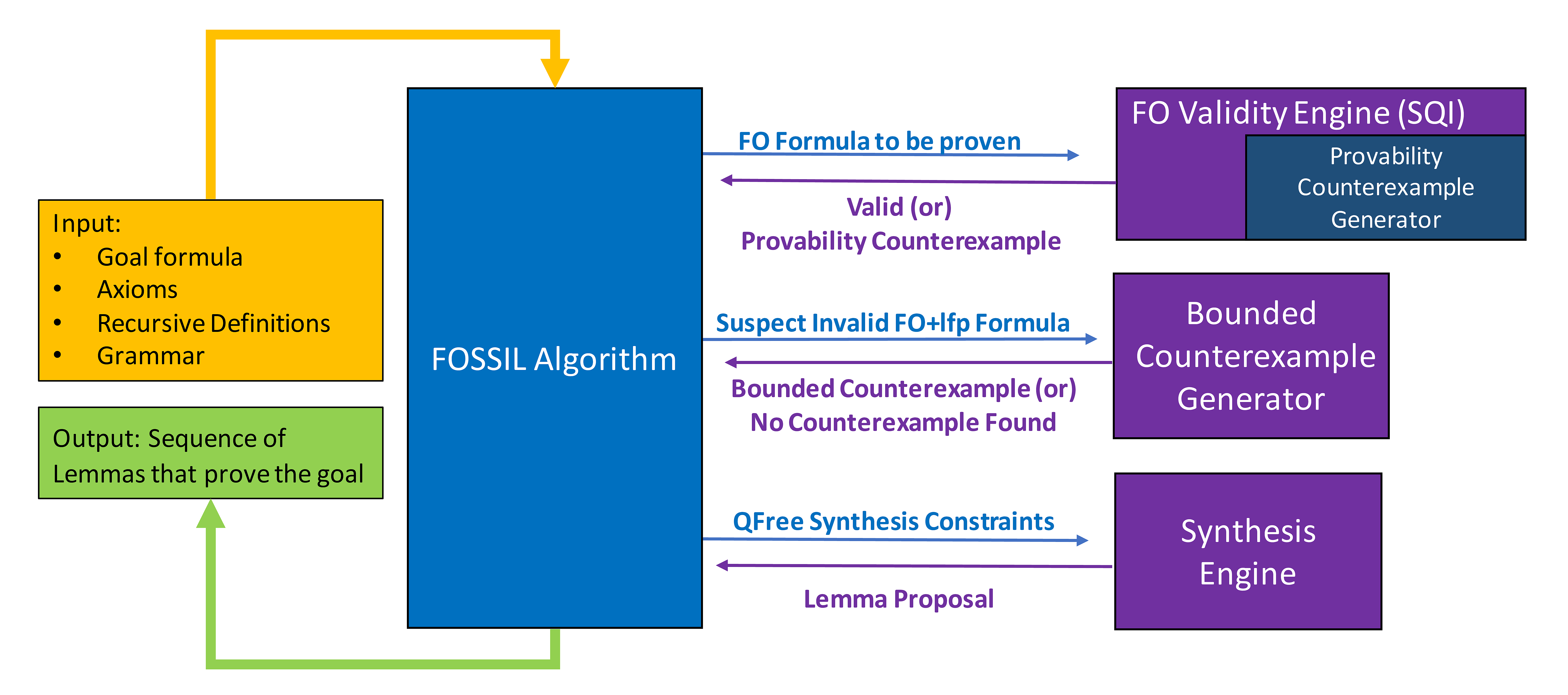}
\caption{Components of \fossil.}
\label{fig:orchestra}
\end{figure}

\subsection{Components of \fossil}
\label{sec:fossil-components}

In this section, we discuss the external components used by the core \fossil~ algorithm. We only describe \emph{what} these components are, deferring implementation details to Section~\ref{sec:evaluation}. 
%\madcomment{Fix.}
%provide a bird's eye view of \fossil.
%and next provide a more detailed step-by-step description of \fossil.
Let us fix 
%an FO+\textit{lfp} theory specified by 
a set $\axioms$ of axioms and a set $\recdef$ of recursive definitions throughout the following presentation. We also fix a \emph{goal} formula $\alpha$ and a grammar $\mathcal{G}$ for lemmas. \highlight{We assume that $\axioms$ consists of universally quantified sentences, and that $\alpha$ is also a universally quantified sentence (using Skolemization if necessary).}

\fossil~ finds a proof of $\alpha$ by synthesizing a sequence of lemmas $\lemmas = (L_1, L_2, \ldots, L_n)$ belonging to $\mathit{Lang}(\mathcal{G})$ such that $\lemmas$ is a sequence of lemmas proving $\alpha$ according to Definition~\ref{def:thm-seq-proof}. %The provability is established in the FO+\textit{lfp} theory given by $\axioms$ and $\recdef$. 
%\paragraph{Bird's Eye View} 
The high-level external components of \fossil~  %and the interactions between them and \fossil~ 
are shown as purple boxes/arrows in Figure~\ref{fig:orchestra}. 
%The main orchestration module (`\fossil~ Algorithm') is described in Section~\ref{sec:algorithm-independent-lemmas}.
%and we give its pseudocode in Figure~\ref{fig:alg1} in terms of the satellite components.
We describe their abstract interface %of these satellite components 
below in terms of formulae and counterexamples. We encourage the reader to think of counterexamples as 
\emph{finite FO models} for now, pending their formalization %Counterexamples will be formalized 
in Section~\ref{sec:counterexamples}.
%is a recursively enumerable search for a sequence of lemmas that prove $\alpha$. We present the pseudocode for the orchestration module labeled by `\fossil~ Algorithm' (Figure~\ref{fig:orchestra}) in Figure~\ref{fig:alg1} which uses the following abstractions corresponding to the high-level components:
The components of \fossil~ are:%is parameterized by the following components:
\begin{enumerate}[leftmargin=10pt]
    \item {\bf First-Order Validity Engine $\natproofsalgo(\varphi, k)$}:
    This is an FO validity checking algorithm based on Systematic Quantifier Instantiation (see Section~\ref{sec:natural-proofs}). It takes as input a formula $\varphi$ and a natural number $k$, and outputs whether $\varphi$ \emph{valid} or \emph{unprovable} at depth $k$ using \SQI.
    %$\models_{\sqi(k)} \varphi$

    \item {\bf Provability Counterexample Generator} $\provabilitycounterex(\varphi,k)$: 
    This is a counterexample generation module that is part of the FO validity engine. When a formula is found to be unprovable (using term instantiation with terms of depth $k$) it returns a \emph{finite counterexample model}. This model is one in which $(\neg \varphi)[T_k]$ holds. $(\neg \varphi)[T_k]$ is the negation of the formula instantiated by terms up to depth $k$. Intuitively, the counterexample witnesses the \emph{non-provability} of $\varphi$ using term instantiation with depth $k$ terms. Note that finite counterexample models (i.e., where the foreground universe is finite) always exist because $(\neg \varphi)[T_k]$ is a quantifier-free formula. This module is used to generate the $\falsemodel$ and $\pfpmodel$ counterexamples (the inputs being the goal or a proposed lemma respectively). The types of counterexamples are explained below in Section~\ref{sec:counterexamples}.
    %
    %that 
    %witnesses the satisfiability of $\varphi$ if it is satisfiable, and otherwise returns \emph{no model}. It is depicted as a submodule of the verifier module in Figure~\ref{fig:orchestra}. When the natural proof engine instantiates the quantified variables in the negation of the theorem with terms, we get a quantifier-free formula, which if found satisfiable, we extract a finite model using this module. 

    \item {\bf Bounded Counterexample Generator} $\truecounterex(\varphi, \mathit{size})$:
    Given an FO+\textit{lfp} formula $\varphi$ and a parameter $\mathit{size}$ this module returns a finite model with at most $\mathit{size}$ elements in the foreground sort that shows that the formula is \emph{not valid}, if possible. It may also return that such a model could not be found (because one may not exist at that size). These models interpret recursively defined predicates using the true \textit{lfp} semantics and will be used as $\lfpmodel$ counterexamples. %    a module that generates FO+\textit{lfp} models of bounded size. 
    %It corresponds to the module labeled `Bounded True Counterexample Generator' in Figure~\ref{fig:orchestra}. The module is given a negated formula $\neg\varphi$, a set of axioms and recursive definitions that define a FO+\textit{lfp} theory $\mathit{Th}$, and a positive integer $\mathit{size\_bound}$. It either finds a FO+\textit{lfp} model of the given theory bounded in size by the given integer that witnesses the invalidity of $\varphi$, or returns that no such model exists with the given size bound. 

    \item {\bf Synthesis Engine} $\synthesismodule(C, \mathcal{G})$:
    This module synthesizes candidate lemmas. It takes as input 
    a set of counterexample models expressed as quantifier-free constraints $C$ and a grammar $\mathcal G$, and generates an expression in $\textit{Lang}(\mathcal G)$, if one exists, that avoids all counterexamples.
    %(this notion is formalized in Section~\ref{sec:counterexamples} below)
    %if one exists. 
    %corresponding to a synthesis module, labeled by `Synthesis Engine' in Figure~\ref{fig:orchestra}. It takes a set of constraints $C$ over a variable $L$ referring to a quantified formula and a grammar $\mathcal{G}$, and either returns a formula belonging to the language $\mathit{Lang}(\mathcal{G})$ satisfying the constraints $C$, or outputs that such a formula does not exist in the given grammar.

\end{enumerate}

\noindent
We now turn to the definition of the three kinds of counterexamples used in \fossil.
%The \fossil~ algorithm utilizes the above components. Before we delve into the algorithm, it is important to define a more precise notion of counterexamples and how they are generated.% and how they curtail the search space of lemmas.

\subsection{Counterexamples}
\label{sec:counterexamples}
\fossil~ is a \emph{counterexample-guided} algorithm that uses the verification and synthesis components in rounds of lemma proposals. In this section, we define the notion of the various counterexamples that we use.

While counterexamples can be intuitively thought of as finite models for the foreground universe, we will formally treat them as conjunctive ground formulae as described in Section~\ref{sec:uct}. 
For example, consider the model depicting a one-element linked list on the pointer $n$. The foreground universe has two elements, say $v_1$ and $v_2$, such that $\mathit{n}(v_1) = v_2$ and $\textit{nil}$ is interpreted to be $v_2$. Then the ground formula $\gf \equiv v_1 \neq v_2 \land v_2 = \nil \land n(v_1) = v_2$ with new constant symbols $v_1$ and $v_2$ %$\gf \equiv e \neq \mathit{nil} \land \mathit{next}(e) = \mathit{nil}$ with a new constant symbol $e$ 
defines a class of models that contains the intended model. 
In general, a ground formula captures a \emph{class} of models where a finite portion of the model is constrained by the formula.

\highlight{In our algorithm we evaluate formulas over tuples of elements on models represented by a ground formula $\gf$. We use the notation $\gf(\overline{c})$ to indicate that the model contains interpretations for the constants in $\overline{c}$, and we use this tuple to instantiate the variables of formulas that we evaluate over the model. For example, see  lines~\ref{alg1:truemodel-constraints} and~\ref{alg1:pfp-constraints} of the \fossil~ algorithm (Figure~\ref{fig:alg1}). Similarly, we refer to a set of elements interpreted by a model by $C$ and use it to evaluate a formula on all tuples over $C$, %over a specific set of elements, then we refer to this set $C$ of elements by $\gf(C)$, 
as in line~\ref{alg1:falsemodel-constraints}.}

The \fossil~ uses three kinds of counterexamples. 
%to guide the synthesis module towards lemmas satisfying these two conditions.
Let us fix $\mathit{ctx} \equiv \bigwedge (\axioms \cup \recdef^{\mathit{fp}} \cup \lemmas)$ to be the \emph{context} formula containing the axioms, recursive definition abstractions, and the valid lemmas $\lemmas$ discovered so far.

%\begin{itemize}
%\subsubsection*{$\bffalsemodel$ counterexamples}
\mypara{$\bffalsemodel$ Counterexamples} $\falsemodel$ counterexamples guide the synthesis toward lemmas that help prove the goal.
Given a term depth $k$, $\falsemodel$ counterexamples witness   
    %that witnesses the current non-FO-provability of 
    non-provability of the goal $\alpha$ 
    using term instantiation with depth $k$ terms. In other words, this is a counterexample to the non-provability of $\mathit{ctx} \Rightarrow \alpha$ using the instantiation. 
    
    \highlight{Formally, a $\falsemodel$ counterexample at depth $k$ is a satisfiable ground formula $\gf_1$ such that $\modelsfo \gf_1 \Rightarrow (\mathit{ctx} \land \neg \alpha[\overline{c}])[T_k]$, where $\overline{c}$ is a tuple of Skolem constants resulting from the Skolemization of the existential quantifiers in the negation of $\alpha$. Such a model witnesses that $\alpha$ cannot be proven from $\mathit{ctx}$ by instantiation with terms of depth $k$. We use $\falsemodel$ counterexamples in \fossil~ in lemma synthesis by accessing tuples of elements that correspond to terms in $T_k$ (line~\ref{alg1:falsemodel-constraints} in Figure~\ref{fig:alg1}). We name these elements and represent them as a set $C$, denoting the counterexample by $\gf_1(C)$.
    %In the constraints for synthesizing new lemmas, we need to access the elements of the model represented by $\gf_1$ that correspond to terms in $T_k$ . We let $C$ be a set of names for these elements and emphasize this by writing $\gf_1(C)$. These elements are used in line~(\ref{alg1:falsemodel-constraints}) of \fossil~ (Figure~\ref{fig:alg1}).
    }
    
    \smallskip
    \noindent\emph{Generating $\falsemodel$ counterexamples:} Recall from Section~\ref{sec:uct} that in our setting, for any satisfiable quantifier-free formula $\varphi$, we can obtain a satisfying model as a conjunctive ground formula. When failing to prove $\mathit{ctx} \Rightarrow \alpha$ using depth $k$ term instantiation we obtain a satisfiable conjunctive ground formula from the satisfiability of $(\mathit{ctx} \land \neg \alpha[\overline{c}])[T_k]$. This is the $\falsemodel$ counterexample.
    %This ground formula hence satisfies the property for being a  $\falsemodel$ counterexample, and this is the counterexample we generate. 
    
%    {\bf How to use the counterexample to search for lemmas:} 
%
%    When searching for a new lemma, we will insist that the new lemma $\forall \overline{x} L(\overline{x})$ satisfy
%    $\modelsfo gf_1 \Rightarrow \neg L[T(gf_1)]$. 
     
%    This guides the synthesis towards new lemmas that help prove the goal as required by condition (i).
    
%    Note that provability at depth $k$ using \SQI~ is defined by $\models_{\sqi(k)} \mathit{ctx} \Rightarrow \alpha$, which is exactly $\modelsfo \left(\mathit{ctx} \land \neg\alpha\right)[T_k]$. Therefore the counterexample formula $\gf_1$ defines a class of models where the goal is not provable using depth $k$ terms. 
    
% \subsubsection*{$\bflfpmodel$ counterexamples}
\mypara{$\bflfpmodel$ Counterexamples}    
These counterexamples correspond to finite models (i.e., those in which the foreground sort is finite) that falsify a candidate lemma $L$ in \folfp. Such a model $\Mm$ satisfies $\Mm \modelslfp \mathit{ctx} \land \neg L$.
When a lemma $L$ is proposed, creating a small model in which $L$ is false can easily show its invalidity.

Formally, a $\lfpmodel$ counterexample for a lemma of the form $\forall \overline{x} R(\overline{x}) \rightarrow \psi(\overline{x})$ is represented as a ground formula $\gf_2(\overline{c})$ with constants $\overline{c}$ of the foreground sort such that $|\overline{c}| = |\overline{x}|$. The formula can include constraints involving relations in $\recsym$. $\gf_2$ interprets the recursively defined predicates with \textit{lfp} semantics. We require that there exists an FO model $\Mm$ whose interpretation for predicates in $\recsym$ matches their recursive definitions $\recdef$ (we describe how we implement this requirement in Section~\ref{sec:counterexample-module}). Finally, we require $\modelsfo \gf_2(\overline{c}) \Rightarrow R(\overline{c}) \wedge \neg \psi(\overline{c})$. %such that $M$ satisfies $\gf_2(\overline{c})$ and  $M \modelslfp \gf_2$. 

%Note that in these counterexamples the valuations for recursively defined functions must be consistent with the least-fixpoint semantics. 
    
For example, consider the finite model consisting of two locations, say $e_1$ and $e_2$, where $e_1$ is the head of a one-element linked list and $e_2$ points to itself on the $n$ pointer. This model is captured by the formula $e_1 \neq \mathit{nil} \land e_2 \neq \mathit{nil} \land \mathit{next}(e_1) = \mathit{nil} \land \mathit{next}(e_2) = e_2 \land \lst(\mathit{nil}) \land \lst(e_1) \land \neg\lst(e_2)$. Note that the correct valuation of $\lst$ on this universe is given by the formula. %, as they cannot be defined in FOL.
    
\smallskip
\noindent\emph{Generating $\lfpmodel$ counterexamples:} %In order to generate these counterexamples 
We fix a bound $\mathit{size} \in \mathbb{N}$ and use an SMT solver to identify a model with at most $\mathit{size}$ elements in the foreground sort that falsifies the lemma, if one exists. We provide further details in Section~\ref{sec:counterexample-module} and Section~\ref{sec:evaluation}.
    %finite FO+\textit{lfp} models of the given theory that witness the invalidity of a lemma proposed by the synthesis module.
    %Since every desired lemma must be valid in \folfp~, these counterexamples help the synthesis module propose fewer invalid lemmas by satisfying condition (ii).
    %These are the models returned by the True Counterexample Generator component in Figure~\ref{fig:orchestra} and are referred to as `true' because the valuation of recursive definitions is consistent with least-fixpoint semantics throughout the model, unlike the other two types of counterexamples.% where the definitions are consistent with \textit{lfp} semantics only \emph{locally} (i.e., only on the instantiated terms).
    
% \subsubsection*{$\bfpfpmodel$ counterexamples}
%finite FO models that witness the failure of an induction proof of a proposed lemma.
\mypara{$\bfpfpmodel$ Counterexamples} $\pfpmodel$ counterexamples guide the search towards lemmas that are inductively provable using their PFP. When the PFP of a  proposed lemma is found to be unprovable (using depth $k$ term instantiation), we obtain a counterexample that witnesses the non-inductiveness of $L$ (with respect to the lemmas discovered so far). Note that we do not actually know whether the lemma is valid/invalid or provable/unprovable as it may require discovering other lemmas or a bigger instantiation depth. This is similar to a $\falsemodel$ counterexample, where instead of the target theorem we generate counterexamples to the PFP of a candidate lemma. 
    
Formally, a $\pfpmodel$ counterexample for a lemma $\forall \overline{x} R(\overline{x}) \rightarrow \psi(\overline{x})$ is a ground formula $\gf_3(\overline{c})$ with $|\overline{c}| = |\overline{x}|$ such that $\modelsfo \gf_3(\overline{c}) \Rightarrow \left(\mathit{ctx} \land \neg \mathit{PFP}(L)[\overline{c}]\right)[T_k]$ holds and $\gf_3$ is satisfiable. The constants $\overline{c}$ are Skolem constants obtained from Skolemizing the existential formula 
$\neg \mathit{PFP}(L)$.
    
    %When the PFP of a  proposed lemma is found to be not valid using $k$-depth term instantiation, we obtain  the resulting quantifier-free formula 
    %A $\pfpmodel$ counterexample witnesses the non-provability of inductiveness for a particular proposed lemma $L$. 
    %Since induction is proved using the $\mathit{PFP}$ formula, 
     %We want the lemmas we discover to be provable by induction, and therefore such a $\gf_3$ 
    %They witness that $L$ is not provable by induction, which is a property that we require of valid lemmas that we want discover.

\smallskip
\noindent\emph{Generating $\pfpmodel$ counterexamples:} 
Similar to $\falsemodel$ counterexamples, the generation of $\pfpmodel$ counterexamples is done using the quantifier-free formula obtained from the proof failure of $\mathit{ctx} \Rightarrow \mathit{PFP}(L)$ using depth $k$ term instantiation.
    
    %They also correspond to condition (ii), but these models do not show that the proposed lemma is invalid, only that we cannot prove it by induction and the current discovered and proven lemmas. They 
    %and help the synthesis module propose fewer unprovable lemmas. They constrain the synthesis less strongly than $\lfpmodel$ models since they are not counterexamples to validity, only counterexamples to provability. However unlike $\lfpmodel$ models, these models can always be found for unprovable lemmas even if we are not able to check if they are truly invalid. 

%Observe that all the counterexample models are finite, i.e., they have finite foreground universes (see Section~\ref{sec:prelim} for a description of First-Order models over multiple sorts as used in this work).

%\subsection*{Counterexamples Guiding Search}
%In the following sections we describe the mechanism by which we formulate synthesis queries for lemmas that avoid these counterexamples. We first define when a formula $\rho$ \emph{avoids} a counterexample t the set of foreground sort (ground) terms appearing in $\gf$ be $T(\gf)$.

%\begin{definition}[Avoiding a counterexample]
%A formula $\rho$ avoids $\gf$ iff $\modelsfo \gf \Rightarrow (\neg\rho)[T(\gf)])$.
%\end{definition}

%Note that this is a stricter notion than $\modelsfo \gf \Rightarrow \neg\rho$ or even $\models_{\sqi(k)} \gf \Rightarrow \neg\rho$ in general. However this is an unimportant difference in our work, therefore we elide these details for the sake of clarity and denote this by $\gf \not\models_\sqi \rho$.

%%%%%%%%%%%%%%% 
%Pseudocode for FOSSIL
\small
\begin{figure}[ht]
\begin{flushleft}
\textsc{{\bfseries \fossil} $(\axioms,\recdef,\mathcal{G},\alpha;k,h)$}

\textsc{Input:} axioms $\axioms$, recursive definitions $\recdef$, grammar $\mathcal{G}$, goal formula $\alpha$, natural proofs depth parameter $k$, lemma production height parameter $h$\\
\textsc{Output:} Sequence of valid lemmas $\lemmas \in L(\mathcal{G})$ (of height at most $h$) that prove $\alpha$ \folfp-valid using $\SQI(k)$\\
\textsc{Imports:} {\bfseries \natproofsalgo,  \provabilitycounterex, \truecounterex, \synthesismodule}

 \begin{enumerate}[leftmargin=15pt]
 \item Compute $\mathcal{G}_h \subseteq \mathcal{G}$ such that $\mathit{Lang}(\mathcal{G}_h)$ does not contain any formulas whose parse-tree in $\mathcal{G}$ has a height greater than $h$.
 \label{alg1:finite-grammar}
 %\item $T^* := T_k(\axioms \cup \recdef \cup \{\neg \alpha\} \cup \{\mathit{pfp}_l \mid \mathit{pfp}_l = PFP(L) \textrm{ for } L \in \mathcal{G}_h\})$ 
 \label{alg1:terms-overapproximation}
 \item $\lemmas := ()$, $\lfpmodel := \emptyset$, and $\pfpmodel := \emptyset$ for each $R \in \recdef$
 \item $\Phi_\alpha:= \left(\bigwedge\axioms \cup \recdef^{\mathit{fp}}\right) \Rightarrow \alpha$
 \label{alg1:PhiAlpha}
 \item \textsc{While} (~{\bfseries \natproofsalgo}$(\Phi_\alpha, k) \neq VALID$~)
 \label{alg1:outer-loop}
 \item \ind $\gf_1(C) =$ {\bfseries \provabilitycounterex}$(\Phi_\alpha,k)$
 \item \ind $\falsemodel = \gf_1(C)$
 \label{alg1:false-model}
 \item \ind \textsc{While} (True) \label{alg1:inner-loop}
     \item \ind\ind $ L =$ {\bfseries \synthesismodule}$(S, \mathcal{G}_h)$ such that $L(\overline{x}) = \forall \overline{x}. R(\overline{x}) \rightarrow \psi(\overline{x})$\\
     \ind\ind and constraints $S$ are:
     \label{alg1:lemma-proposal}
     \hspace{10em}
     \begin{enumerate}[leftmargin=1.3cm]
     %\item  $L \in \mathit{Lang}(\mathcal{G}_h)$ \\
     %       \label{alg1:height-constraints}
     \item  $\modelsfo \gf_1(C) \Rightarrow \neg (\bigwedge L[C]$), where $gf_1(C)$ is the current $\falsemodel$ model\\ 
            \label{alg1:falsemodel-constraints}
     \item  $\modelsfo \gf_2(\overline{c}) \Rightarrow L(\overline{c})$ for all $(\gf_2(\overline{c}), R) \in \lfpmodel$ \\ 
            \label{alg1:truemodel-constraints}
     \item $\modelsfo \gf_3(\overline{c}) \Rightarrow PFP(L)(\overline{c})$ for all $(\gf_3(\overline{c}),  R)\!\in\! \pfpmodel$ \\
            \label{alg1:pfp-constraints}
     \end{enumerate}
     \item \ind\ind If no lemma found, call \fossil$(\axioms,\recdef,\mathcal{G},\alpha;k\!+\!1,h\!+\!1)$
     \label{alg1:increase-depth}
     \item \ind\ind $\Phi_L := \left(\bigwedge\axioms \cup \recdef^{\mathit{fp}} \cup \lemmas\right) \Rightarrow PFP(L)$
     \label{alg1:PhiL}
     \item \ind\ind \textsc{If} ({\bfseries \natproofsalgo}$(\Phi_L, k) = VALID$)
     \label{alg1:PhiL-validity}
     \textsc{Then} // Valid Lemma
     \label{alg1:valid-lemma}
%     \item \ind\ind\ind // Valid Lemma
     \item \ind\ind\ind $\lemmas := \lemmas \circ (L)$ (sequence extension)
     \label{alg1:valid-lemma-addition}
     \item \ind\ind\ind $\Phi_\alpha := \left(\bigwedge\axioms \cup \recdef^{\mathit{fp}} \cup \lemmas\right) \Rightarrow \alpha$
     \item \ind\ind\ind $\pfpmodel := \emptyset$
     \label{alg1:pfpmodel-reset}
     \item \ind\ind\ind \textsc{Continue Loop on Line~\ref{alg1:outer-loop}}
     \label{alg1:continue-outer-loop}
     \item \ind\ind \textsc{Else} // Unprovable Lemma
    %\item \ind\ind\ind  // Unprovable Lemma
     \label{alg1:lemma-validity} 
     \item \ind\ind\ind $\gf_2(\overline{c}) =$ {\bfseries \truecounterex}$(L, %\axioms \!\cup\! \recdef, 
     \mathit{size})$ 
     \label{alg1:truemodel-generation}
     \item \ind\ind\ind \textsc{If} ($\gf_2(\overline{c})$ found) // Invalid Lemma
     \label{alg1:invalid-lemma}
%     \item \ind\ind\ind\ind 
     \item \ind\ind\ind\ind $\lfpmodel := \lfpmodel \cup \{(\gf_2(\overline{c}),R)\}$
     \label{alg1:truemodel-addition}
     \item \ind\ind\ind \textsc{Else} // Irrefutable and Unprovable Lemma
     \label{alg1:unprovable-lemma}
     %\item \ind\ind\ind\ind 
     \item \ind\ind\ind\ind $\gf_3(\overline{c}) =$ {\bfseries \provabilitycounterex}$(\Phi_L, k)$
     \label{alg1:get-pfp-countermodel}
     \item \ind\ind\ind\ind $\pfpmodel := \pfpmodel \cup \{(\gf_3(\overline{c}),R)\}$
     \label{alg1:pfpmodel-addition}
     \item \ind\ind\ind \textsc{Continue loop on Line~\ref{alg1:inner-loop}}
     \label{alg1:continue-inner-loop}
 \end{enumerate}
\end{flushleft}
 
 \caption{The \fossil~ algorithm.} \label{fig:alg1}
\end{figure}
\normalsize

%Box1: described here
%Box3: True-CE-G: next section (ranks, etc.)
%Box4: ynth-Engine: another section (two algorithms- sygus, smt)
%Box2: FO-NP-Verifier as before

%\medskip
%\noindent
\subsection{The \fossil\ Algorithm}
\label{sec:algorithm-independent-lemmas}

We now present the main contribution of this paper, the \fossil~ algorithm, which synthesizes lemmas in order to prove a theorem in \folfp.

Figure~\ref{fig:alg1} shows the pseudocode of \fossil~ using the external components \SQI, \provabilitycounterex, \truecounterex, and \synthesismodule~ described in Section~\ref{sec:fossil-components}. The input is a set of axioms $\axioms$, a set of recursively defined predicates $\recdef$, a grammar $\mathcal{G}$ whose language potentially contains the lemmas of interest, and the goal $\alpha$. %which is a formula quantified universally and only over the foreground sort.
The algorithm is parameterized over a depth $k$ for term instantiation and a bound $h$ on the height of the expressions to synthesize from $\mathcal{G}$.

The algorithm has an \emph{outer loop} for proving the goal correct on line~\ref{alg1:outer-loop} and an \emph{inner loop} for discovering valid lemmas on line~\ref{alg1:inner-loop}. At a general point in the execution on line~\ref{alg1:outer-loop}, we try to prove the formula $\Phi_\alpha$
(which says that the valid lemmas found imply the goal) using \SQI~ with terms of depth $k$. If it is valid, we halt and return the sequence of lemmas found.
%a set of terms $T^*$ used for instantiating the quantified variables in the goal, a grammar $\mathcal{G}_h$ to search for lemmas, a sequence $\lemmas$ of valid lemmas in $\mathit{Lang}(\mathcal{G}_h)$ discovered hitherto, and countermodels $\falsemodel$, $\lfpmodel$ and $\pfpmodel$ whose meaning is as explained above.

If $\Phi_\alpha$ is unprovable, we obtain a $\falsemodel$ counterexample with a foreground universe $C$ on line~\ref{alg1:false-model} and enter the inner loop to discover valid lemmas that will help the proof.
%We first try to prove the goal with the current bag of lemmas on line~\ref{alg1:outer-loop} using the Natural Proofs verifier module. If the module returns that the goal is provable (i.e., that the query is \emph{Valid}), then we exit and declare that a proof has been found. Otherwise, we seek a $\falsemodel$ countermodel from the counterexample generation module on line~\ref{alg1:false-model}. This model witnesses that the goal is not currently provable at the current depth of instantiation (using $T^*$).

At a general point in the inner loop execution on line~\ref{alg1:inner-loop}, we have a $\falsemodel$ counterexample, along with a set of $\lfpmodel$ counterexamples and a set of $\pfpmodel$ counterexamples. We call the \synthesismodule~ module to find a lemma in $\mathcal{G}$ of the form $L(\overline{x}) = \forall \overline{x}. R(\overline{x}) \rightarrow \psi(\overline{x})$ and height bounded by $h$ such that: (\ref{alg1:falsemodel-constraints}) the lemma is false on the $\falsemodel$ model, i.e., false on some tuple of elements from $C$ \highlight{(in line~(\ref{alg1:falsemodel-constraints}), $L[C]$ denotes the set of all instantiations of $L$ by elements from $C$, and $\bigwedge L[C]$ their conjunction)}; (\ref{alg1:truemodel-constraints}) the lemma holds on every $\lfpmodel$ counterexample at the tuple $\overline{c}$ witnessing the invalidity of a previously proposed lemma for $R$ (i.e., with $R$ appearing in the antecedent); and (\ref{alg1:pfp-constraints}) the $\mathit{PFP}$ of the lemma holds on every $\pfpmodel$ counterexample at the tuple $\overline{c}$ witnessing the non-inductiveness of a previously proposed lemma for $R$.

%We use the obtained $\falsemodel$ model along with the currently available $\lfpmodel$ and $\pfpmodel$ models and query the synthesis module $\synthesismodule$ on line~\ref{alg1:lemma-proposal} for a lemma $L$ from the grammar $\mathcal{G}_h$ satisfying the following constraints: (a) $L$ is false on the $\falsemodel$ model on the sub-universe $T^*$ (line~\ref{alg1:falsemodel-constraints}) --- this is because the $\falsemodel$ model represents the non-provability of the goal, and a lemma that eliminates this model could help prove the goal; (b) $L$ holds on every model in the $\lfpmodel$ models (line~\ref{alg1:truemodel-constraints}) --- since $\lfpmodel$ models are `true' models of the FO+\textit{lfp} theory, this constraint enforces that the proposed lemma is a valid formula in the given theory; (c) $L$ is inductive on every relevant model from the $\pfpmodel$ models (line~\ref{alg1:pfp-constraints}) --- since $\pfpmodel$ models witness the failure of an induction proof of earlier proposal, this constraint enforces that the proposed lemma be inductively provable. 

If no such lemma is found, we halt and restart the \fossil~ algorithm with higher values for $k$ and $h$. If a lemma $L$ is found, we try to prove $\Phi_L$ valid on line~\ref{alg1:PhiL-validity} using terms of depth $k$, which says that $\mathit{PFP}(L)$ holds (i.e., $L$ is inductive) given the other valid lemmas discovered. If it is valid, then we add $L$ to our assumptions and the current sequence of lemmas, discard $\pfpmodel$ counterexamples, stop the inner loop, and finally retry the proof of the theorem on line~\ref{alg1:outer-loop}. We discard $\pfpmodel$ counterexamples since previously non-provable lemmas may now be provable.

%If the synthesis module returns that no such lemma can be found, we stop and call \fossil~ again with larger parameters $k$ for depth of instantiation for verification and $h$ for the subset of the grammar to be explored for synthesis on line~\ref{alg1:increase-depth}. If a lemma $L$ is proposed, we first construct the first-order formula $\Phi_L$ that represents an induction proof of $L$ on line~\ref{alg1:PhiL}. This involves using the PFP of $L$ as described in Section~\ref{sec:natural-proofs}. Then, we query the $\natproofsalgo$ module on line~\ref{alg1:lemma-validity} with $\Phi_L$. 

If $\Phi_L$ is unprovable, we try to obtain a $\mathit{size}$-bounded $\lfpmodel$ counterexample $\gf_2(\overline{c})$ on line~\ref{alg1:truemodel-generation} such that $L$ does not hold on the tuple $\overline{c}$ of the foreground universe. %elements in $\gf_2$ represented by the constants $\overline{c}$.
If we cannot obtain a $\lfpmodel$ counterexample, then we obtain a $\pfpmodel$ counterexample $\gf_3(\overline{c})$ such that $\mathit{PFP}(L)$ does not hold at $\overline{c}$ in the model $\gf_3$. We add these counterexamples to their respective sets and continue searching for valid lemmas on line~\ref{alg1:inner-loop}.

\subsection{Running Example: List Segments}
\label{sec:illustrative-example}
\highlight{
In this section, we present a full execution of our algorithm on the running example introduced in Example~\ref{ex:running}. 
Let us recall the Verification Condition (VC) $\eggoal$ introduced earlier:
% \begin{equation*}
%   \lseg(x,y_1) \Rightarrow \bigg(\ite(y_1 = \nil, y_2 = y_1, y_2 = n(y_1)) \Rightarrow lseg(x,y_2) \bigg) \tag{$\eggoal$}  
% \end{equation*}

\centerline{$\lseg(x,y_1) \Rightarrow \bigg(\ite(y_1 = \nil, y_2 = y_1, y_2 = n(y_1)) \Rightarrow lseg(x,y_2) \bigg)$}

%\noindent
%where we use $y_1$ and $y_2$ to refer to the value of the variable $y$ before and after one pass of the loop respectively. 
%Notice that this VC is the formula $\eggoal$ that was introduced in Example~\ref{ex:running}, with $\egrecdef$ containing just the definition of $\lseg$: $\lseg(x,y) \lfpdef \ite(x = y, \mathit{true}, \lseg(n(x),y))$

% We first generate the verification condition for the triple as an FO+\textit{lfp} formula:
% \[ \forall x,y,ret.\;\lseg(x, y) \Rightarrow (\lst(y) \Rightarrow (\ite(x = \nil, ret = \nil, ret = n(x)) \Rightarrow \lst(ret))) \]
%}

We illustrate a run of our algorithm that proves $\eggoal$ First, it turns out that $\eggoal$ is not FO-valid and therefore not provable using \SQI. It is also not provable by induction using the formula itself as the induction hypothesis, i.e., $\egrecdef^{\mathit{fp}} \modelsfo PFP(\eggoal)$ does not hold. %Despite the relative simplicity of this formula, not only is it not provable using FO techniques, but it is also not provable using $\SQI$ and not provable by induction with itself as the induction hypothesis. 

\mypara{$\bffalsemodel$ Counterexample} We feed our goal $\eggoal$ to the \natproofsalgo~ module with $k=1$ from which we obtain a $\falsemodel$ counterexample $\Mm_1$ (line~\ref{alg1:false-model} in Figure~\ref{fig:alg1}):
\begin{equation*}
\begin{gathered}
u_1 \mapsto u_2 \mapsto u_3 \mapsto u_3\\
u_4 \mapsto u_3,\ u_5 \mapsto u_5
\end{gathered}
\qquad and \qquad
\begin{gathered}
x = u_1,\ y_1 = u_4,\ y_2 = u_3,\ \nil = u_5\\
\lseg(u_1,u_3) = \mathit{false},\text{ and } \lseg \text{ is } \mathit{true} \text{ otherwise}
\end{gathered}
\end{equation*}
% setting $x = 0, y = 2, \nil = 3$, and $ret = 1$, with $0 \mapsto 1 \mapsto 5 \mapsto 6 \mapsto 1$ and $3 \mapsto 3$. We use 
\noindent
where we use $u \mapsto v$ to represent $n(u) = v$ and $u_i$ are elements of the model returned by the solver (one can think of them as new constants). %This corresponds to the $\falsemodel$ counterexample from line~\ref{alg1:false-model} in Figure~\ref{fig:alg1}. 
We make some observations here about the interpretation of $\lseg$ in $\Mm_1$. The interpretation is not consistent with \textit{lfp} semantics as $\lseg(u_1,u_4) = \mathit{true}$ but $u_1$ never reaches $u_4$ following the $n$ pointer. In fact, the interpretation is not even consistent with the fixpoint semantics $\egrecdef^\mathit{fp}$ as the definition does not hold for $\lseg(u_1,u_3)$. %Combined with the fact that our premise is $\egrecdef^\mathit{fp}$ and not $\egrecdef$, one may speculate that the interpretation is consistent with fixpoint semantics. However, this is also not the case since $u_1$ clearly reaches $u_3$ in two steps but $\lseg(u_1,u_3) = \mathit{false}$. 
This is because \SQI~ at $k=1$ only enforces the fixpoint interpretation for $\lseg$ if the two locations are one step away. Therefore, $\Mm_1$ merely witnesses the non-provability of $\eggoal$ using \SQI~ with $k=1$\footnote{\highlight{The reader may wonder whether using \SQI~ at $k=2$ proves $\eggoal$. However, this is also not true as one can construct a model similar to $\Mm_1$ where $u_3$ is three steps away from $u_1$ instead of two. In fact, there exists such a counterexample for any $k$.}}. %Note that valuation of $\lseg(x, y)$ in the lfp semantics is false since $0$ never reaches $2$, so it merely witnesses the non-provability of the theorem using $\SQI$. 

%but the extent of quantifier instantiation performed   is not true when we fully evaluate the \textit{lfp} definitions. However, our natural proof depth does not allow sufficient unfolding of the definitions to witness the ``true'' meaning of a list segment in this case.

\mypara{$\bflfpmodel$ Counterexample} %With the model generated, we next 
We now search for a lemma using the $\synthesismodule$ module %satisfying the properties from line
(line~\ref{alg1:lemma-proposal}), %in Figure~\ref{fig:alg1}. %The first lemma proposed is
which could propose the lemma $L_1 \equiv \forall x,y.\; \lseg(x,\nil) \Rightarrow \lseg(y, x)$. $L_1$ is not true on $\Mm_1$ and eliminates it as expected, but it is not valid (and is hence found not provable on line~\ref{alg1:PhiL-validity}). We now give it to the \truecounterex~ module (line~\ref{alg1:truemodel-generation}) which returns the $\lfpmodel$ counterexample $\Mm_2$:
\begin{equation*}
\begin{gathered}
v_1 \mapsto v_2 \mapsto v_2
\end{gathered}
\qquad and \qquad
\begin{gathered}
x = v_1,\ y = v_2,\ \nil = v_2\\
\lseg(v_2,v_1) = \mathit{false},\text{ and } \lseg \text{ is } \mathit{true} \text{ otherwise}
\end{gathered}
\end{equation*}
%to find a $\lfpmodel$ model exhibiting its invalidity. 

$\Mm_2$ is a model of a one-element linked list where the interpretation of $\lseg$ is consistent with the \textit{lfp} semantics. %The generated model sets $x = 1$ and $y = \nil = 2$, with $1 \mapsto 2$. The antecedent is true as $x$ represents a single node immediately pointing to $\nil$. However, the consequent is equivalent to $\lseg(\nil, x)$ which is false as $x$ is not $\nil$. 
We add $\Mm_2$ to the set of $\lfpmodel$ models (line~\ref{alg1:truemodel-addition}) ensuring that future lemmas at least hold true on this simple model and continue our search.%to ensure better future proposals the process then continues, and newly proposed lemmas must also be true on this model.

%%% This para is about a provable but useless lemma --- removing for space considerations
%This repeats, with new models being added, when the \synthesismodule{} module proposes the lemma $\forall x.\;\lst(x) \Rightarrow \lseg(x, \nil)$. It turns out that the prefixpoint (PFP) of this lemma is provable using $\SQI$. We then query the \natproofsalgo{} module with this lemma added to check if the original theorem is now provable. It still is not provable, so the process continues, with this lemma added to our set $\lemmas$ of valid lemmas.

\mypara{$\bfpfpmodel$ Counterexample} At some point in the search we obtain the lemma
%the synthesizer proposes
%Invalid lemmas continue to be proposed, with $\lfpmodel$ models being generated, until the valid lemma
$L_2 \equiv \forall x,y.\; \lseg(x, y)\newline \Rightarrow (\lseg(y, \nil) \Leftrightarrow \lseg(x,\nil))$. %is proposed.
$L_2$ is valid but, as it turns out, $PFP(L_2)$ is not FO-valid (under $\egrecdef^{\mathit{fp}}$) and therefore $L_2$ is not provable. The failure of the check on line~\ref{alg1:PhiL-validity} leads to the generation of a $\pfpmodel$ counterexample\footnote{\highlight{Observe here that a $\pfpmodel$ counterexample can always be generated for an unprovable lemma, regardless of whether the lemma is truly invalid or not.}} (line~\ref{alg1:get-pfp-countermodel}) $\Mm_3$ which is similar in spirit to $\Mm_1$ as it witnesses the non-provability of $PFP(L_2)$ by \SQI. We do not present the model here in the interest of brevity. We add $\Mm_3$ to our set of countermodels (line~\ref{alg1:pfpmodel-addition}) to ensure that $L_2$ is not re-proposed (until we get another valid proposal) and continue lemma search.

% The PFP of this lemma is:\\
% %\begin{align*}
% \centerline{$    \forall x,y.\;\ite(x == y, \top, \lseg(y, \nil) \Leftrightarrow \lst(n(x)) \land \lseg(n(x), y)) 
%     \Rightarrow (\lseg(y, \nil) \Leftrightarrow \lst(x))
% $}
% %\end{align*}
% However, unlike the previous lemma proposed, we are unable to prove the PFP of this lemma. So the \provabilitycounterex{} module is called (line~\ref{alg1:get-pfp-countermodel}), and a $\pfpmodel$ model is generated. In this case, we have $x = y = 0$ and $\nil = 2$, with
% $ 0 \mapsto 1 \mapsto 6 \mapsto 3 \mapsto 4 \mapsto 5 \mapsto 1$.
% The antecedent of the PFP is trivially true since $x = y$. However, similar to the first model discussed in this example, the looping behavior shows that $\lseg(y, \nil)$ and $\lst(x)$ are both untrue with the recursive definitions fully evaluated. %But, the reasoning depth is again too small to exhibit this equivalence. 
% The $\pfpmodel$ model is then added on line~\ref{alg1:pfpmodel-addition}, and the process continues.

\mypara{Denouement} After many such rounds of lemma proposal and counterexample generation, the synthesizer %Finally, the \synthesismodule{} module 
proposes the lemma $\eglemma \equiv \forall x,y_1,y_2.\;\lseg(x, y_1) \Rightarrow (\lseg(y_1,y_2) \Rightarrow \lseg(x,y_2))$ %. Observe that $L_3$ is precisely $\eglemma$ 
introduced in our running example (Example~\ref{ex:running-pfp} in Section~\ref{sec:induction-proofs}). We know from Examples~\ref{ex:running-pfp} and~\ref{ex:running-thm-and-lemma} that $\eglemma$ is inductive and proves $\eggoal$, and in fact it is provable with \SQI~ at $k=1$. Therefore, the checks on line~\ref{alg1:PhiL-validity} and subsequently on %as well as the subsequent proof attempt on
line~\ref{alg1:outer-loop} both succeed, whereupon \fossil~ terminates and reports that $\eggoal$ is valid along with the lemma $\eglemma$ used to prove it.
%The prefixpoint of this lemma is proved valid using the \natproofsalgo{} module. This time, the \natproofsalgo{} module is able to prove the original goal with this lemma (and the other valid lemmas at this point), so we break from the loop on line~\ref{alg1:outer-loop}, completing our proof. 
%We output both the valid lemmas above as the valid synthesized lemmas. In this case, only the most recent lemma synthesized was needed to prove the original goal, but that will not always be the case.
}

\section{Synthesis and Counterexample Generation Engines}
\label{sec:component-implementations}
In this section, we provide details of the individual modules from Figure~\ref{fig:orchestra}. We refer the reader to Section~\ref{sec:natural-proofs} for the \natproofsalgo~ module and only describe the synthesis and counterexample generation modules below.

\subsection{Synthesis Engine}
\label{sec:synthesis-module}

The module $\synthesismodule$ takes a finite grammar for expressing lemmas along with a set of \emph{ground} constraints $\psi(\mathit{exp})$ over an expression variable $\mathit{exp}$. A finite grammar is one that generates a finite language. It produces a formula $\varphi$ in the grammar such that $\psi$ is valid when $\mathit{exp}$ is replaced with $\varphi$. 

This problem formulation is similar to SyGuS~\cite{alur15, alur18} in that we have a grammar and constraints on the synthesized expression. However, SyGuS specifications are of the form $\forall \overline{x}.\, \psi(\mathit{exp},\overline{x})$ and can therefore be more complex. In contrast, our constraints have no variables or quantification and are \emph{grounded}. We can of course use SyGuS solvers as synthesis engines, and indeed we do so in a version of our implementation of \fossil~ (see Section~\ref{sec:implementation}).

We now describe our custom synthesis engine tailored for ground constraints. \highlight{First, since our lemmas are all purely universally quantified over the foreground sort we make the quantifiers implicit and only synthesize quantifier-free expressions}. Second, we reduce the synthesis to a quantifier-free query over a combination of theories that can be effectively handled by modern SMT solvers~\cite{nelson80, demoura08}. %This is possible because we have a finite grammar and ground constraints. 
Since derivations from the grammar are of finite height, it is easy to see that we can encode any expression in the language using a finite set of boolean variables representing choices of production rules for each nonterminal in a derivation. Encodings like these are typical in constraint-based synthesis. Combined with the fact that the constraints are grounded, synthesis reduces to a quantifier-free SMT query that asks for an assignment to the boolean variables representing a candidate lemma that satisfies the constraints.

\highlight{
\mypara{Grounded Constraints and Using Boolean Constraint Solvers} One important optimization that we did in the synthesis engine is to solve it using (essentially) Boolean constraints. Counterexamples in our setting
are finite models that can be captured using \emph{grounded formulas} as
described in the previous section. Given a grammar, we first bound the depth of the grammar (this bound is incremented in an outer loop) and we model the choices of which production rules are applied using a set of Boolean variables $\overline{b}$. Consequently, each valuation of $\overline{b}$ stands for a formula $\psi[\overline{b}]$. For conforming to a  counterexample $ce$, we need to write a formula $\textit{Eval}_{ce}(\overline{b})$ that checks whether the formula
$\psi[\overline{b}]$, the formula encoded by $\overline{b}$, holds on the model $ce$ for a particular instantiation of the free variables in $\psi$. (The actual lemma universally quantifies over variables and asserts $\psi$.) 

The straightforward encoding of this problem will essentially evaluate the parse tree of the formula, examining the appropriate Boolean variables in $\overline{b}$ to interpret subformulas or subterms at each node of the parse tree, introducing variables of appropriate sort for subterms. This introduction of variables causes the problem to be an SMT query. However, if we restrict to grammars where all nonterminals generate only formulas (no terms), then it turns out that we can encode the problem without additional variables.

Grammars can be made to have nonterminals generate only formulas by enumerating terms in the derivation rules of atomic formulas. Furthermore, evaluation of atomic formulas over models can be effected using just ground formulae, for a particular instantiation of the free variables over a model, which can be modeled using Skolem constants. 

This leads to formulae over $\overline{b}$ that are all grounded constraints, which is essentially Boolean satisfiability. We implement the above optimization and find it extremely effective on our benchmarks. 
}

We implement the above technique in a custom synthesis engine (see Section~\ref{sec:implementation}) and evaluate its efficacy in Section~\ref{sec:evaluation}.

%\adcomment{Should we more formally give the encoding functions in the appendix? This can cover, for instance how we separate the lhs and rhs into two different functions to synthesize so we can write pfp constraints.}

\subsection{Counterexample Generators}
\label{sec:counterexample-module}
\fossil~ uses three kinds of finite counterexample models to guide lemma synthesis. The $\falsemodel$ model witnesses non-provability of the goal given the current set of synthesized lemmas and makes the synthesis goal-directed. $\lfpmodel$ models witness the invalidity of lemmas proposed and guide synthesis towards producing valid lemmas. Finally, the $\pfpmodel$ models witness non-inductiveness of lemmas proposed and guide synthesis towards producing provable lemmas.

Among these, the $\falsemodel$ and $\pfpmodel$ counterexamples are generated using the $\provabilitycounterex$ module as shown on lines~\ref{alg1:false-model} and~\ref{alg1:get-pfp-countermodel} in Figure~\ref{fig:alg1}. These are obtained as a by-product of using the $\natproofsalgo$ module for verification since it reduces the validity of a quantified formula $\varphi$ to the satisfiability of a \emph{quantifier-free formula} $\psi$ (see Section~\ref{sec:natural-proofs}).
%If $\phi$ is unsatisfiable then the original query is valid, and it is satisfiable if the query is unprovable at the current depth of instantiation. If $\phi$ is satisfiable, there exists a satisfying (FO) model over the finitely many terms in $\phi$ that witnesses the non-provability (see discussion on Herbrand's theorem for many-sorted FO in Section~\ref{sec:prelim}). This model can be obtained by using modern SMT solvers to check the satisfiability of $\phi$.

The generation of $\lfpmodel$ models is more involved. We realize the $\truecounterex$ module which generates them using an SMT solver. Given a bound $\mathit{size}$ on the size of the model, we construct a formula that represents the existence of $\mathit{size}$-many elements $u_1,u_2\ldots ,u_{\mathit{size}}$ such that the valuation of functions (including recursively defined predicates) satisfies the axioms and falsifies the given lemma. \highlight{The key aspect of our construction is the notion of the \emph{rank} of $(R,\overline{u})$ for every $R \in \recsym$ and argument $\overline{u}$ in the domain of $R$. The rank of $(R,\overline{u})$ is an integer in the range $[-1,\infty)$ which we constrain to ensure that the valuation of recursively defined predicates on a $\lfpmodel$ model is consistent with their definitions interpreted using \textit{lfp} semantics.

Let us consider the simple case where we only have one recursively defined predicate $R$ which is unary and has the definition $R(x) :=_\mathit{lfp} \rho(x,R)$. Since there is only one recursively defined predicate, we drop $R$ from the notation for simplicity and simply refer to the rank of $u$ instead of the rank of $(R,u)$. Assume that the definition $\rho(x,R)$ refers to $R$ over a particular set of terms---say $R(t_1(x)), R(t_2(x)), \ldots R(t_m(x))$. The rank of $u$ is an integer variable $\mathit{Rank}_u$ whose value is in the range in the range $[-1,\infty)$. We then enforce the following constraints: (a) $R$ holds on $u$ iff the rank of $u$ is not $-1$, (b) if the base case of the definition holds then the rank is $0$, i.e., iff $\rho(u,\bot)$ holds then the rank of $u$ is $0$, (c) if the rank of $u$ is positive, then the witnessing atomic formulae $R(t_i(u))$ that make $\rho(u,R)$ true are such that each $t_i$ gets a smaller non-negative rank than the rank of $u$,
and (d) if the rank of $u$ is $-1$, then in any set of witnessing atomic formulae $R(t_i(u))$ we pick such that their truth would make $\rho(u,R)$ true, there is at least one $t_i$ whose rank is $-1$.

Intuitively, the rank of $(R,\overline{u})$ mimics the iteration order of the usual iterative least fixpoint computation of R at which the tuple $\overline{u}$ is ``added'' to R.} It is easy to see that if we assign ranks this way, %according to the standard iterative computation of least fixpoints
i.e., assigning the rank of $u$ to be the iteration number at which it is added to $R$ (and $-1$ if it is never added), then the ranks will satisfy the above constraints. Furthermore, if an assignment of ranks satisfying the constraints exists, then we are assured that $R$ evaluates to the true least fixpoint. Finally, since we only want a bounded model the above constraints can be expressed as a quantifier-free SMT query. We use this technique to produce true counterexamples to lemmas.

\highlight{
\mypara{Computing Least-Fixpoints versus Using Under-Approximations} The reader may wonder whether it is possible to use under-approximations of the least-fixpoint instead of computing the precise \textit{lfp} valuations for $\lfpmodel$ counterexamples. After all, if a predicate $R$ holds in an under-approximation, then it certainly holds in the least-fixpoint semantics. However, under-approximations will not work because of the presence of negation in two ways. First, our lemmas and theorems can mention recursively defined functions/predicates in negated form. In this case, computing an under-approximation of the lfp will not be correct. For example, consider a lemma $\forall x. R(x) \Rightarrow S(x)$ for recursively defined predicates $R$ and $S$. Negating this lemma would require a model of $R(x) \wedge \neg S(x)$. An under-approximate computation of $S$ will not work in this case as we may obtain models that do not satisfy this negated formula. Second, negations are also needed in recursive definitions. Our general theoretical treatment allows negations in layers. Such definitions do occur in our experiments. For example, the definition of a binary tree (see Example~\ref{ex:tree-lfpdef}) recursively requires the root not to be present in the heaplets of subtrees rooted at the left and right children of the root. This involves negation of the heaplet function $\mathit{htree}$ which is recursively defined. 

}

\section{Soundness and Relative Completeness}
\label{sec:fosil-ind-completeness}
The soundness of \fossil~ is clear from the problem description and the termination conditions in Figure~\ref{fig:alg1}: the branch on line~\ref{alg1:valid-lemma} is only taken when a lemma is proved valid, and the loop condition on line~\ref{alg1:outer-loop} establishes that if \fossil~ terminates, it does so with a sequence of lemmas that prove $\alpha$. We can now ask whether the algorithm will always find a sequence of lemmas in $\mathcal{G}$ that prove $\alpha$ if one exists. It turns out that  \fossil~ is not complete for the problem of sequential lemma synthesis. However, \fossil~ is complete with respect to \emph{independent lemmas} (see Definition~\ref{def:thm-ind-proof}). That is, if there is a set of independent lemmas that prove $\alpha$, then it is guaranteed that \fossil\ will find a sequential proof of $\alpha$. 

\begin{theorem}[Relative completeness of \fossil~ with respect to independent lemmas]
\label{thm:complete-independent-lemma}
If $\alpha$ is provable from $\axioms$ and $\recdef$ by a finite set of independent inductive lemmas in $\mathcal{G}$ in the sense of Definition~\ref{def:thm-ind-proof}, then there is an instantiation depth $k$ and a grammar height $h$ such that \fossil~
%$(\axioms,\recdef,\mathcal{G},\alpha,k,h)$ 
%(see Figure~\ref{fig:alg1})
terminates and returns a sequence $\lemmas$ of lemmas that proves $\alpha$.
\end{theorem}
%We provide the proof in Appendix~\ref{app:algorithm}.
%\pldicommentout{
\begin{proof}[Proof Gist]
Assume that there exists some set of independent lemmas $\{L_1, L_2, \ldots, L_n\}$ that proves $\alpha$. 
%Let us fix $k$ and $h$ to be such that every $L_i$ as well as the goal (given the lemmas) is provable with a depth $k$ instantiation, and the maximum height of any of the productions in $\mathcal{G}$ that yield a lemma $L_i$ is $h$. We claim that $\fossil$ with parameters $k$ and $h$ will terminate having found a sequence of lemmas that prove $\alpha$.
%Since the algorithm is sound, if it terminates there is clearly a sequence of lemmas that proves $\alpha$.
We establish that %either the algorithm will terminate with a proof of the goal, or
at least one $L_i, 1 \leq i \leq n$ will be eventually (at some finite time) chosen by the synthesis module, i.e., it cannot be that the algorithm restarts {\fossil} with new parameters in line~\ref{alg1:increase-depth} or runs forever without choosing one of the lemmas $L_i$. %If some $L_i$ is chosen by the synthesis module, since we know by our choice of $k$ that $L_i$ is provable with depth $k$ instantiation, it will be added to $\Phi_\alpha$ (see line~\ref{alg1:valid-lemma-addition}) before all the variables are reset, which reduces the problem to discovering at most $n-1$ independent lemmas whereupon we will appeal to the induction on number of lemmas to be discovered.

It is clear from the definition of $\mathcal{G}_h$ that $\mathit{Lang}(\mathcal{G}_h)$ is finite for any $h$. Observe from the description of the algorithm in Section~\ref{sec:algorithm-independent-lemmas} that in each round the candidate proposal $L$ will either: (i) be prevented from being proposed again in the inner loop (line~\ref{alg1:inner-loop}) by the addition of a $\pfpmodel$ model, or (ii) be prevented from being proposed again permanently during the execution of \fossil~ (with parameters $k$ and $h$) because it was proved valid and added to $\Phi_\alpha$ or it was proved invalid using a $\lfpmodel$ model. %This eliminates the possibility that the algorithm keeps on proposing lemmas that are not provable. It either finds a provably valid lemma, or it has no further candidate lemmas to propose, and thus would restart the algorithm with new parameters in  line~\ref{alg1:increase-depth}.
%If it finds a valid lemma, the search space for the next round of lemma synthesis is reduced (because the discovered valid lemma will not be proposed anymore). So this can happen only finitely often.
%This means that the grammar will be exhausted in at most $\mathcal{O}(|\mathit{Lang}(\mathcal{G}_h)|^2)$ rounds. 
Therefore we can eliminate the possibility that the algorithm will run forever without choosing a lemma from $\lemmas$.

This leaves us with the possibility that the algorithm reaches line~\ref{alg1:increase-depth} without finding a new candidate lemma. In particular, this means that none of the $L_i$ satisfies the constraints in line 8. 
%We show that this cannot be the case, i.e., that at least one $L_i, 1 \leq i \leq n$ satisfies the constraints (and is therefore a viable proposal for the synthesis module).
%To negate this, assume that the algorithm has not already terminated, and no lemma $L_i$ has been chosen, and the algorithm reaches line~\ref{alg1:lemma}. We will show that at least one $L_i, 1 \leq i \leq n$ satisfies the constraints (and is therefore a viable proposal for the synthesis module). Combined with the fact that the grammar will be exhausted in finite time, we will have that the algorithm either chooses some $L_i$ eventually (and we appeal to our induction on the size of $\lemmas$) or terminates finding a different proof of $\alpha$.
It is easy to see that each $L_i, 1 \leq i \leq n$ satisfies constraints~\ref{alg1:truemodel-constraints} and~\ref{alg1:pfp-constraints} since the former constraint is satisfied by any lemma valid in the FO+\textit{lfp} theory defined by $\axioms$ and $\recdef$, and the latter is satisfied by any lemma that is provable by induction. %at depth $k$. Both of these conditions are true of every $L_i$. 
This leaves us with constraint~\ref{alg1:falsemodel-constraints}. Assume for the sake of contradiction that no lemma satisfies the constraint, i.e., there is a model $M$ (namely the current $\falsemodel$ model) such that $M \models (\axioms \cup \recdef \cup \{\neg \alpha\} \cup \{L_i\})[T_k]$ for any $L_i, 1 \leq i \leq n$. This yields that $M \models (\axioms \cup \recdef \cup \{\neg \alpha\} \cup \{L_i | 1 \leq i \leq n\})[T_k]$, which contradicts our initial assumption that $\{L_1, \ldots, L_n\}$ collectively prove $\alpha$ at depth $k$, i.e., $(\axioms \cup \recdef \cup \{\neg \alpha\} \cup \{L_i | 1 \leq i \leq n\})[T_k]$ is unsatisfiable. Therefore some $L_i$ satisfies the constraint on line~\ref{alg1:falsemodel-constraints} and will eventually be proposed. Finally, we use induction on the number of lemmas $n$ to reduce the given problem to a smaller one. See Appendix~\ref{app:algorithm} for a detailed proof.
\end{proof}
%} % end of pldicommentout
%The above algorithm assumes that we have a synthesis routine that finds expressions in a grammar, corresponding to lemmas, that satisfies the conditions as given in Line~(\ref{alg1:lemma}). These conditions involve the synthesized formula for the lemma's body to hold true/false in a finite set of finite models. We will design (fair) formula synthesis routines to realize this algorithm in Section~\ref{sec:implementation}.

There are several possibilities for extending \fossil~ to achieve completeness for sequential lemma synthesis. One particular extension is an algorithm called \fossilseq. The key idea is that when a lemma is neither provable nor refutable we add the \emph{induction principle} $PFP(L) \Rightarrow L$ as an assumption (instead of adding $\pfpmodel$ counterexamples). %and carry on with lemma synthesis. 
The induction principle can be added because it is always valid. 
%that removes $\pfpmodel$ counterexamples and instead directly discovers a sequence of $\mathit{PFP}s$ of lemmas rather than the lemmas themselves, ensuring that no dependencies between possible lemmas are broken. 
We discuss the possible extensions, describe $\fossilseq$, and prove its relative completeness for sequential lemma synthesis in Appendix~\ref{sec:algorithm-dependent-lemmas}. We do not pursue these extensions in our work any further as they are significantly more expensive than $\fossil$.

\section{Implementation and Evaluation}
\label{sec:evaluation}

In this section, we describe our implementation and evaluation of \fossil~ (see Section~\ref{sec:algorithm}). % on our benchmarks. 
We also compare with lemma synthesis tools over ADTs and Separation Logic.

\subsection{Implementation}
\label{sec:implementation}
We implement \fossil~ in Python, building the components given in Figure~\ref{fig:orchestra} using Z3Py (an API for the SMT solver Z3~\cite{demoura08}) to handle the various SMT queries for verification and generation of counterexamples. Our implementation covers the various external modules as well as the main \fossil~ algorithm.

The first component is an implementation of the $\natproofsalgo$ module (see Section~\ref{sec:natural-proofs}). As far as we know, this is the first implementation of systematic quantifier instantiation~\cite{loding18, pek14, qiu13} that realizes a complete FO validity engine for quantified formulae using SMT. 
% To our knowledge this is the first implementation of Natural Proofs~\cite{loding18, pek14, qiu13} that gives a general tool operating over a multi-sorted first-order logic with recursive definitions for modeling heaps.
% We added to this implementation a mechanism for generating induction templates/induction principles for formulae that can use Natural Proofs to prove formulae using induction 
The second component is an extension of the SQI engine that provides provability counterexamples (used for $\falsemodel$ and $\pfpmodel$ models). 
%This component utilizes the satisfying models returned by SMT in order to build counterexamples. 
The third component is the bounded counterexample generator which we implement using the technique described in Section~\ref{sec:counterexample-module}. %We implement this component by creating a quantifier-free SMT constraint querying the existence of a bounded model over the foreground universe that violates a formula in \folfp. In order to ensure the correct interpretation of recursive definitions, we exploit the boundedness of the universe. For every recursively defined relation $R \in \R$, we use a rank assigning a unique number to every tuple for which $R$ holds that corresponds to when the tuple is added in the iterative computation of $R$.

%mechanism to provide finite counterexample models, implementing the $\provabilitycounterex$ module. Lastly, we also implement the mechanism described in Section~\ref{sec:counterexample-module} to generate bounded true counterexamples.

The fourth component is an implementation of a custom synthesis engine (a SyGuS solver) that uses constraint solvers (SMT) to synthesize expressions from a grammar given ground constraints. We implement this based on the technique described in Section~\ref{sec:synthesis-module}. As we show in our experiments, reductions to off-the-shelf synthesis engines did not work well. Our synthesis engine exploits the fact that constraints are grounded and carefully generates constraints so that synthesis can be done using SMT solvers. %for synthesis so that they can be solved using SMT solvers. 
%It encodes grammars and interpretations of formulae so that the synthesis query can be formulated using \emph{grounded} formulae that can be solved using SMT. 
These optimizations were crucial to ensuring efficiency of the synthesis engine. The synthesis engine explores the space of terms and the space of formulae independently, prioritizing exploring the space of formulae. It only explores terms of depth $0$ or $1$ as we found this sufficient to solve all our benchmarks. % prioritizing further depth in formulae. %We describe the synthesis and bounded counterexample implementation in more detail in Section~\ref{app:synth-and-cex}.

Finally, we implement the core \fossil~ algorithm (Figure~\ref{fig:alg1}), utilizing the components above. 
%The tool is written primarily in Python, uses Z3 

%This is also to our knowledge the first implementation that solves the problem of (finite) model-based synthesis. We implement this solver on top of Z3, using the technique described in Section~\ref{sec:synthesis-module} that reduces synthesis to the satisfiability of quantifier-free formulae.

\subsection{Research Questions}
Our evaluation aims to answer the following Research Questions (RQs).\smallskip\\
%We evaluate our algorithm using the implementation as the FOSSIL tool. We evaluate our tool on theorems derived from a variety of domains and data structures expressible in FO+RD. In the remainder of the section, we describe our benchmark suite, and we assess the effectiveness of our techniques by considering the following research questions (RQs): \smallskip \\
\noindent {\bf  RQ1}: How effective is FOSSIL in synthesizing inductive lemmas to prove theorems?%\\ %\smallskip \\

\noindent {\bf RQ2}: How effective are countermodels in FOSSIL? %\\

%\smallskip \\
\noindent {\bf RQ3}: How effective is our constraint-based synthesis approach in FOSSIL? %\smallskip

%\noindent {\bf RQ3: How effective is the MBLSCS technique in FOSSIL?} \smallskip
%\noindent {\bf RQ4: Can minor changes to FOSSIL improve the tool's performance? }
% We also consider a few other ablation studies.

%We first describe the benchmark suite that we created for the problem of sequential lemma synthesis before discussing the results on experiments that evaluate the above RQs.

\subsection{Benchmarks}
We curate two classes of benchmarks. The first suite
consists of 50 theorems that were 
distilled from the work on VCDryad~\cite{pek14} repository\footnote{The repository can be found at \url{https://madhu.cs.illinois.edu/vcdryad/examples/}.} which verifies heap manipulating programs. VCDryad converts \textsc{Dryad}, a variant of separation logic, to \folfp. From about 450 VCs (Verification Conditions), we eliminated those that were provable using pure FO reasoning, those that were provable by induction (using the theorem itself as the induction hypothesis), or those that could be proved using frame reasoning~\cite{reynolds02}. The goal was to retain only those VCs that required lemma synthesis. 
From these, we distilled a set of theorems (removing trivial reformulations) and added them to our suite. We also formulate several theorems that capture 
static properties of data structures. 
Six more theorems were obtained by modeling partial correctness of scalar programs with loops. We capture
the computation of the program as a linked list of configurations and use \textit{lfp} to determine \emph{reachable} states, demanding that unsafe states are not reached. Table~\ref{tbl:minisy-experiments} shows the list of theorems that we include in our suite. For example, `bst-leftmost' requires proving that the leftmost node in a binary search tree has the smallest key in the entire tree. We also include theorems about linked lists, sorted linked lists, list segments, dags, binary search trees, maxheaps, etc. The benchmarks obtained from scalar programs are labeled by the prefix `reachability'.

The second suite of benchmarks consists of 673 synthetic theorems that are automatically generated using fixed recipes. The data structure is a dag/tree with a $\mathit{key}$ field and data fields $d_1,\ldots,d_n$, all of type integer. The theorem requires proving that a predicate $P$ holds on the $d_1$ field of the root of the tree. The predicates chosen were inspired by induction exercises for undergraduate students in discrete math courses. The inductive lemma requires stating the properties of several data fields. % $d_1$ as well as certain other fields. 
The data structures also satisfy other properties based on structure as well as the $\mathit{key}$ field (dag, tree, binary search tree, max heap, and trees with parent pointers). The suite was obtained from combinations of predicates, the number of data fields, and properties of data structures.

\mypara{Lemma Grammars} \highlight{Since our lemmas are of the form $L \equiv \forall \overline{x}. R(\overline{x}) \rightarrow \psi(\overline{x})$ we make the universal quantifiers implicit, and the grammars only restrict the quantifier-free formula $\psi$.} 
%First, recall the left hand side for all lemmas we consider are of the form $R(\overline{x})$ for some recursive definition $R$.
%Thus, the left hand side will simply iterate over all recursive definitions present for each benchmark.
For the first suite, we systematically generate grammars based only on the syntax of the recursive definitions and the theorem. 
All variables $\overline{x}$ and all foreground constants from the theorem are added to the grammar. All constants mentioned in the definitions and the theorem are also added.
We allow all terms over these variables and constants. 
%Next, we add implication, the only logical connective we include in the grammar.
%We also include the literal $\mathit{false}$ to allow (with implication) for negation and disjunction.
For atomic formulae, we add all relations (including recursively defined relations) over the foreground sort. 
If integers appear in the theorem, we add equality and inequality for integer terms.
If sets appear in the theorem, we add membership and other set operators.
%Finally, we add disequality of terms of foreground sort to all constant of foreground sort. 
%with $\nil$.
The only Boolean connective allowed is implication. We stratify the grammars by the complexity of formulae (primarily split according to the inclusion or exclusion of set operations) to allow for efficient exploration. %We allow only implication as a Boolean connective. 

%For example, the grammar for the right-hand side of \texttt{tree-reach} is as follows:
%\small
%\begin{verbatim}
%    (Start Bool (false (=> Start Start)
%                (not (= Loc nil))
%                (= I I) (<= I I)
%                (member Loc (htree Loc))
%                (reach_lr Loc Loc) (tree Loc)))
%    (I Int (k))
%    (Loc Int (x y))
%\end{verbatim}
%\normalsize
%However, this is far from exhaustive. 
For the second benchmark suite, we design the grammar automatically. We add the variables and constants of the foreground sort from the theorem. We add $0$ and the integer terms built from $\mathit{key}$ and the other data fields. The atomic formulae included are the data structure relation, equalities and disequalities between foreground sort terms, and the fixed predicate $P$ from the benchmark. Finally, we allow implication and conjunction as Boolean operators.

\begin{table*}[!htb]
\centering
\caption{Experiment results of the FOSSIL tool. The Syn column is the number of lemmas synthesized; the Val column is the number of valid lemmas synthesized; the Time column is the runtime in seconds.}
%$\bot$ represents a 900s timeout.
\label{tbl:minisy-experiments}
\begin{minipage}{.49\linewidth}
%\arraystretch{2.0}
    \centering
    \small
	\begin{tabular}{|lccr|} 
		\hline
		\multicolumn{1}{|l}{Theorem} & \multicolumn{1}{c}{Syn} & \multicolumn{1}{c}{Val} & \multicolumn{1}{c|}{Time (s)} \\
		\hline
				dlist-list			& 1	& 1	& 1 \\
		slist-list			& 2	& 1	& 1 \\
		sdlist-dlist			& 2	& 1	& 2 \\
		sdlist-dlist-slist		& 4	& 2	& 3 \\
		listlen-list			& 1	& 1	& 0 \\
		even-list			& 3	& 1	& 1 \\
		odd-list			& 5	& 2	& 3 \\
		list-even-or-odd		& 11	& 4	& 124 \\
		lseg-list			& 7	& 1	& 5 \\
		lseg-next			& 6	& 1	& 6 \\
		lseg-next-dyn			& 1	& 1	& 1 \\
		lseg-trans			& 5	& 1	& 5 \\
		lseg-trans2			& 7	& 1	& 7 \\
		lseg-ext			& 12	& 1	& 12 \\
		lseg-nil-list			& 6	& 1	& 4 \\
		slseg-nil-slist			& 5	& 1	& 4 \\
		list-hlist-list			& 6	& 1	& 2 \\
		list-hlist-lseg			& 4	& 1	& 2 \\
		list-lseg-keys			& 7	& 1	& 4 \\
		list-lseg-keys2			& 7	& 1	& 4 \\
		rlist-list			& 2	& 1	& 2 \\
		rlist-black-height		& 21	& 7	& 125 \\
		rlist-red-height		& 20	& 7	& 124 \\
		cyclic-next			& 20	& 2	& 126 \\
		tree-dag			& 3	& 1	& 3 \\
		\hline
	\end{tabular}
\end{minipage}
\begin{minipage}{.49\linewidth}
    \centering
	\small
	\begin{tabular}{|lccr|} 
		\hline
		\multicolumn{1}{|l}{Theorem} & \multicolumn{1}{c}{Syn} & \multicolumn{1}{c}{Val} & \multicolumn{1}{c|}{Time (s)} \\
		\hline
		bst-tree			& 2	& 1	& 4 \\
		maxheap-dag			& 2	& 1	& 3 \\
		maxheap-tree			& 2	& 1	& 3 \\
		tree-p-tree			& 2	& 1	& 3 \\
		tree-p-reach			& 14	& 2	& 17 \\
		tree-p-reach-tree		& 12	& 3	& 18 \\
		tree-reach			& 9	& 2	& 25 \\
		tree-reach2			& 4	& 1	& 7 \\
		dag-reach			& 5	& 1	& 20 \\
		dag-reach2			& 6	& 1	& 4 \\
		reach-left-right		& 12	& 3	& 40 \\
		bst-left			& 10	& 1	& 57 \\
		bst-right			& 8	& 1	& 104 \\
		bst-leftmost			& 39	& 10	& 167 \\
		bst-left-right			& 27	& 6	& 104 \\
		bst-maximal			& 5	& 1	& 5 \\
		bst-minimal			& 7	& 1	& 7 \\
		maxheap-htree-key		& 29	& 3	& 155 \\
		maxheap-keys			& 9	& 2	& 140 \\
		reachability			& 4	& 1	& 4 \\
		reachability2			& 2	& 1	& 2 \\
		reachability3			& 3	& 1	& 3 \\
		reachability4			& 2	& 1	& 2 \\
		reachability5			& 4	& 1	& 4 \\
		reachability6			& 4	& 1	& 3 \\
		\hline
	\end{tabular}
\end{minipage}
\end{table*}
\normalsize
\subsection{RQ1: Effectiveness of FOSSIL in Proving Theorems}
%\input{table-vcs}

%synthesis algorithm and implementation can succeed in finding valid and useful lemmas for our theorems. 

We study the effectiveness of our tool in solving both benchmark suites.

\subsubsection*{Benchmark Suite \#1}
Table ~\ref{tbl:minisy-experiments} gives the names of the 50 theorems in Suite \#1, along with the total time taken by our tool to prove each theorem. We find that our tool solves all benchmarks within 5 minutes per benchmark, splitting time between the grammar strata. %on each systematically generated grammar for 4 minutes, stratified by complexity. 
\highlight{Guided by early empirical results, we put in an optimization of the general description of $\fossil$ in our tools by incrementing $h$ but not $k$ when we exhaust the given grammar (line~\ref{alg1:increase-depth} in Figure~\ref{fig:alg1}).} % and up to a total of two grammars per benchmark. The \fossil~ tool refers to our main tool with all three kinds of counterexamples, and with our own synthesis algorithm from ground constraints. 
The table also reports the total number of lemmas synthesized and the number of lemmas among those that were proved valid.

%Evidently, \fossil~ is very effective on these benchmarks.
\fossil~ is effective on these benchmarks. The average time per theorem was \SI{30}{s} (with a maximum of \SI{167}{s}). The total number of lemmas proposed varied from 1 (i.e., the first proposed lemma was sufficient) to $39$, with up to $10$ valid lemmas discovered when solving some benchmarks. %proven valid for a single theorem.
\pldicommentout{We refer the reader to Appendix~\ref{sec:apendix-all-lemmas} for a list of all lemmas proven valid during the process of 
proving each theorem.}
%Several of our benchmarks required the synthesis of sequential lemmas.
Most benchmarks were solved with formula depth $h=3$ and term instantiation depth $k = 1$. For 14 benchmarks, %including all 6 reachability benchmarks, 
the tool reached $h = 4$ and $k = 1$.

%We report on the lemmas discovered by the tool. 
%Some simple lemmas assert that one data structure is more general than another. 
\pldicommentout{
For example, consider again the following short program from~\ref{sec:illustrative-example}
\begin{verbatim}
    if (x == nil) { ret := nil }  else { ret := n(x) }
\end{verbatim}
this time with precondition $\dlist(x)$ (stating that $x$ points to a doubly linked list) and postcondition $\lst(x)$. When converted to a VC, this is not provable in first-order logic. Further, the PFP of this theorem is also not provable. However, \fossil~ quickly synthesize the lemma $\forall x.\; \dlist(x) \Rightarrow \lst(x)$ (in fact it is the first lemma proposed). The PFP of this lemma is provable using the natural proof solver of \fossil, and with this lemma assumed, the VC of the above triple is provably valid in FOL
(this is the \textsf{dlist-list} example in Table~\ref{tbl:minisy-experiments}).
}

The tool finds interesting lemmas such as those characterizing properties of data structures, and relating different structures (like lists and list segments), relating different constraints on data structures, etc. We refer the reader to Appendix~\ref{sec:apendix-all-lemmas} for valid lemmas discovered in proving each theorem.
%including list segment transitivity, relating lists and list segments, loop invariants for partial correctness, etc. %more complex data structures like heaplets with sets as a background theory, and integer arithmetic. 
For example, for \texttt{bst-left-right}, the tool proposes 27 lemmas of which 6 were proved valid, including complex lemmas such as\\
\centerline{$\bst(x) \Rightarrow (y \in \hbst(x) \Rightarrow \minr(x) \leq \minr(y))$}\\
Here, $\bst(x)$ means $x$ is the root of a binary search tree, and $\minr(x)$ denotes the minimum key in the subtree rooted at $x$; both are recursively defined. The lemma states that for every node $y$ in a bst, the minimum key in the subtree of $y$ is less than or equal to the minimum key of the whole tree.
While intuitively true for any bst, formal proof of this property requires induction.

\begin{figure}
    \centering
    %\includegraphics[width=3.9in]{FOSSIL-synthetic.png}
    %\caption{Cumulative sum graph of FOSSIL on the synthetic benchmark suite of 673 theorems.\adcomment{Remove figure and describe in words}}
    \includegraphics[width=3.8in]{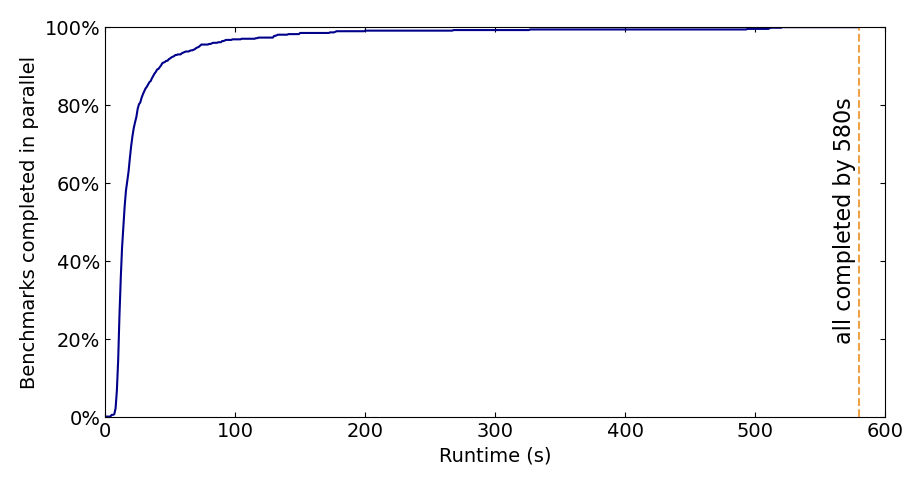}
    \caption{Cumulative sum graph of FOSSIL on the synthetic benchmark suite of 673 theorems.}
    \label{fig:synthetic}
\end{figure}

\subsubsection*{Benchmark Suite \#2}
Figure~\ref{fig:synthetic} contains a cumulative sum graph depicting the time taken by our tool on the synthetic benchmarks.
Our tool performs well, proving all 673 theorems within the timeout of 10 minutes.
$628$ of the benchmarks, approximately $93\%$, were solved within one minute.
%Of these, most benchmarks did not approach the 10 minute mark, as 486 of the 578 benchmarks that terminated did so in under 5 minutes.

%These benchmarks often defined a data structure that updates certain fields, and then claim an invariant on one of the fields. 
%For example, consider the following definition:
%\begin{align*}
% \lst(x) \lfpdef (x = \nil) \vee (n(x) = \nil \land a(x) = 1 \land b(x) = 1)\\
% \quad \vee (a(x) = b(n(x)) + 1 \land b(x) = a(n(x)) + 2 \land \lst(n(x)))
%\end{align*}
%with the theorem: $\forall x.\; \lst(x) \Rightarrow (x \neq \nil \Rightarrow a(x) > 0)$.
%This is neither provable in first-order logic nor is the PFP provable.
%Still, \fossil~ synthesizes and proves the lemma 
%$\forall x.\; \lst(x) \Rightarrow (x \neq \nil \Rightarrow a(x) > 0 \land b(x) > 0)$
%which implies the original theorem.

\subsection{RQ2: Comparison to Synthesis without Use of  Counterexamples}

\begin{figure*}[t!]
\captionsetup[subfigure]{font=small,labelfont=small}
    \centering
    \begin{subfigure}[t]{0.48\textwidth}
        \centering
        \includegraphics[width=2.5in]{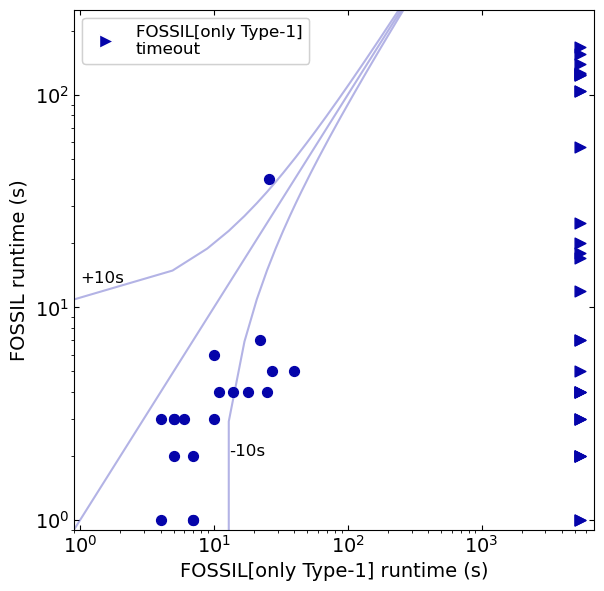}
        \caption{Runtime comparison of FOSSIL vs. FOSSIL with no $\pfpmodel$ or $\lfpmodel$ counterexamples.}
        \label{fig:streaming-comparison}
    \end{subfigure}
    \hfill
    \begin{subfigure}[t]{0.48\textwidth}
        \centering
        \includegraphics[width=2.5in]{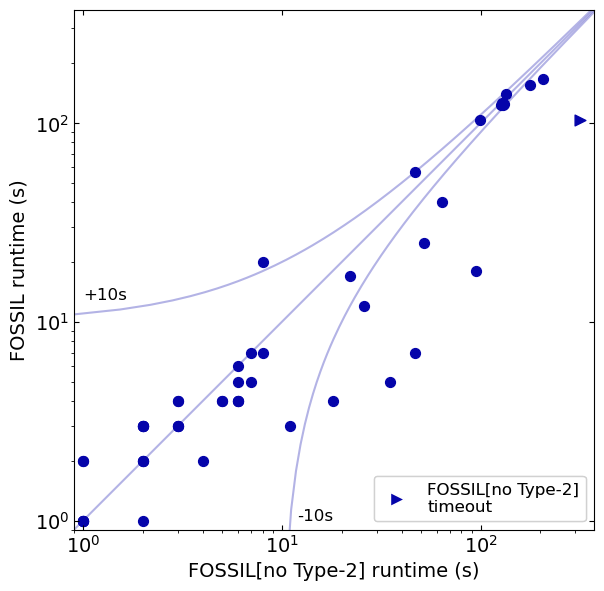}
        \caption{Runtime comparison of FOSSIL vs. FOSSIL with no $\lfpmodel$ counterexamples.}
        \label{fig:no-lfp-comparison}
    \end{subfigure}
    \vskip\baselineskip
    \begin{subfigure}[t]{\textwidth}
        \centering
        \includegraphics[width=2.5in]{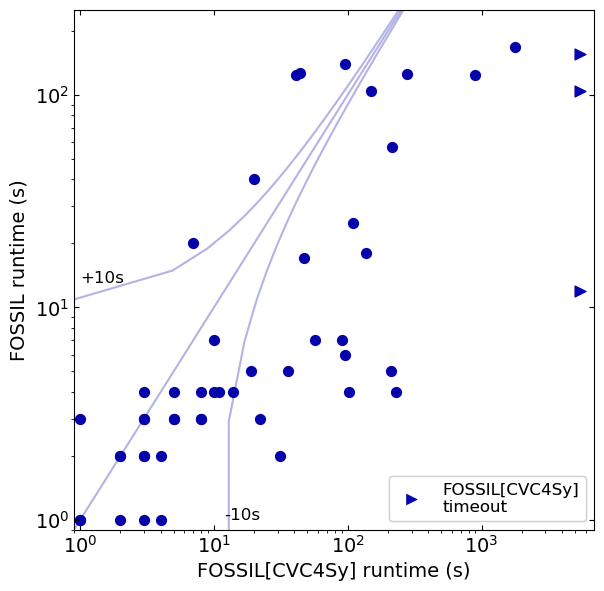}
        \caption{Runtime comparison of FOSSIL vs. FOSSIL using CVC4Sy.}
        \label{fig:cvc4sy-comparison}
    \end{subfigure}
    \caption{Ablation studies of the FOSSIL tool. The timeout is 1 hour. In~\ref{fig:streaming-comparison},~\ref{fig:no-lfp-comparison}, and~\ref{fig:cvc4sy-comparison}, the diagonal lines represent equal running time for both axes. Points on the super-diagonal curves signify FOSSIL is 10 seconds slower than its ablated counterpart, while points on the sub-diagonal curves signify FOSSIL is 10 seconds faster.}
    \label{fig:ablation-main}
\end{figure*}

\begin{figure}
    \centering
    \includegraphics[width=4.25in]{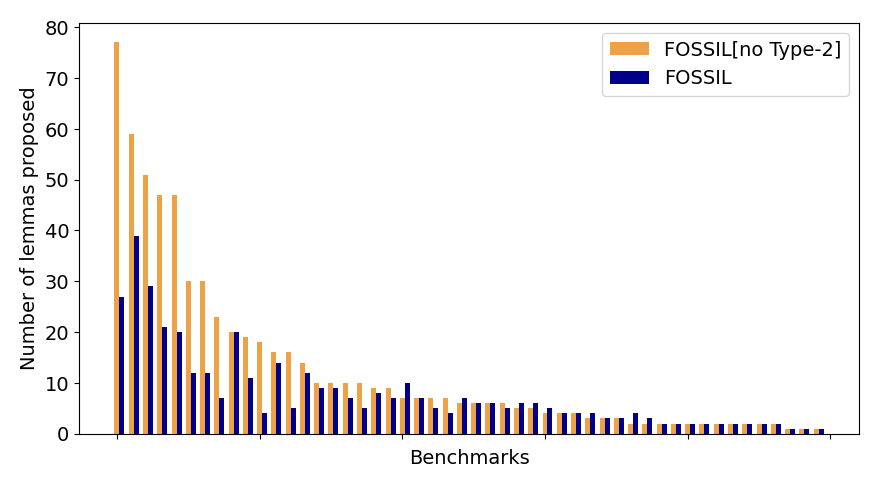}
    \caption{Comparison of lemma proposal counts by FOSSIL vs. FOSSIL without $\lfpmodel$ counterexamples.}
    %\includegraphics[width=3.0in]{bar_FOSSIL_lemmas-proposed_geq5.png}
    %\caption{Comparison of the number of lemmas proposed by FOSSIL vs. FOSSIL without $\lfpmodel$ counterexamples on the 13 benchmarks with difference of proposals at least 5.}
    \label{fig:no-lfp-comparison-lemmas}
\end{figure}

\begin{comment}
\begin{figure*}[t!]
\captionsetup[subfigure]{font=small,labelfont=small}
    \centering
    \begin{subfigure}[t]{0.49\textwidth}
        \centering
        \includegraphics[width=2.675in]{bar_FOSSIL_lemmas-proposed-geq5.png}
    \end{subfigure}
    \hfill
    \begin{subfigure}[t]{0.49\textwidth}
        \centering
        \includegraphics[width=2.675in]{bar_FOSSIL_lemmas-proposed-le5.png}
    \end{subfigure}
    \caption{Comparison of lemma proposal counts by FOSSIL vs. FOSSIL without $\lfpmodel$ counterexamples, with benchmarks grouped according to whether each difference of proposals is greater than or less than 5.}
    \label{fig:no-lfp-comparison-lemmas}
\end{figure*}
\end{comment}

We test the efficacy of counterexamples by removing each kind during synthesis. We do not ablate $\falsemodel$ counterexamples since proposed lemmas would be unrelated to the theorem and a comparison is not meaningful. 
%Observe that $\lfpmodel$ counterexamples are alone not sufficient to build a robust engine since they may not always exist--- the cardinality bound on them may be too small to find a lemma invalid or the lemma may be actually valid but not provable in FO or in FO with current bound on quantifier instantiation. %
%stronger than $\pfpmodel$ (i.e., whenever a $\lfpmodel$ model exists it is guaranteed that a $\pfpmodel$ model exists, but the reverse is not true). 
We perform ablation studies removing 
both $\lfpmodel$ and $\pfpmodel$ counterexamples or only removing  $\lfpmodel$ counterexamples.

\subsubsection*{Efficacy of $\lfpmodel$ and $\pfpmodel$ counterexamples:}
\highlight{It is not possible to directly run our  synthesis engine without $\lfpmodel$ and $\pfpmodel$ counterexamples as the same invalid lemmas can be continuously re-proposed. We hence modify our algorithm to perform the ablation study. The algorithm differs from \fossil~ (Figure~\ref{fig:alg1}) in two ways. First, the $\synthesismodule$ module can \emph{skip} solutions, proceeding to others. Second, when a lemma is not provable (line~\ref{alg1:lemma-validity} in Figure~\ref{fig:alg1}) we simply discard the lemma by asking the synthesis engine to skip to the next solution. We do this until a valid lemma is found, at which point we move to the outer loop (line~\ref{alg1:outer-loop}) and attempt to prove the goal again. Of course, in this algorithm, we also do not maintain sets of $\lfpmodel$ or $\pfpmodel$ counterexamples and only use the $\falsemodel$ counterexample in the synthesis query.
}

\highlight{In our implementation, we integrate a version of \fossil~ with the 
state-of-the-art SyGuS solver in CVC4 (CVC4Sy), providing only $\falsemodel$ counterexamples during synthesis. We used the efficient \emph{streaming} mode of CVC4Sy that can skip solutions. This mode generates a stream of solutions to a synthesis query without repetition, and we simply skip along this stream when we reject candidate lemmas. CVC4Sy is well-optimized, performing symmetry and semantic reductions~\cite{cvc4sy}. We used a timeout of 1 hour for the ablated algorithm.
%The enumerative mode of CVC4Sy is designed to avoid  re-proposing equivalent yet structurally different lemmas. CVC4Sy will continue to propose potential lemmas until either no additional lemmas are possible or a timeout is reached. The potential benefit of this method is that time to create countermodels is avoided, though at the cost of proposing additional lemmas that would otherwise be eliminated. 

Figure~\ref{fig:streaming-comparison} compares the ablated tool against our tool (with all types of counterexamples) on Suite \#1 benchmarks.
Apart from a few outliers where the lemmas proposed are very simple, \fossil~ with only $\falsemodel$ counterexamples performs drastically worse than \fossil~ with all three counterexamples. $31$ of the $50$ benchmarks did not terminate with the ablated tool before the timeout. This shows the efficacy of $\lfpmodel$ and $\pfpmodel$ counterexamples in guiding search.
}

%at search for lemmas using (inequivalent) enumeration is not a feasible strategy. 
%Once the lemma(s) needed become more complex, other techniques are required, as a blind search will not find such lemmas in a reasonable timeframe.
\begin{comment}
We used the streaming mode of CVC4Sy with a timeout of 30 minutes (and including all benchmarks involving lists).
Besides a few benchmarks that completed instantly, all proofs of theorems without counterexample models were far slower than their counterparts with counterexample models. Even more telling, 12 of these 27 tests (including 10 of the 15 list examples) timed out when
not using counterexample models, drastically worse 
than \fossil. This clearly demonstrates that counterexample models are not only useful, but integral to the efficacy of
the \fossil~ tool.
\end{comment}

We also perform this experiment with the synthetic benchmarks (Suite \#2). \fossil~ using only $\falsemodel$ counterexamples surprisingly solves only \emph{1 out of the 673 benchmarks} within 10 minutes. This
again demonstrates the efficacy of  
$\lfpmodel$ and $\pfpmodel$ counterexamples.
%, demonstrating that  synthesis is necessary for these benchmarks as well.

\subsubsection*{Efficacy of $\lfpmodel$ counterexamples} We evaluate the efficacy of $\lfpmodel$ countermodels in \fossil~ by building a version of \fossil~ that does not use $\lfpmodel$ counterexamples. 

\highlight{The ablated algorithm is similar to the one in Figure~\ref{fig:alg1} except for the case where a lemma is not provable (line~\ref{alg1:lemma-validity}). If a lemma cannot be proven valid, we do not try to generate a $\lfpmodel$ counterexample (lines~\ref{alg1:truemodel-generation}-~\ref{alg1:truemodel-addition}) and skip directly to generating a $\pfpmodel$ counterexample (line~\ref{alg1:get-pfp-countermodel}). A $\pfpmodel$ counterexample can always be generated since it witnesses the non-provability of a lemma (see Section~\ref{sec:counterexamples}). It also ensures that such unprovable lemmas will not be re-proposed.}
%As discussed in Section~\ref{sec:algorithm}, the $\lfpmodel$ countermodels are given priority over 
%$\pfpmodel$ in \fossil; when the former are not found, we rely on the latter. We can, however, build a version of the tool with $\pfpmodel$ countermodels only, omitting $\lfpmodel$ countermodels.
%, and subsumed by the $\pfpmodel$ countermodels. That is, provided a proposed lemma is not valid, we will still provide a model witnessing its falsehood, which prevents the lemma from being proposed again. The difference, however, is that the model witnesses the falsehood of $\pfpmodel(\mathit{Lemma})$, rather than of the lemma itself. For that reason, we can't use this model as liberally in ruling out other potential lemmas.

Figure~\ref{fig:no-lfp-comparison} shows the running time comparison between the \fossil~ tool and the \fossil~ tool without $\lfpmodel$ counterexamples.
%For the majority of benchmarks, the runtime is improved when $\lfpmodel$ countermodels are utilized. In fact, even on theorems where turning off $\lfpmodel$ countermodels resulted in an improvement in runtime, this improvement is very minimal, and are theorems that take 10s or less with both modes.
%The majority of benchmarks see improved runtimes when $\lfpmodel$ countermodels are utilized; the runtimes worsen only slightly when they do. 
%worsen usually do so slightly or are otherwise quick ($<\!10s$).
%Specifically, $19$ of the $50$ benchmarks took $>\!10s$ to run with $\lfpmodel$ countermodels turned off. Of these 19, 16 terminate faster with $\lfpmodel$ countermodels turned on. 
The ablated tool does not solve one of the benchmarks and is slower in general for many benchmarks, especially those that require more than 10 seconds to solve. $\lfpmodel$ countermodels seem to have a higher impact in pruning the search space for more complex theorems. Figure~\ref{fig:no-lfp-comparison-lemmas} shows a comparison in the number of proposed lemmas for \fossil~ vs. \fossil~ without $\lfpmodel$ counterexample models. Fewer lemmas are proposed for most benchmarks in the \fossil~ tool, showing the efficacy of the guidance of $\lfpmodel$ counterexamples. 
%The instances where fewer lemmas are proposed with $\lfpmodel$ countermodels disabled had a small difference in the number of lemmas proposed, and mostly for the simpler benchmarks. 
%madhu: remove the sentence below? -- Adithya: Done
%Further, without $\lfpmodel$ countermodels, five theorems lead to over 30 proposed lemmas each, while there are no such examples exceed in the FOSSIL tool, and only three theorems required 15 proposed lemmas.

\subsection{RQ3: Comparison with CVC4 SyGuS Solver}
To evaluate the efficacy of our custom synthesis tool that learns from first-order models with grounded constraint solving, we compare our synthesis tool with CVC4Sy (in standard mode with \emph{all} counterexamples), utilizing the synthesis engines in an identical fashion to the \fossil~ tool.
\highlight{We use a timeout of 1 hour for the ablated algorithm.} 
%to propose lemmas, generating $\lfpmodel$ and $\pfpmodel$ models from the lemmas proposed.
Figure~\ref{fig:cvc4sy-comparison} shows the results of this evaluation and indicates that as theorems become more complex, \fossil~ with our custom constraint-based synthesis solver solidly outperforms \fossil~ with CVC4Sy as the synthesis solver.
Thus, exploiting the form of synthesis in this domain that has ground constraints is useful. 

\subsection{Comparison with ADT/Separation Logic Tools}\label{sec:comparison-adts-sl}
The idea of discovering inductive hypotheses to prove theorems is a problem that has been studied in many logical contexts. We are not aware of any tools that synthesize inductive lemmas for FO+\textit{lfp}, especially ones that can handle foreground and background sorts as in our setting. 
%As far as we know, our technique is the first one that is uses finite first-order models as counterexamples to help synthesizing lemmas. 
% what techniques are known in other logic domains

Comparing tools that work for different logics (FO+\textit{lfp}, algebraic datatypes, separation logic) is inherently hard and poses several challenges: the logics being different, the hardness of translating theorems between them, translation bloat, translations that make theorems harder and required lemmas more complex, tools supporting only restricted fragments, and so on. These make fair comparisons hard. 

In this section however, we attempt to compare our tool with tools for algebraic datatypes (ADTs) and separation logic on our benchmarks, making the best translation effort. Though our tool performs much better than these tools on our benchmarks, this should not be construed as evidence that the other tools are inferior in their native settings. Yet, as the comparison below will show, 
solving theorems in \folfp~ effectively by reducing them to tools for other logics does not seem possible. We also believe that incorporating our ideas into lemma synthesis tools for
these logics natively is an interesting future direction.

%The reader should, however, take these comparisons with a pinch of salt as it is hard to establish the fairness of these comparisons. In each of these comparisons, we describe the challenges we face and motivate the limited comparison we make.
%Our comparative study does show empirical evidence that our benchmarks in FO+\textit{lfp}, when translated as well as we can to other logics, is not amenable to effective reasoning using existing tools that synthesize inductive lemmas (in particular, 
%CVC4+ig~\cite{reynolds15} over ADTs and  SLS~\cite{trung17} for the SL-COMP competition for separation logic).

%One logical domain where inductive lemma synthesis has been studied is logics for algebraic data types (ADTs). 

%In this section we discuss our experimental comparison with existing solvers. Though no solver that we are aware of can perform lemma synthesis over full first-order logic with least fixpoints, we compare with alternatives in two different domains: Algebraic Datatypes and Separation Logic, specifically the Separation Logic Competition (SL-COMP).

%\vspace*{-0.1cm}
\subsubsection*{Comparison with Tools for Algebraic Datatypes}
Theoretically, the logic FO+\textit{lfp} and FO logic over algebraic datatypes are very different. 
%Our work is very different from reasoning with ADTs as supported. 
In pure ADT logics, the universe is a 
%For example, in CVC4~\cite{reynolds15}, since ADTs form a 
\textit{single} universe while FO+\textit{lfp} admits a multitude of universes. 
Furthermore, our benchmarks are motivated by reasoning over pointer-based heaps that embed data structures, which are different from pure mathematical algebraic datatypes (heaps admit a spaghetti of pointers that embed overlapping data structures). 
Consequently, we find it impossible to encode our benchmarks in a pure ADT logic.

%However, when first-order logic over ADTs includes \emph{uninterpreted functions} (or higher-order functions), we can effect an encoding (see Appendix~\ref{app:comparison-adts-sl} for details of this encoding). 

However, when a first-order logic over ADTs includes \emph{uninterpreted functions} (or higher-order functions), we can find reasonable encodings. We can model locations using elements of some ADT (say $0$ with $\textit{succ}$) or even a background theory of integers if supported. We can model pointers using uninterpreted functions from locations to locations. 
Least fixpoint definitions can be modeled in several ways. We choose one that does not involve specific background sorts (such as true natural numbers) and instead uses the structure of ADTs.
%
%The challenge now is in defining \emph{least fixpoint definitions} using pure first-order logic over ADTs.

We encode finite pointer-linked data structures such as linked lists and linked trees using ADTs such as lists and trees, respectively, that store \emph{locations} constituting the linked data structure. Now, recursive definitions on ADTs can capture whether a list/tree of locations corresponds to a linked list/tree by checking, recursively, that the relevant pointers ($\nxt$, or $\lft/\rght$) relate the locations stored in the ADT correctly. Using a mild generalization of this technique, we can encode recursively defined data structures of all kinds used in our benchmarks (including list segments, cyclic lists, doubly linked lists, binary search trees, etc.) in a fairly natural way.

We encoded all 50 benchmarks from Suite~\#1 into CVC4+ig~\cite{reynolds15} and ADTInd~\cite{Weikun19}, both of which use induction and lemma synthesis. %for inductive proofs as well as subgoal generation (akin to our lemma synthesis) in CVC4. 
CVC4+ig solved 1/50 benchmarks, and ADTInd solved 8/50 benchmarks within 15 minutes. This demonstrates that our tool performs significantly better on our benchmarks than reductions to these tools do. 

\subsubsection*{Comparison with Separation Logic Tools}
We consider tools in the Separation Logic Competition (SL-COMP)~\cite{slcomp19} (note that these tools do not have grammars for lemmas). 
%Note that though several tools do support lemma synthesis, they do not take in a specific grammar for lemmas, which is a difference from our tool.
%We investigated the most recent iteration of this competition, SL-COMP'19. 
%Specifically the division closest to our work is the \textbf{qf\_shid\_entl} division, representing entailment problems in quantifier-free separation logic with inductive definitions.

There are many restrictions imposed by the various divisions and tools that make encoding our benchmarks challenging. None of the tools for the closest division \textbf{qf\_shid\_entl} support \emph{conjunction of heap formulas} that we require to encode our benchmarks. Also, some of our benchmarks mention heaplets explicitly and thus are hard to encode.

We consider the solver SLS (Songbird+Lemma Synthesis)~\cite{trung17} that  won the 2019 SL-COMP competition for the \textbf{qf\_shid\_entl} division.
%: SLS~\cite{songbird}, or Songbird with Lemma Synthesis. 
SLS has support for synthesizing inductive lemmas. 
As mentioned above, many of our examples cannot be translated faithfully into SLS. 
We were able to encode and prove valid 14 of the 50 examples from Suite \#1.
There were several examples that we could encode but which SLS was unable to prove (at least 8 such: \texttt{cyclic-next}, \texttt{list-even-or-odd}, and the 6 program reachability examples).

%Some formulas in our benchmarks that involve synthesizing lemmas that relate a more specific data structure to a general one (such as \texttt{dlist-list}, \texttt{slist-list}, etc.) were solved by SLS efficiently. Other examples encoded over simple structures like list and list segments (\texttt{lseg-trans}, \texttt{lseg-nil-list}) were also solved by SLS efficiently.

%We report on encodings of some additional examples that SLS failed to prove. First, \texttt{cyclic-next} in Table~\ref{tbl:minisy-experiments} states that if $x$ is a cyclic list, then so is the value $n(x)$ that $x$ points to. This requires a lemma about list segments, which SLS does not synthesize. Next, \texttt{list-even-or-odd} states that a list will either have an even or odd length. For this, lemmas are needed relating odd and even lists, such as $\oddlist(x) \Rightarrow \evenlist(n(x))$.
%While SLS is able to prove other properties about even and odd lists or list segments, it returns \texttt{unknown} for this example. %(perhaps since one must begin traversing the list in order to come up with the relevant lemma).
%madhu: read till end below
%Finally, all the program reachability examples (\texttt{reachability}-\texttt{reachability6} above) were faithfully encodable in SLS, though SLS returned \texttt{unknown} for all 6 examples.
%Of our 50 theorems, we were able to faithfully encode and prove in SLS only 14 theorems (which were all relatively simple).

\section{Related Work}
\label{sec:related-work}
\mypara{Quantifier Instantiation} Quantifier instantiation is a common tool for reasoning using SMT solvers~\cite{Reynolds16}. E-matching is an instantiation technique used in the \textsf{Simplify}
theorem prover~\cite{detlefs05}, which chooses instantiations based on matching pattern terms. Similar methods are implemented in other SMT solvers~\cite{barrett11,demoura08,rummer12}, as well as methods for combining term instantiation with background SMT solvers~\cite{ge09}. \highlight{The work in~\cite{feldman2017boundedqi} considers bounded quantifier instantiation in a pure FO setting (EPR) without background theories.}

\mypara{Natural Proofs} Our work directly builds off work related to natural proofs~\cite{madhusudan12,qiu13,pek14,suter10,loding18}. VCs similar to the theorems in our experiments are present in~\cite{qiu13}, though lemmas needed to be \textit{user-provided}. Work in~\cite{loding18} provided foundations for the work on natural proofs that preceded it. Completeness results in~\cite{loding18} directly contribute to our completeness results in this paper, and the techniques outlined in~\cite{loding18} are directly implemented in our tool.

\mypara{Reasoning with Recursive Definitions} There is vast literature on reasoning with recursive definitions. The NQTHM prover developed by Boyer and Moore~\cite{boyer88} and its successor ACL2~\cite{kaufmann97,kaufmann00} had support for recursive functions and had several induction heuristics to find inductive proofs. Recent works on cyclic proofs~\cite{brotherston11,ta16} also use heuristics for reasoning about recursive definitions. Additionally, an ongoing area of research involves decidable logics for recursive data structures~\cite{le17}. Naturally, the expressive power of these logics is restricted in order to obtain a decidable validity problem. Further techniques in Dafny~\cite{leino12} and Verifast~\cite{jacobs11} allow for verification via unfolding (or folding) recursive definitions, potentially based on user suggestions.
These are instances of ``unfold-and-match''~\cite{madhusudan12,nguyen08,pek14,qiu13,suter10}, a common heuristic for reasoning with recursive definitions that involves unfolding a recursive definition a few times and finding a proof of validity with the unfolded formulas, treating recursive definitions as uninterpreted. 

\mypara{Lemma Synthesis} The work in~\cite{chu15} uses proof-theoretic techniques to discover subgoals during proofs that serve as inductive hypotheses to help the proof. This relies on chancing upon inductive lemmas during proof, and the paper does not provide any relative completeness results. In contrast, our technique is syntax-guided for arbitrary lemmas and is relatively complete. Other lemma synthesis approaches include that of the work in~\cite{zhang21}, which also uses SyGuS for lemma generation, but operates over the simpler domain of bitvector problems. SLS (Songbird+Lemma Synthesis)~\cite{ta16} is a tool for lemma synthesis over Separation Logic. SLS identifies candidate lemma templates by looking at the heap structure of a given entailment. It then conducts structural induction proofs to generate constraints on top of a lemma template, then solves the constraints to refine the template and discover inductive lemmas. SLS only supports a constrained version of SL, disallowing, for example, conjunctions of heap formulas. As a result, many FO+\textit{lfp} formulas are inexpressible. Refer to Section~\ref{sec:comparison-adts-sl} for a detailed comparison between FOSSIL and SLS on our set of benchmarks.

\highlight{\mypara{Formula Synthesis} The problem of synthesizing or learning first-order formulas from first-order models has seen recent development.
In this space, our synthesis technique is specialized; we only learn universally quantified prenex formulas and hence use a mechanism that synthesizes the quantifier-free matrix using constraints that implicitly assume universal quantification.
The work in~\cite{learning-finite-variables} tackles the harder problem of synthesizing formulas with unboundedly many quantifiers but over finitely many variables (possibly reusing variables) and proves decidability results using tree automata. However, they do not present any practically effective algorithms, and naive implementations of tree automata techniques suffer from state-space explosion.
The work in~\cite{fo-sep} develops a synthesis technique that reduces to SAT, but does not handle grammars; therefore, whether the technique can be used in our work is unclear.
Further, whether either of these two works can be extended to effectively synthesize formulas involving background theories like integers and sets, which we require in $\fossil$, is also unclear.}

%Another counterexample-guided form of learning is \textit{counterexample-guided abstraction refinement} (CEGAR)~\cite{kurshan94,kurshan02}. Unlike counterexample-guided inductive synthesis, CEGAR techniques produce an abstract model, upon which counterexamples are analyzed and the abstract model is accordingly refined.
\adcomment{Is CEGAR relevant here? Look at commented out text here in the document and add if relevant.}

\highlight{\mypara{ICE Learning} Our counterexamples and framework bear resemblance to the work on ICE Learning~\cite{garg2014ICE} for invariant synthesis, as invariants in imperative program verification are similar to inductive hypotheses. $\fossil$ cannot handle invariant synthesis problems, however, despite the fact that inductive invariants are similar to inductive lemmas when programs are written as \folfp formulae. This is because we do not synthesize lemmas that quantify over background sorts such as integers (this distinction also applies to other methods catering to loop invariant synthesis~\cite{neider2018invsynthincompleteengines}, and as such we do not compare with such works). Our synthesis algorithm that uses essentially Boolean constraint solving exploits the fact that expressions synthesized do not have constants over the background sort. Second, while negative counterexamples in ICE correspond roughly to $\falsemodel$ models (though the latter are \emph{non-provability} counterexamples, similar to~\cite{neider2018invsynthincompleteengines}), the positive and implication counterexamples in ICE do not seem to have a strict counterpart in our framework. 
However, $\pfpmodel$ counterexamples come close to implication counterexamples. In the ICE setting programs change configurations, leading to implication counterexamples. In contrast, in the pure \folfp theorem proving setting, there are no changes to models that call for having two separate models as in an implication counterexample.
%and indeed $\pfpmodel$ models cannot be thought of as being composed of two models.
However, $\pfpmodel$ models are a single positive counterexample over which the \emph{PFP} of a lemma must hold. The PFP formula is itself an implication where the lemma to be synthesized ``appears'' on both sides, which makes them similar in spirit to implication counterexamples in ICE. 
}

%\noindent {\bf ADTs and Term Algebras:} 
\mypara{ADTs and Term Algebras} Turning to related work in proving properties of \emph{term algebras} and \emph{algebraic datatypes} (ADTs), the work in~\cite{kovacs17} focuses on automating logics over arbitrary term algebras using FO approximations. 
For lemma synthesis, the work in \cite{Weikun19} is another effort to synthesize inductive lemmas and also uses SyGuS (but without counterexample-guidance). The work in~\cite{reynolds15} also aims to synthesize inductive lemmas, and we provide a detailed comparison with this work in Section~\ref{sec:comparison-adts-sl}. The work in~\cite{racer} infers lemmas for synthesizing invariants but does not use counterexamples.

We emphasize again, however, that work on ADTs/term algebras and our work here on FO+\textit{lfp} are very different and hard to compare both theoretically and experimentally. 
First, a term algebra universe (ADTs) (without background universe) is a \emph{single} universe/model (with fixed interpretation of functions such as constructors/destructors) that is negation-complete.
Our universes model heaps and admit a \emph{multitude} of universes. 
Second, the universe of a term algebra has a complete recursive axiomatization~\cite{malcev62,hodges97}, and hence FO properties of ADTs are in fact decidable, while FO+\textit{lfp} does not even admit complete procedures, let alone decidable ones. 
Third, several data structures we work with do not even have analogous structures in the ADT world--- e.g., list segments between two locations, doubly-linked lists, cyclic lists. 
And destructive pointer updates on them are not expressible in the world of ADTs.
Also defining data structures common in ADTs in the heap
world are considerably more difficult, as we need to express separation (for example, 
even the definition of a tree requires such separation constraints; see Section~\ref{sec:prelim}).
Consequently, a fair experimental comparison of our tools against those developed for ADTs~\cite{sonnex12,cruanes17,hajdu20,boyer88,kaufmann97,passmore20,claessen13,johansson19,racer} is difficult.
Still, when the logic admits uninterpreted or higher-order functions, an encoding is possible (Section~\ref{sec:comparison-adts-sl}).
Extending our techniques to build a lemma synthesis technique/tool with built-in support for ADTs, especially to reason with functional programs, is an interesting future direction.

\section{Conclusions}\label{sec:conclusion}
The primary contribution of this paper is an inductive lemma synthesis technique for \folfp~ with background theories that learns from semantically-rich counterexample FO models that witness non-provability%. of theorems and lemmas. We believe that 
Such a search for inductive lemmas based on the semantics of theorems/lemmas %is interesting and 
can be useful in other contexts--- e.g, in identifying lemmas from a large corpus to help prove theorems. For instance, %such as 
the work in~\cite{bansal19} uses machine learning to find proofs, but currently little semantic information is used in learning. %, where machine learning is used to find proofs, and where there is currently little semantic information used in learning). We believe 
Extending our work to synthesizing lemmas for other logics, especially over ADTs (see Section~\ref{sec:related-work}) as well as Separation Logic is also interesting. 
We also believe that building general lemma synthesis engines that extend SMT solvers %incorporated as extensions to SMT solvers 
can be valuable for researchers who use %routines for researchers who wish to use such 
automated theorem proving in a variety of application domains. 

%In this paper, we present multiple relatively complete algorithms for synthesizing and proving inductive lemmas from a given formula in first-order logic with recursive definitions. We also present a way to realize these algorithms using SMT solvers and \sygus solvers, and a tool that implements this technique with promising preliminary results. Within the tool, we also implement an engine for proving inductive theorems valid using only first-order techniques and a (first-order) induction principle.

%This work represents significant progress in the field of verification of formulas in first-order logic with recursive definitions. An interesting future direction is to define a class of logics around FO-RD, similar to SAT/SMT, and develop automated logic engines for reasoning. We believe that such a standard interface will facilitate ready use of such solvers and techniques, and give a platform for solvers to compete and innovate to find induction proofs automatically.
%Future work includes improving the tooling: decrease running time especially in true model generation, generate SyGuS grammars automatically, and more. Another area for future work involves more detailed invariant synthesis examples, and comparing our tool with other invariant synthesis techniques~\cite{garg14, flanagan01, cobleigh03}.

\section*{Data Availability Statement}
The code and data artifacts required to reproduce the experiments on the \fossil~ tool and various ablation studies are available via ACM DL at~\cite{fossil-artifact}.

\section*{Acknowledgments} This work is supported in part by a research grant from Amazon and a Discovery Partners Institute (DPI) science team seed grant.

%% Bibliography
\bibliography{refs}

\newpage
%% Appendix
\appendix
\section{Appendix}

\subsection{Further details for Section~\ref{sec:algorithm}}
\label{app:algorithm}

\noindent
\textbf{Proof of Theorem~\ref{thm:complete-independent-lemma}.}
\begin{proof}
Assume that there exists some set of independent lemmas $\{L_1, L_2, \ldots, L_n\}$ that proves $\alpha$. Let us fix $k$ and $h$ to be such that every $L_i$ as well as the goal (given the lemmas) is provable with a depth $k$ instantiation, and the maximum height of any of the productions in $\mathcal{G}$ that yield a lemma $L_i$ is $h$. We claim that $\fossil$ with parameters $k$ and $h$ will terminate having found a sequence of lemmas that prove $\alpha$.

We induct on the number $n$ of lemmas in the set. Since the algorithm is sound, if it terminates there is clearly a sequence of lemmas that proves $\alpha$. We establish that either the algorithm will terminate with a proof of the goal, or at least one $L_i, 1 \leq i \leq n$ will be eventually (at some finite time) chosen by the synthesis module, i.e., it cannot be that the algorithm restarts {\fossil} with new parameters in line~\ref{alg1:increase-depth} or runs forever without choosing one of the lemmas $L_i$. If some $L_i$ is chosen by the synthesis module, since we know by our choice of $k$ that $L_i$ is provable with depth $k$ instantiation, it will be added to $\Phi_\alpha$ (see line~\ref{alg1:valid-lemma-addition}) before all the variables are reset, which reduces the problem to discovering at most $n-1$ independent lemmas whereupon we will appeal to the induction on number of lemmas to be discovered.

It is clear from the definition of $\mathcal{G}_h$ that $\mathit{Lang}(\mathcal{G}_h)$ is finite for any $h$. Observe from the description of the algorithm in Section~\ref{sec:algorithm-independent-lemmas} that in each round the candidate proposal $L$ will either: (i) be prevented from being proposed again in the inner loop (line~\ref{alg1:inner-loop}) by the addition of a $\pfpmodel$ model, or (ii) be prevented from being proposed again permanently during the execution of \fossil~ (with parameters $k$ and $h$) because it was proved valid and added to $\Phi_\alpha$ or it was proved invalid using a $\lfpmodel$ model. This eliminates the possibility that the algorithm keeps on proposing lemmas that are not provable. It either finds a provably valid lemma, or it has no further candidate lemmas to propose, and thus would restart the algorithm with new parameters in  line~\ref{alg1:increase-depth}.

If it finds a valid lemma, the search space for the next round of lemma synthesis is reduced (because the discovered valid lemma will not be proposed anymore). So this can happen only finitely often.

%This means that the grammar will be exhausted in at most $\mathcal{O}(|\mathit{Lang}(\mathcal{G}_h)|^2)$ rounds. This eliminates the possibility that the algorithm will run forever without choosing a lemma from $\lemmas$.

This leaves us with the possibility that the algorithm reaches line~\ref{alg1:increase-depth} without finding a new candidate lemma. In particular, this means that none of the $L_i$ satisfies the constraints in line 8. We show that this cannot be the case, i.e., that at least one $L_i, 1 \leq i \leq n$ satisfies the constraints (and is therefore a viable proposal for the synthesis module).

%To negate this, assume that the algorithm has not already terminated, and no lemma $L_i$ has been chosen, and the algorithm reaches line~\ref{alg1:lemma}. We will show that at least one $L_i, 1 \leq i \leq n$ satisfies the constraints (and is therefore a viable proposal for the synthesis module). Combined with the fact that the grammar will be exhausted in finite time, we will have that the algorithm either chooses some $L_i$ eventually (and we appeal to our induction on the size of $\lemmas$) or terminates finding a different proof of $\alpha$.

It is easy to see that each $L_i, 1 \leq i \leq n$ satisfies constraints~\ref{alg1:truemodel-constraints} and~\ref{alg1:pfp-constraints} since the former constraint is satisfied by any lemma valid in the FO-lfp theory defined by $\axioms$ and $\recdef$, and the latter is satisfied by any lemma that is provable by induction at depth $k$. Both of these conditions are true of every $L_i$. This leaves us with constraint~\ref{alg1:falsemodel-constraints}. Assume for the sake of contradiction that no lemma satisfies the constraint, i.e., there is a model $M$ (namely the current $\falsemodel$ model) such that $M \models (\axioms \cup \recdef \cup \{\neg \alpha\} \cup \{L_i\})[T_k]$ for any $L_i, 1 \leq i \leq n$. This yields that $M \models (\axioms \cup \recdef \cup \{\neg \alpha\} \cup \{L_i | 1 \leq i \leq n\})[T_k]$, which contradicts our initial assumption that $\{L_1, \ldots, L_n\}$ collectively prove $\alpha$ at depth $k$, i.e., $(\axioms \cup \recdef \cup \{\neg \alpha\} \cup \{L_i | 1 \leq i \leq n\})[T_k]$ is unsatisfiable. Therefore some $L_i$ satisfies the constraint on line~\ref{alg1:falsemodel-constraints} and will eventually be proposed, which concludes our proof.
\end{proof}

\subsection{Lemma Synthesis Algorithms Relatively Complete wrt Sequential Lemmas}
\label{sec:algorithm-dependent-lemmas}

In this section we briefly discuss the problem of designing algorithms for sequential lemma synthesis that are also relatively complete wrt sequential lemmas (instead of just 
being relatively complete wrt independent lemmas as in Theorem~\ref{thm:complete-independent-lemma}). To do this we must first see why \fossil~ is not already complete for sequential lemmas. They key obstacle is the $\falsemodel$ model that makes the lemma synthesis goal-directed. Consider the following scenario:

\begin{example}[\fossil~ is not complete for sequence of lemmas]
\label{ex:seqlemmas-deadlock}
Consider the case where $\alpha$ can be proved using a sequence $(L_1,L_2)$ of two lemmas. Let  $L_1$ be provable on its own, $L_2$ be provable assuming $L_1$, and $\alpha$ is provable assuming $L_2$. At the beginning of the algorithm on line~\ref{alg1:false-model} in Figure~\ref{fig:alg1}, $L_2$ would be false on $\falsemodel$ since it helps prove $\alpha$. But there is nothing that prevents $L_1[T^*]$ from being true on $\falsemodel$, so let us suppose that it is true. If that is the case, then $L_2$ might be selected by the algorithm and then quickly dismissed since it cannot be proved valid without $L_1$. We would then add a counterexample for it on line~\ref{alg1:pfpmodel-addition} witnessing that $L_2$ has no inductive proof. However, the $\falsemodel$ model has not changed (we only recompute it when we find a valid lemma) and therefore $L_1$ will never be proposed as well. We cannot guarantee that a proof of $\alpha$ will be found by \fossil.
\end{example}

%Assume that $\alpha$ can be shown by a sequence $(L_1,L_2)$ of two lemmas. So $L_1$ is inductive for $\axioms \cup \recdef$, and $L_2$ is inductive for $\axioms \cup \recdef \cup \{L_1\}$. Now $L_1$ might be true on the model $M$ (obtained on line~\ref{alg1:false-model} of Figure~\ref{fig:alg1}) while $L_2$ is false. This means that $L_2$ satisfies constraint (a) but $L_1$ does not. Then $L_2$ might be selected by the algorithm. The inductive proof fails because $L_2$ is only inductive when $L_1$ is already given. So a corresponding model is added to the appropriate set $\modR$ for witnessing that $L_2$ currently has no inductive proof. The model $M$ is not changed (it is only updated when a new valid lemma is discovered). Hence, $L_1$ still does not satisfy constraint (a) and is thus not a candidate. So we cannot guarantee that $L_1$ is found by the algorithm.

\noindent We propose three different strategies to address the above issue:
\begin{enumerate}
    \item The simplest way to achieve the relative completeness is to utilize \fossil~ as described in Figure~\ref{fig:alg1}, but eliminate constraints corresponding to $\falsemodel$ models. This eliminates the problem described in Example~\ref{ex:seqlemmas-deadlock} where we need to necessarily synthesize lemmas that help prove the goal, and instead reduces the algorithm to only generating lemma proposals and eliminating spurious proposals using $\lfpmodel$ and $\pfpmodel$ models. This approach has the obvious disadvantage of not being goal-directed and could lead to large execution time for proof if the sequence of lemmas needed consists of large lemmas (by size) and smaller lemmas could be easily eliminated if given the goal.
    
    \item A second approach is to have the algorithm branch into two-subroutines (both branches searched fairly, dovetailing between them) when given a lemma that is unprovable, one assuming that the lemma is valid and other assuming that it is not. We can then pursue each subroutine until we find a proof or reach a contradiction. However, this algorithm could quickly explode in the number of subroutines even with a few unprovable lemmas and likely impractical.
    
    \item We propose a third alternative that  generalises \fossil. Looking at the Example~\ref{ex:seqlemmas-deadlock}, it would be useful if we could update $\falsemodel$ to include the failure to prove $L_2$ so that the lemma synthesis is guided towards $L_1$. What should be the constraint with which we update the model? The updated model could be such that $L_2$ holds (on the instantiated terms), or it could witness that $L_2$ is not inductive, i.e., cannot be proved by induction. However, these two possibilities are precisely those expressed by the induction principle for $L_2$. Recall the definition from Section~\ref{sec:induction-proofs}: the induction principle of a lemma $L(\bar{x})$ is given by $(\forall \bar{x}. PFP(L(\bar{x}))) \rightarrow (\forall \bar{x}. L(\bar{x})) \equiv \neg(\forall \bar{x}. PFP(L(\bar{x}))) \lor (\forall \bar{x}. L(\bar{x}))$ where $PFP$ represents the condition that $L$ is inductive. Therefore the induction principle captures the two possibilities of $L$ either being valid or not inductive. Our third alternative proposal is thus to use the induction principle to address the problem of completeness for sequences of lemmas in an algorithm we call \fossilseq.
\end{enumerate}

\paragraph{\fossilseq} Let us discuss the third strategy in more detail. Simply put, we would like to add the induction principle for any lemmas that we cannot prove to our axioms and retain the rest of the algorithm. In particular, with respect to the algorithm description in Figure~\ref{fig:alg1} we would maintain a set $\ip$ of induction principles starting out with an empty set and include it in the construction of $\Phi_\alpha$ and $\Phi_L$ on lines~\ref{alg1:PhiAlpha} and~\ref{alg1:PhiL}. Then, given a proposal $L$ that we can neither prove nor establish as being invalid using a $\lfpmodel$ model (line~\ref{alg1:unprovable-lemma}), we would eliminate falling back to a $\pfpmodel$ model on lines~\ref{alg1:get-pfp-countermodel} and~\ref{alg1:pfpmodel-addition} and replace it with the update of $\ip$ with the induction principle of $L$. This algorithm, which we call \fossilseq, is relatively complete for the problem of sequential lemma synthesis:
\!\!\!\!\!\!
\begin{theorem}[Relative completeness of \fossilseq~ wrt sequential lemmas]
\label{thm:complete-induction-principle}
If $\alpha$ is provable from $\axioms$ and $\recdef$ by a finite sequence of inductive lemmas, then there is an instantiation depth $k$ and grammar height $h$ such that \fossilseq~ 
%(see Figure~\ref{fig:alg2} in Appendix~\ref{sec:fossilseq}) 
terminates and returns a set $\lemmas$ of lemmas and a set $\ip$ of induction principles proving $\alpha$.
\end{theorem}
We detail the formulation of proving a theorem using induction principles
%, the pseudocode for the \fossilseq~ algorithm and the proof of its relative completeness
in Appendix~\ref{sec:fossilseq}. 
Admittedly, 
%despite the elegance of being able to handle unprovable lemmas uniformly and being a natural extension to \fossil~ using induction principles, 
our solution could still create large synthesis queries that could be difficult to handle and therefore has potential disadvantages as do the other two strategies.

\subsubsection{Discussion about induction principles and description of \fossilseq}
\label{sec:fossilseq}
As illustrated in Example~\ref{ex:seqlemmas-deadlock}, we cannot guarantee that {\fossil} finds a sequence of inductive lemmas proving the goal if such a sequence exists. We need the stronger assumption that a set of independent lemmas exists for proving the goal.

The algorithm {\fossilseq} is a modification of {\fossil} that is guaranteed to find a proof of the goal if it can be proven by a sequence of inductive lemmas. In addition to the sequence of valid lemmas that is constructed in a similar way as {\fossil} does, {\fossilseq} additionally
uses induction principles of lemmas for which it does not find an inductive proof. It might happen that these induction principles help proving $\alpha$ without the algorithm being able to prove the actual lemmas valid. 
We illustrate the  difference between induction principles and lemmas proving $\alpha$ for an (artificial) example situation.
\begin{example}
Consider the definition of $\textit{list}$ from above. Add two constants $c_1,c_2$ to the signature, and two recursive definitions $\textit{list}_1$ and $\textit{list}_2$:
\[
\begin{array}{rcl}
\textit{list}_1(x) & \lfpdef&
(x = nil) \vee ((\textit{list}_1(n(x)) \land (c_1 = c_2 \rightarrow x \not= c_1)) \\
\textit{list}_2(x) & \lfpdef&
(x = nil) \vee ((\textit{list}_2(n(x)) \land (c_1 \not= c_2 \rightarrow x \not= c_1))
\end{array}
\]
So both are defined as $\textit{list}$ with the only difference that the recursion stops at $c_1$ for $\textit{list}_1$ if $c_1 = c_2$, and for $\textit{list}_2$ if $c_1 \not= c_2$.

Take $\alpha = \forall x. \textit{list}(x) \rightarrow (\textit{list}_1(x) \lor \textit{list}_2(x))$. This is certainly true in LFP semantics because if $c_1 = c_2$, then $\textit{list}_2$ is the same as $\textit{list}$, otherwise $\textit{list}_1$ is the same as $\textit{list}$.
Consider the lemmas $L_1 = \forall x. \textit{list}(x) \rightarrow \textit{list}_1(x)$ and $L_2 = 
\forall x. \textit{list}(x) \rightarrow \textit{list}_2(x)$.
For each lemma, there are clearly LFP models in which the lemma does not hold (if $c_1 = c_2$ and $\textit{list}(c_1)$, then $L_1$ is false, similarly for $L_2$).
However, we have that $\axioms \cup \recdef \cup \{IP(L_1),IP(L_2)\} \modelsfo \alpha$ because on each model either $PFP(L_1)$ or $PFP(L_2)$ is satisfied.
\end{example}

This illustrates that provability by induction principles does not yield provability by the corresponding lemmas. The other direction, however, is always true, as stated in the following lemma.
\begin{lemma}
\label{lem:lemma-sequence-ip}
If $(L_1, \ldots, L_n)$ is a sequence of inductive lemmas that prove $\alpha$ then the set $\ip = \{IP(L_1), \ldots, IP(L_n)\}$ proves $\alpha$.
\end{lemma}
%A more nuanced discussion of induction principles proving a theorem can be found in Appendix~\ref{app:algorithm}.
%If $(L_1, \ldots, L_n)$ is an inductive sequence of lemmas that prove $\alpha$, then also the set $\ip = \{IP(L_1), \ldots, IP(L_n)\}$ proves $\alpha$.
\begin{proof}
If $\ip$ does not prove $\alpha$, then there is a model $M$ of $\axioms \cup \recdef \cup \ip \cup \{\neg \alpha\}$ (in the FO semantics). Since the lemmas from the sequence $(L_1, \ldots, L_n)$ prove $\alpha$, one of the lemmas $L_i$ has to be false on $M$. Since $IP(L_i)$ is true on $M$, we obtain that $PFP(L_i)$ is false on $M$. If $i$ is the smallest index such that $L_i$ is false on $M$, then we get a contradiction to the fact that $(L_1, \ldots, L_n)$ is an inductive sequence of lemmas, and hence $\axioms \cup \recdef \cup \{L_1, \ldots,L_{i-1}\} \modelsfo PFP(L_i)$.
\end{proof}

\normalsize

\small
\begin{table*}
	\centering
	\caption{Sample valid lemmas synthesized and proven correct by our tool.}
	\label{tbl:valid-lemmas-1}
	\small
	\begin{tabular}{| p{2.6cm} | p{10.0cm} |} 
		\hline
		\multicolumn{1}{|c|}{Theorem} & \multicolumn{1}{c|}{Valid Lemmas} \\ 
		\hhline{|==|}
		\hline dlist-list & $\dlist(x) \Rightarrow \lst(x)$ \\
\hline slist-list & $\slist(x) \Rightarrow \lst(x)$ \\
\hline sdlist-dlist & $\sdlist(x) \Rightarrow \dlist(x)$ \\
\hline sdlist-dlist-slist & $\sdlist(x) \Rightarrow \dlist(x)$ \\
& $\sdlist(x) \Rightarrow \slist(x)$ \\
\hline listlen-list & $\lst(v, l) \Rightarrow \lst(x)$ \\
\hline even-list & $\mathit{even\mathrm{-}lst}(x) \Rightarrow \lst(x)$ \\
\hline odd-list & $\mathit{odd\mathrm{-}lst}(x) \Rightarrow \lst(x)$ \\
\hline list-even-or-odd & $\mathit{even\mathrm{-}lst}(x) \Rightarrow (n(x) \neq \nil \Rightarrow \mathit{odd\mathrm{-}lst}(n(x)))$ \\
& $\mathit{odd\mathrm{-}lst}(x) \Rightarrow \mathit{even\mathrm{-}lst}(n(x))$ \\
& $\lst(x) \Rightarrow ((\mathit{even\mathrm{-}lst}(\nxt(x)) \Rightarrow \mathit{false}) \Rightarrow  \mathit{even\mathrm{-}lst}(x))$ \\
\hline lseg-list & $\lseg(x, y) \Rightarrow (\lst(x) \Rightarrow \lst(y))$ \\
\hline lseg-next & $\lseg(x, y) \Rightarrow (\lseg(y, z) \Rightarrow \lseg(x, z))$ \\
\hline lseg-next-dyn & $\mathit{lsegy}(x) \Rightarrow \mathit{lsegz}_p(x)$ \\
\hline lseg-trans & $\lseg(x, y) \Rightarrow (\lseg(y, z) \Rightarrow \lseg(x, z))$ \\
\hline lseg-trans2 & $\lseg(x, y) \Rightarrow (\lseg(y, z) \Rightarrow lseg(x, z))$ \\
\hline lseg-ext & $\lseg(x, y) \Rightarrow (\lseg(y, z) \vee (\lseg(x, z) \Rightarrow \lseg(z, y)))$ \\
\hline lseg-nil-list & $\lseg(x, y) \Rightarrow (\lst(y) \Rightarrow \lst(x))$ \\
\hline slseg-nil-slist & $\slseg(x, y) \Rightarrow (\slst(y) \Rightarrow \slst(x))$ \\
\hline list-hlist-list & $\lst(x) \Rightarrow (y \in \hlist(x) \Rightarrow \lst(y))$ \\
\hline list-hlist-lseg & $\lst(x) \Rightarrow (y \in \hlist(x) \Rightarrow \lseg(x, y))$ \\
\hline list-lseg-keys & $\lseg(x, y) \Rightarrow (k \in \keys(y) \Rightarrow k \in \keys(x))$ \\
\hline list-lseg-keys2 & $\lseg(x, y) \Rightarrow (\lst(x) \Rightarrow \lst(y))$ \\
& $\lseg(x, y) \Rightarrow (k \in \keys(y) \Rightarrow k \in \keys(x))$ \\
\hline rlist-list & $\rlist(x) \Rightarrow \lst(x)$ \\
\hline rlist-black-height & $\rlist(x) \Rightarrow \mathit{red\mathrm{-}height}(n(x)) \leq \mathit{black\mathrm{-}height}(x) + 1$ \\
& $\rlist(x) \Rightarrow 1 \leq \mathit{red\mathrm{-}height}(n(x)) + 1$ \\
& $\rlist(x) \Rightarrow \mathit{red\mathrm{-}height}(x) = 1 + \mathit{black\mathrm{-}height}(n(x))$ \\
& $\rlist(x) \Rightarrow \mathit{black\mathrm{-}height}(n(x)) + \mathit{black\mathrm{-}height}(x) \leq \mathit{red\mathrm{-}height}(n(x)) + \mathit{black\mathrm{-}height}(n(x))$ \\
\hline rlist-red-height & $\rlist(x) \Rightarrow \mathit{red\mathrm{-}height}(x) = 1 + \mathit{black\mathrm{-}height}(n(x))$ \\
& $\rlist(x) \Rightarrow (\mathit{black}(x) \Rightarrow \mathit{red}(\nxt(x))$ \\
& $\rlist(x) \Rightarrow 1 \leq \mathit{red\mathrm{-}height}(x)$ \\
\hline cyclic-next & $\lseg(x, y) \Rightarrow \lseg(n(x), n(y))$ \\
\hline tree-dag & $\tree(x) \Rightarrow \dagraph(x)$ \\
\hline bst-tree & $\bst(x) \Rightarrow tree(x)$ \\
\hline maxheap-dag & $\maxheap(x) \Rightarrow \dagraph(x)$ \\
\hline maxheap-tree & $\maxheap(x) \Rightarrow \tree(x)$ \\
\hline tree-p-tree & $\tree_p(x) \Rightarrow \tree(x)$ \\
\hline
	\end{tabular}
\end{table*}
\begin{table*}
	\centering
	\caption{Sample valid lemmas synthesized and proven correct by our tool.}
	\label{tbl:valid-lemmas-2}
	\begin{tabular}{| p{2.6cm} | p{10.0cm} |} 
		\hline
		\multicolumn{1}{|c|}{Theorem} & \multicolumn{1}{c|}{Valid Lemmas} \\ 
		\hhline{|==|}
\hline tree-p-reach & $\reach(x, y) \Rightarrow (\tree_p(x) \Rightarrow \tree_p(y))$ \\
\hline tree-p-reach-tree & $\tree_p(x) \Rightarrow \tree(x)$ \\
& $\reach(x, y) \Rightarrow (\tree_p(x) \Rightarrow \tree(y))$ \\
& $\reach(x, y) \Rightarrow (y \neq \nil \Rightarrow y \in \htree(x))$ \\
\hline tree-reach & $\reach(x, y) \Rightarrow (\tree(x) \Rightarrow \tree(y))$ \\
\hline tree-reach2 & $\reach(x, y) \Rightarrow (\tree(x) \Rightarrow \tree(y))$ \\
\hline dag-reach & $\reach(x, y) \Rightarrow (\dagraph(x) \Rightarrow \dagraph(y))$ \\
\hline dag-reach2 & $\reach(x, y) \Rightarrow (\dagraph(x) \Rightarrow \dagraph(y))$ \\
\hline reach-left-right & $\reach(x, y) \Rightarrow (y \in \htree(y) \Rightarrow x \in \htree(x))$ \\
& $\reach(x, y) \Rightarrow (y \in \htree(y) \Rightarrow y \in \htree(x))$ \\
& $\reach(x, y) \Rightarrow (\tree(x) \Rightarrow \tree(y))$ \\
& $\tree(x) \Rightarrow (y \in \htree(x) \Rightarrow \tree(y))$ \\
\hline bst-left & $\bst(x) \Rightarrow (k \in \keys(x) \Rightarrow \minr(x) \leq k)$ \\
\hline bst-right & $\bst(x) \Rightarrow (k \in \keys(x) \Rightarrow k \leq \maxr(x))$ \\
\hline bst-leftmost & $\bst(x) \Rightarrow \minr(\leftmost(x)) = \minr(x)$ \\
& $\bst(x) \Rightarrow \maxr(\leftmost(x)) \leq \maxr(x)$ \\
& $\bst(x) \Rightarrow \bst(\leftmost(x))$ \\
& $\bst(x) \Rightarrow \key(\leftmost(x)) \leq \key(x)$ \\
& $\bst(x) \Rightarrow (x \neq \nil \Rightarrow \leftmost(x) \neq \nil)$ \\
& $\bst(x) \Rightarrow ((\bst(\leftmost(x)) \Rightarrow x \in \hbst(x)) \Rightarrow \key(\leftmost(x)) \leq \minr(\leftmost(x)))$ \\

\hline bst-left-right & $\bst(x) \Rightarrow (y \in \hbst(x) \Rightarrow \minr(x) \leq \minr(y))$ \\
& $\bst(x) \Rightarrow (y \in \hbst(x) \Rightarrow \maxr(y) \leq \maxr(x))$ \\
& $\bst(x) \Rightarrow (y \in \hbst(x) \Rightarrow \bst(y))$ \\
& $\bst(x) \Rightarrow (y \in \hbst(x) \Rightarrow y \neq \nil)$ \\
& $\bst(x) \Rightarrow (\minr(x) \leq \maxr(y) \Rightarrow y \in \hbst(y))$ \\
& $\bst(x) \Rightarrow (\minr(y) \leq \maxr(x) \Rightarrow (\bst(y) \Rightarrow y \in \hbst(y)))$ \\
\hline bst-maximal & $\bst(x) \Rightarrow (y \in \hbst(x) \Rightarrow \bst(y))$ \\
\hline bst-minimal & $\bst(x) \Rightarrow (y \in \hbst(x) \Rightarrow \bst(y))$ \\
\hline maxheap-htree-key & $\maxheap(x) \Rightarrow (y \in \htree(x) \Rightarrow \key(y) \leq \key(x))$ \\
& $\maxheap(x) \Rightarrow (y \in \htree(x) \Rightarrow \maxheap(y))$ \\
\hline maxheap-keys & $\maxheap(x) \Rightarrow (k \in \keys(x) \Rightarrow k \leq \key(x))$ \\
\hline reachability & $\reach(z) \Rightarrow (c = y(z) \vee n(x(z)) = n(y(z)))$ \\
\hline reachability2 & $\reach(z) \Rightarrow y(z) = x(z)$ \\
\hline reachability3 & $\reach(z) \Rightarrow x(z) = y(z)$ \\
\hline reachability4 & $\reach(z) \Rightarrow y(z) = x(z)$ \\
\hline reachability5 & $\reach(z) \Rightarrow (n(y(z)) = x(z) \vee y(z) = c)$ \\
\hline reachability6 & $\reach(z) \Rightarrow n(y(z)) = x(z)$ \\
\hline
	\end{tabular}
\end{table*}

\normalsize
\subsection{Lemmas Proved}
\label{sec:apendix-all-lemmas}

Tables~\ref{tbl:valid-lemmas-1} and~\ref{tbl:valid-lemmas-2} represent all the lemmas proved valid by our tool. All variables ($x, y, z, k$, etc.) are implicitly universally quantified. Notably, different runs of our tool may produce different valid lemmas. Additionally, not all lemmas are guaranteed to be useful in proving the given theorem.

\end{document}